\documentclass{article}
\usepackage[utf8]{inputenc}

\usepackage{authblk}
\usepackage[T1]{fontenc}
\usepackage{subeqnarray}
\usepackage{enumitem}
\usepackage{fullpage}
\usepackage{todonotes}
\usepackage{ bbold }
\usepackage{multirow}
\usepackage[title]{appendix}
\usepackage[cmex10,intlimits]{amsmath}
\usepackage{amssymb}
\usepackage{array}
\usepackage{mdwmath}
\usepackage{mdwtab}
\usepackage{siunitx}
\usepackage{cite}
\usepackage{graphicx}
\usepackage{upgreek}
\usepackage{hyperref}
\usepackage{amsthm}
\usepackage{amsmath}
\usepackage{color}
\usepackage{todonotes}
\usepackage{physics}
\usepackage{siunitx}
\usepackage[symbol]{footmisc}
\usepackage{mathtools}
\usepackage{float}
\usepackage{booktabs}

\graphicspath{{figures/}}
\renewcommand{\bra}[1]{{\langle{#1}\vert}}
\renewcommand{\ket}[1]{{\vert{#1}\rangle}}

\newcommand{\bracket}[2]{\langle #1 \vert #2 \rangle}
\renewcommand{\ketbra}[1]{\vert #1 \rangle \ \!\!\! \langle #1 \vert}
\newcommand{\com}{\mathtt{com}}
\newcommand{\open}{\mathtt{open}}
\newcommand{\ver}{\mathtt{ver}}
\newcommand{\syn}{\mathtt{syn}}
\newcommand{\dec}{\mathtt{dec}}

\newcommand{\bs}{\lbrace 0,1 \rbrace}

\newcommand{\ccom}{\textnormal{com}}
\newcommand{\CCOM}{\textnormal{COM}}
\newcommand{\oopen}{\textnormal{open}}
\newcommand{\OOPEN}{\textnormal{OPEN}}

\newtheorem{theorem}{Theorem}[section]
\newtheorem{lemma}{Lemma}[section]
\newtheorem{proposition}[theorem]{Proposition}
\newtheorem{definition}{Definition}[section]
\title{Performance of Practical Quantum Oblivious Key Distribution}
\author[1,2]{Mariano Lemus}
\author[3]{Peter Schiansky}
\author[5]{Manuel Goul\~ao}
\author[3]{Mathieu Bozzio}
\author[4,5]{David Elkouss}
\author[1,2]{Nikola Paunkovi\'c}
\author[1,2]{Paulo Mateus}
\author[3]{Philip Walther}

\affil[1]{Instituto de Telecomunica\c{c}\~oes, 1049-001 Lisbon, Portugal}
\affil[2]{Departamento de Matem\'{a}tica, Instituto Superior T\'{e}cnico, Universidade de Lisboa, Av. Rovisco Pais 1, 1049-001 Lisboa, Portugal}
\affil[3]{Vienna Center for Quantum Science and Technology (VCQ), Faculty of Physics, University of Vienna, Boltzmanngasse 5, Vienna A-1090, Austria}
\affil[4]{QuTech, Delft University of Technology, Lorentzweg 1, 2628 CJ Delft, Netherlands}
\affil[5]{Networked Quantum Devices Unit, Okinawa Institute of Science and Technology Graduate University, Okinawa, Japan}
\date{\today}

\begin{document}
\maketitle

\begin{abstract}

Motivated by the applications of secure multiparty computation as a privacy-protecting data analysis tool, and identifying oblivious transfer as one of its main practical enablers, we propose a practical realization of randomized quantum oblivious transfer. By using only symmetric cryptography primitives to implement commitments, we construct computationally-secure randomized oblivious transfer without the need for public-key cryptography or assumptions imposing limitations on the adversarial devices. We show that the protocol is secure under an indistinguishability-based notion of security and demonstrate an experimental implementation to test its real-world performance. Its security and performance are then compared to both quantum and classical alternatives, showing potential advantages over existing solutions based on the noisy storage model and public-key cryptography.

\end{abstract}

\section{Introduction}
\label{sec:introduction}

Cryptography is a critical tool for data privacy, a task deeply rooted in the functioning of today's digitalized world. Whether it is in terms of secure communication over the Internet or secure data access through authentication, finding ways of protecting sensitive data is of utmost importance. The one-time pad encryption scheme allows communication with perfect secrecy~\cite{shannon49}, at the cost of requiring the exchange of single-use secret (random) keys of the size of the communicated messages. Distribution of secret keys, therefore, is considered one of the most important tasks in cryptography. Modern cryptography relies heavily on conjectures about the computational hardness of certain mathematical problems to design solutions for the key distribution problem. However, as quantum computers threaten to make most of the currently used cryptography techniques obsolete~\cite{shor94}, better solutions for data protection are needed. This transition towards quantum-resistant solutions becomes particularly crucial when it comes to protecting data associated with the government, finance and health sectors, being already susceptible to \textit{intercept-now-decrypt-later} attacks.
Cryptography solutions secure in a post-quantum world, where large-scale quantum computers will be commercially available, have been explored in two directions. Classical cryptography based solutions, also referred as post-quantum cryptography~\cite{lyubashevsky2013, regev2009, bernstein2017}, involve using a family of mathematical problems that are conjectured to be resilient to quantum computing attacks. On the other hand, quantum cryptography based solutions~\cite{pirandola2020} using the laws of quantum mechanics can offer information-theoretic security, depending on the physical properties of quantum systems rather than computational hardness assumptions. Quantum Key Distribution (QKD)~\cite{wiesner83} is the most well-studied and developed of these quantum solutions, while other works beyond QKD have been proposed~\cite{broadbent2016}.

It is noteworthy that secure communication is not the only cryptographic task where end-users' private data may be exposed to an adversary. Cryptography beyond secure communication and key distribution includes zero-knowledge proofs, secret sharing, contract signing, bit commitment (BC), e-Voting, secure data mining, etc.~\cite{lindell2005}. A huge class of such problems can be cast as Multi-Party Computation (MPC), where distrustful parties can benefit from a joint collaborative computation on their private inputs. It requires parties' individual inputs to remain hidden from each other during the computation, among other security guarantees such as correctness, fairness, etc.~\cite{lindell2020}. Secure MPC is a powerful cryptographic tool with a vast range of applications as it allows collaborative work with private data. Generic MPC protocols work by expressing the function to evaluate as an arithmetic or Boolean circuit and then securely evaluating the individual gates. These protocols are based on one of two main fundamental primitives\cite{zhao19,yao82, goldreich2019, damgard12}: \emph{Oblivious Transfer} (OT) and \emph{Homomorphic encryption}, the former of which is the focus of this work.

A 1-out-of-2 OT~\cite{even85}, is the task of sending two messages, such that the receiver can choose only one message to receive, while the sender remains oblivious to this choice. The original protocol, now called all-or-nothing OT, was proposed by Rabin in 1981~\cite{rabin2005}, where a single message is sent and the receiver obtains it with 1/2 probability. The two flavours of OT were later shown to be equivalent~\cite{crepeau88}. Notably, it has been shown that it is possible to implement secure MPC using only OT as a building block~\cite{yao86, kilian88}. Relevant to our work is a variation of OT called Random Oblivious Transfer (ROT), which is similar to 1-out-of-2 OT, except that both the sent messages and the receiver's choice are randomly chosen during the execution of the protocol. This can be seen as analogous to the key distribution task, in which both parties receive a random message (the key) as output. By appropriately encrypting messages using the outputs of a ROT protocol as a shared resource, it is possible to efficiently perform 1-out-of-2 OT. As an important consequence, parties expecting to engage in MPC in the future can execute many instances of ROT in advance and save the respective outputs as keys to be later used as a resource to perform fast OTs during an MPC protocol~\cite{lemus2020}. Because of this, we can think of ROT as a basic primitive for secure MPC. 

In the context of quantum cryptography, OT is remarkable because, unlike classically, there exists a reduction from OT to commitment schemes~\cite{bennett91}. This result is somewhat undermined by the existence of several theorems regarding the impossibility of unconditionally secure commitments both in classical~\cite{canetti01} and quantum~\cite{lo1997, mayers1997} cryptography, and it was further proven impossible in the more general abstract cryptography framework~\cite{Maurer11}. These results, in turn, imply that unconditionally secure OT itself is impossible. In light of this, approaches with different technological or physical constraints on the adversarial power have been proposed. Practical solutions based on hardware limitations, such as bounded and noisy storage~\cite{damgaard2008, erven2014, liu2014, furrer2018}, have the disadvantage that the performance of such protocols decreases as technology improves.

Computationally-secure classical protocols have also been proposed~\cite{mansy2019, mi2018, mi2019, branco2021}, which work under the assumptions of post-quantum public-key cryptography. Alternatively, we can take advantage of quantum reduction from OT to commitments by implementing commitment schemes using (non-trapdoor) one-way functions (OWF) such as Hash functions~\cite{halevi96} and pseudo-random generators~\cite{naor91} which allows us to construct OT from symmetric cryptography primitives. The existence of general OWFs is a weaker assumption than public-key cryptography~\cite{impagliazzo95, barak17}, which requires the existence of the more restrictive \textit{trapdoor} OWFs. This difference is significant, as the latter are defined over mathematically rich structures, such as elliptic curves and lattices, and the computational hardness of the associated problems is less understood than that of their private-key counterparts. For an in-depth study of the relation between OT and OWFs see~\cite{grilo21}.

Having established that there is a theoretical merit in using computationally-secure quantum protocols to implement secure MPC, it is also important to understand how practical quantum protocols compare with currently used classical solutions in security, computational and communication complexity, and practical speed in current setups. This work focuses in studying the performance of a practical quantum ROT protocol and its potential advantages compared to currently used classical solutions for OT during MPC.
 

The idea of using quantum conjugate coding and commitments for oblivious transfer was originally proposed by Cr{\'e}peau and Kilian~\cite{crepeau88} and then refined by Bennet et al, in~\cite{bennett91} with the BBCS92 protocol (shown in Fig.~\ref{fig:simplifiedprotocol}). This construction has been extensively studied from the point of view of its theoretical security~\cite{yao95, damgard09, unruh10, grilo21, bartusek21, santos22}. However, while practical security analyses and experimental implementations have been made for quantum OT in the noisy storage model~\cite{liu2014,furrer2018}, there are no works analyzing the quantum resource requirements and the resulting performance of implementing the BBCS92 protocol using existing computationally-secure commitment schemes based on OWFs. Such analyses are needed to demonstrate secure experimental implementations, and provide an important step in bringing quantum OT to real-world usage.

\begin{figure}
\noindent
\framebox{\parbox{\dimexpr\linewidth-2\fboxsep-2\fboxrule}{\small
  {\centering \textbf{BBCS92 Quantum OT protocol} \par}
  \textbf{Parties:} The sender Alice and the receiver Bob. \\
  \begin{enumerate}[leftmargin=0.6cm]
      \item Alice prepares $N$ entangled states of the form $\frac{1}{\sqrt{2}}(\ket{00} + \ket{11})$ and, for each state prepared, sends one of the qubits to Bob.
      \item Alice randomly chooses a measurement bases string $\theta^A \in \lbrace +, \times \rbrace^{N}$ and, for each $i = 1,\ldots,N$ measures her share of the $i$-th entangled state in the $\theta_i$ basis to obtain outcome $x^{A}_i$ and the outcome string $x^{A}= (x^A_1,\ldots, x^A_N)$.
      \item Bob uses the same process to obtain the measurement bases and outcome strings $\theta^B$ and $x^B$, respectively.
      \item For each $i$, Bob commits $(\theta^B_i,x^B_i)$ to Alice.
      \item Alice chooses randomly a set of indices $T \subset \lbrace 1,\ldots, N \rbrace$ of some fixed size and sends $T$ to Bob.
      \item For each $j \in T$, Bob opens the commitments associated to $(\theta^B_j,x^B_j)$.
      \item Alice checks that $x^A_j = x^B_j$ whenever $\theta^A_j = \theta^B_j$ within the test set. If the test fails Alice aborts the protocol, otherwise she sends the string $\theta^A$ to Bob.
      \item Bob separates the remaining indices in two sets: $I_0$ - the indices where Bob's measurement bases match Alice's, and $I_1$ - the set of indices where their bases do not match. Then, he samples randomly $c$ and sends the ordered pair $(I_c, I_{\bar{c}})$ to Alice.
      \item Alice defines the strings $\mathbf{x}^A_c, \mathbf{x}^A_{\bar{c}}$ using the indices in the respective sets $(I_c, I_{\bar{c}})$. Then, she samples randomly a function $f$ from a universal hash family, sends $f$ to Bob and outputs $m_c = f(\mathbf{x}_0)$ and $m_1= f(\mathbf{x}_{\bar{c}})$ to Bob.
      \item Similarly, Bob defines the string  $\mathbf{x}^B$ from the set $I_0$ and outputs $m_c = f(\mathbf{x}^B)$ and $c$.
    \end{enumerate}
  }}
  \caption{Quantum oblivious transfer protocol based on commitments}
  \label{fig:simplifiedprotocol}
\end{figure}

Motivated by practical considerations, we consider Naor-style statistically binding and computationally hiding commitments, as these are well understood and efficient to implement (note that stronger commitments can be considered, such as the quantum-based commitments studied in~\cite{grilo21, bartusek21}, however, implementing those requires significantly more computational and quantum resources).

The contributions of this work can be summarized as follows: 

We introduce the definition for a quantum ROT protocol, satisfying a strong indistinguishability-based security notion equivalent to the one presented in~\cite{konig12}, which generalizes the security of classical ROT protocols. We present a protocol that realizes said quantum ROT based on the BBCS construction. The protocol uses a weakly-interactive string commitment scheme which is statistically binding and computationally hiding, and can be implemented in practice using current QKD setups.

We present a formal finite-key security proof of the proposed protocol accounting for noisy quantum channels assuming only the existence of quantum-secure OWFs, together with security bounds as functions of the protocol's parameters. We also present calculations for the maximum usable channel error, as well as for the key rate as a function of the number of shared signals per instance of the protocol. Additionally, we study the composability properties of said protocol. In particular, we show that there is a family of weakly-interactive commitments which, when used in the quantum OT protocol, result in universally composable quantum OT in the classical access random oracle model. We experimentally demonstrate our protocol using current technology with a setup based on polarization-entangled photons. We also present a security analysis which accounts for potential implementation-specific attacks and how they can be circumvented using an appropriate reporting strategy. Finally, we compare our performance results with the performance of current ROT solutions and point out the advantages and disadvantages of using quantum ROT in the context of MPC.

\section{Quantum Random Oblivious Transfer (ROT)}
\label{sec:basicdefinitions}

In this work, the concept of indistinguishability will be often used to compare the state of systems in a ``real'' run of the protocol versus another ``ideal'' desired state. These relations are defined over families of quantum states parametrized by the security parameter of the respective protocol. Hence, indistinguishability relations are statements on the asymptotic behavior of the protocol as the security parameter is increased. For formal definitions of both statistical and computational indistinguishability see Appendix~\ref{sec:preliminaries}.

When talking about two indistinguishable families $\lbrace \rho_1^{(k)} \rbrace$ and $\lbrace \rho_2^{(k)} \rbrace$, if the parameter $k$ is implicit, we will just refer to them as $\rho_1$ and $\rho_2$ and use the following notation to denote indistinguishability:
\begin{align*}
    \rho_1 \approx \rho_2 &\quad \textnormal{for statistically indistinguishable}; \\
    \rho_1 \approx^{(c)} \rho_2 &\quad \textnormal{for computationally indistinguishable}.
\end{align*}

Additionally, in this work we consider protocols that can abort if certain conditions are satisfied. Mathematically, it is useful to consider the state of the aborted protocol as the zero operator. This means that events that trigger the protocol to abort are described as \textit{trace-decreasing} operations, and hence, the operator representing the associated system at the end of the protocol is, in general, not normalized. The probability of the protocol finishing successfully is given then by the trace of the final state of the output registers. Note that the above definitions of indistinguishability can be naturally extended to non-normalized operators since the outcomes of a quantum program can always be represented by the outcomes of a POVM $\lbrace F_i \rbrace$, whose probabilities are given by $\Tr[F_i \rho]$, which is a well defined quantity even for non-normalized $\rho$. 

\begin{definition} (Quantum Random Oblivious Transfer)\\
\label{def:ROT}
An $n$-bit Quantum Random Oblivious Transfer with security parameter $k$ is a protocol, without external inputs, between two parties \textit{S} (the sender) and \textit{R} (the receiver) which, upon finishing, outputs the joint quantum state $\rho_{M_0,M_1,C,M_C}$ satisfying: 
\begin{enumerate}
    \item (Correctness) The final state of the outputs when the protocol is run with both honest parties satisfies
    \begin{align}
        \label{eq:qot_correctness}
        \rho_{M_0,M_1,C,M_C} \approx \frac{p_{\textnormal{succ}}}{2^{(2n+1)}}
        \!\!\!\!\!\!\!\!\!
        \sum\limits_{\substack{m_0,m_1 \in \bs^{n} \\ c \in \bs}}
        \!\!\!\!\!\!\!\!\!
        \Big( \ketbra{m_0}_{M_0} \ketbra{m_1}_{M_1} \ketbra{c}_{C} \ketbra{m_c}_{M_C} \Big),
    \end{align}
    where $p_{\textnormal{succ}} = \textnormal{Tr}[\rho_{M_0,M_1,C,M_C}]$ is the probability of the protocol finishing successfully.
    \item (Security against dishonest sender) Let $H_S$ be the Hilbert space associated to all of the sender's memory registers. For the final state after running the protocol with an honest receiver it holds that
    \begin{equation}
        \label{eq:uncertain_b}
        \rho_{S,C} \approx  \rho_{S} \otimes \mathsf{U}_{C}.
    \end{equation}
    \item (Security against dishonest receiver) Let $H_R$ be the Hilbert space associated to all of the receiver's memory registers. For the final state after running the protocol with an honest sender, there exists a binary probability distribution given by $(p_0, p_1)$ such that
    \begin{equation}
        \label{eq:uncertain_m}
        \rho_{R,M_0,M_1} \approx  \sum_{b} \left(p_b \, \rho_{R,M_{\bar{b}}}^b \otimes  \mathsf{U}_{M_{b}}\right).
    \end{equation}
\end{enumerate}
The above properties define statistical security for each feature of the ROT protocol. If any of them holds for the case of a dishonest party being limited to efficient quantum operations and the notion of computational indistinguishability $\approx^{(c)}$ instead, we say that the ROT protocol is computationally secure in the respective sense. 
\end{definition}

We expect the outputs $m_0,m_1,c$ to be uniformly distributed and the receiver always obtaining the correct corresponding $m_c$. The first property is typically called \textit{correctness} and it states that, when both parties follow the protocol, the probability of it not aborting \textit{and} having incorrect outputs is neglible in the security parameter. The probability $p_{\textnormal{succ}}$ of the protocol finishing appears explicitly in this expression as the success of quantum protocols often depends on external conditions, most notably the noise in the quantum communication channels. For any specific value of $p_{\textnormal{succ}}$ and any $\varepsilon^{r} \leq 1 - p_{\textnormal{succ}}$ we say that, under the associated external conditions, the protocol is $\varepsilon^{(r)}$\textit{-robust}.

The second property, called \textit{security against dishonest sender}, states that regardless of how much the sender deviates from the protocol, their final quantum state (which includes all the information accessible to them) is uncorrelated to the uniformly distributed value of the receiver's choice bit $c$. Analogously, the third property, called \textit{security against dishonest receiver}, states that even for a receiver running an arbitrary program, by the end of the protocol there is always at least one of the two strings $m_0,m_1$ that is completely unknown to them (denoted by $m_b$).

\subsection{Additional schemes}

In this section, we define the subroutines used inside of our main protocol. We start by defining a weakly-interactive commitment scheme, which gets its name from the fact that the verifier publishes a single random message at the start, which defines the operations that the committer performs. 

\begin{definition} (String commitment scheme)\\
\label{def:sc}
Let $k, n \in \mathbb{N}$. A weakly-interactive $n$-bit string commitment scheme with security parameter $k$ is a family of efficient (in $n$, as well as in $k$) programs $\com, \open, \ver$

\begin{equation}
\begin{split}
    \com &: \bs^n \times \bs^{n_s(k)} \times \bs^{n_r(k)}\rightarrow \bs^{n_c(k)}; \\
    \open &: \bs^n \times \bs^{n_s(k)} \rightarrow \bs^{n_o(k)}; \\
    \ver &: \bs^{n_c(k)} \times \bs^{n_o(k)} \times \bs^{n_r(k)}\rightarrow \bs^n \cup \lbrace \bot \rbrace,
\end{split}
\end{equation}

such that
\begin{enumerate}
    \item (correctness) $\ver\big(\com(m,s,r), \open(m,s), r\big) = m$ for all $m \in \bs^n$, $s \in \bs^{n_s}$, and $r \in \bs^{n_r}$.
    \item (hiding property) For all $m_1, m_2 \in \bs^n$ and $r \in \bs^{n_r}$ the distributions for $\com(m_1,s_1,r)$ and $\com(m_2, s_2, r)$ are computationally (or statistically) indistinguishable in $k$ whenever $s_1, s_2$ are sampled uniformly.
    \item (binding property) For uniformly sampled $r$, the probability $\varepsilon_{\textnormal{bind}}(k)$ that there exists a tuple $(\ccom, \oopen_1, \oopen_2)$ such that $\ver(\ccom, \oopen_{1/2}, r) \neq \bot$ and
    \begin{align}
        \label{eq:commitmentbinding}
        \ver(\ccom, \oopen_1, r) \neq \ver(\ccom, \oopen_2, r), 
    \end{align}
    is negligible in $k$.
\end{enumerate}
\end{definition}

Weakly-interactive string commitment schemes can be implemented using common cryptographic primitives like hash functions or pseudo-random generators. Most notably, Naor's commitment protocol~\cite{naor91} provides a black box construction of weakly-interactive commitments from OWFs.

\begin{definition} (Verifiable information reconciliation scheme)\\
\label{def:vir}
Let $\mathcal{C} \subseteq \lbrace 0,1 \rbrace^{n} \times \lbrace 0,1 \rbrace^{n}$. A verifiable one-way Information Reconciliation (IR) scheme with security parameter $k$ and leak $\ell$ for $\mathcal{C}$ is a pair of efficient programs $(\syn, \dec)$ with
\begin{equation}
\begin{split}
    \syn&: \lbrace 0,1 \rbrace^{n} \rightarrow \lbrace 0,1 \rbrace^{\ell},\\
    \dec&:\lbrace 0,1 \rbrace^{\ell} \times \lbrace 0,1 \rbrace^{n} \rightarrow \lbrace 0,1 \rbrace^{n} \cup \lbrace \bot \rbrace,
\end{split}
\end{equation}
such that, 
\begin{enumerate}
    \item (correctness) Whenever $(x,y) \in \mathcal{C}$ it holds that $\dec(\syn(x),y)=x$ except with negligible probability in~$k$.
    \item (verifiability) For any $(x,y) \in \lbrace 0,1 \rbrace^{n} \times \lbrace 0,1 \rbrace^{n}$ it holds that either $\dec(\syn(x),y)=x$ or $\dec(\syn(x),y)=\bot$, except with negligible probability $\varepsilon_{\textnormal{IR}}(k)$.
\end{enumerate}
\end{definition}

Due to Shannon's Noisy-channel coding theorem, the size of the leak $\ell$ for any IR scheme over a discrete memoryless channel is lower bounded by $h(p)$, where $p$ represents the bit-error probability, and $h(\cdot)$ denotes the binary entropy function. For concrete IR schemes, we can usually describe their efficiency using the ratio between the scheme's leak and the theoretical optimal: $f = \frac{\ell}{h(p)}$.

\section{The protocol}
\label{sec:protocol}
In this section we present the protocol $\pi_{\textnormal{QROT}}$ for an $n$-bit quantum ROT based on the primitives described in the previous section and the use of quantum communication. The protocol's main security parameter is $N_0$, which corresponds to the number of quantum signals sent during the quantum phase. Additionally, it has two secondary security parameters $k,k'$, which define the security of the underlying commitment and IR schemes, respectively. 

In order to facilitate the finite-key security analysis, the description of $\pi_{\textnormal{QROT}}$ features two statistical tolerance parameters, denoted as $\delta_1, \delta_2 $. The role of $\delta_1$ is to account for the error in the estimation of the Qubit Error Rate (QBER), while the role of $\delta_2$ is to account for the small variations in the frequency of outcomes of $50/50$ events. These parameters can be ignored (set to zero) when considering very large values of $N_0$. 

In the following description of the protocol we use the common conjugate coding notation used in BB84-based protocols. The bit values $0,1$ the denote the computational and Hadamard bases for qubit Hilbert spaces, respectively. For added clarity, we use the superscripts $A$ and $B$ to respectively denote Alice and Bob. Additionally, we use variable $x$ to denote measurement outcomes and $\theta$ to denote measurement bases (e.g. the pair $(\theta^A_i, x^A_i)$ denotes that Alice measured her $i$-th subsystem in the $\theta^A_i$ basis and obtained $x^A_i$ as the outcome). We use $\ket{\Phi^+}$ to denote the Bell state $\frac{1}{\sqrt{2}}(\ket{00} + \ket{11})$. Finally, we will use the relative (or normalized) Hamming weight function $r_H:\bs^{n} \rightarrow [0,1]$ defined for any $x=(x_1, \ldots, x_n)$ as
\begin{equation}
          r_{\textnormal{H}}(x)=\frac{1}{n}\sum_{i=1}^{n}x_i.
\end{equation}

\noindent
\textbf{Parameters:}
\begin{itemize}
    \item Parameter estimation sample ratio $0 < \alpha < 1$
    \item Statistical tolerance parameters $\delta_1, \delta_2 $
    \item Maximum qubit error rate $0 \leq p_{\text{max}} \leq 1/2$
    \item Coincidence block size $N_{0} \in \mathbb{N}$, test set size $N_{\textnormal{test}} = \alpha N_{0}$, minimum check set size $N_{\textnormal{check}} = (\frac{1}{2}- \delta_2) \alpha N_{0}$, and raw string block size $N_{\textnormal{raw}} = (\frac{1}{2}- \delta_2) (1 - \alpha) N_{0}$
    \item Weakly-interactive 2-bit string commitment scheme $(\com, \open, \ver)$, which is computationally hiding and statistically binding, with security parameter $k \in \mathbb{N}$ and associated string lengths $n_s, n_r, n_c, n_o$ 
    \item Verifiable one-way information reconciliation scheme $(\syn, \dec)$ on the set $\mathcal{C} = \lbrace (x,y) \in \bs^{N_{\text{raw}}} \times \bs^{N_{\text{raw}}} : r_{\textnormal{H}}(x \oplus y)  < p_{\text{max}} + \delta_1 \rbrace $, with security parameter $k' \in \mathbb{N}$ and leak $\ell = f \cdot h(p_{\text{max}} + \delta_1)$
    \item Universal hash family $\textnormal{\textbf{F}}= \big\{ f_{i}: \bs^{N_{\text{raw}}} \rightarrow \bs^{n} \big\}_i$
\end{itemize}
\textbf{Parties:} The sender Alice and the receiver Bob. \\
\textbf{Protocol steps:} \\
\textit{Quantum phase}
\begin{enumerate}
      \item Alice generates the state $\bigotimes_{i=1}^{N_0} \ket{\Phi^+}_i$ and sends one qubit of each entangled pair to Bob through a (potentially noisy) quantum channel. Then she samples the string $\theta^{A}\in \bs^{N_0}$ and for each $i \in I = \lbrace 1,\ldots,N_0 \rbrace$ performs a measurement in the basis $\theta^{A}_i$ on her qubit of $\ket{\Phi^+}_i$ to obtain the outcome string $x^{A}$.
      \item Bob samples the string $\theta^{B}\in \bs^{N_0}$ and for each $i \in I$ performs a measurement in the basis $\theta^{B}_i$ on his qubit of $\ket{\Phi^+}_i$ to obtain the outcome string $x^{B}$.
\end{enumerate}
\textit{Commit/open phase} 
\begin{enumerate}
\setcounter{enumi}{2}
      \item Alice uniformly samples the string $r \in \bs^{n_r}$ and sends it to Bob.
      \item For each $i \in I$, Bob samples a random string $s_i \in \bs^{n_s}$, computes
      \begin{align}
            (\ccom_i, \oopen_i) = 
            \Big(&\com\big((\theta^{B}_i,x^{B}_i),s_i, r \big), \nonumber \\ &\open\big((\theta^{B}_i,x^{B}_i),s_i\big)\Big),
      \end{align}
      and sends the string $\ccom = (\ccom_i)$ to Alice.
      \item Alice randomly chooses a subset test $I_{t} \subset I$ of size $\alpha N_0$ and sends $I_{t}$ to Bob.
      \item For each $j \in I_{t}$, Bob sends $\oopen_j$ to Alice.
      \item For each $j \in I_{t}$, Alice checks that $\ver(\ccom_j, \oopen_j, r) \neq \bot$. If so, she sets $(\tilde{\theta}^{B}_j, \tilde{x}^{B}_j) = \ver(\ccom_j, \oopen_j, r)$.  Then, Alice computes the set $I_{s}=\lbrace j \in I_{t}  \vert \theta^{A}_j = \tilde{\theta}^{B}_j \rbrace$ and the quantity
      \begin{equation}
          p=r_{\textnormal{H}}\left( x^{A}_{I_{s}} \oplus \tilde{x}^{B}_{I_{s}} \right), 
      \end{equation}
      and checks that $|I_{s}| \geq N_{\textnormal{check}}$ and $p \leq p_{\text{max}}$. If any of the checks fail Alice aborts the protocol.
\end{enumerate}
\textit{String separation phase}
\begin{enumerate}
\setcounter{enumi}{7}
      \item Alice sends $\theta^{A}_{\Bar{I}_t}$ to Bob.
      \item Bob constructs the set $I_0$ by randomly selecting $N_{\text{raw}}$ indices $i \in \Bar{I}_t$ for which $\theta^{A}_i = \theta^{B}_i$. Similarly, he constructs $I_1$ by randomly selecting $N_{\text{raw}}$ indices $i \in \Bar{I}_t$ for which $\theta^{A}_i \neq \theta^{B}_i$. He then samples a random bit $c$ and sends the ordered pair $(I_c, I_{\Bar{c}})$ to Alice. If Bob is not able to construct $I_0$ or $I_1$, he aborts the protocol.
\end{enumerate}
\textit{Post processing phase}
\begin{enumerate}
\setcounter{enumi}{9}
      \item Alice computes the strings $\left(\syn(x^{A}_{I_{c}}), \syn(x^{A}_{I_{\Bar{c}}})\right)$ and sends the result to Bob.
      \item Bob computes $\dec\left(x^{B}_{I_{0}}, \syn(x^{A}_{I_{0}})\right) = y^{B}$. If $y^{B} = \bot$ Bob aborts the protocol. 
      \item Alice randomly samples $f \in \textnormal{\textbf{F}}$, computes $m^{A}_0 = f(x^{A}_{I_{c}})$ and $m^{A}_1 = f(x^{A}_{I_{\Bar{c}}})$, sends the description of $f$ to Bob and outputs $(m^{A}_0, m^{A}_1)$.
      \item Bob computes $m^{B}=f(y^B)$ and outputs $(m^{B},c)$.
\end{enumerate}

\subsection{Security and performance of the main protocol}
We start by stating the main theorem regarding security of the proposed $\pi_{\textnormal{QROT}}$ protocol.

\begin{theorem} (Security of $\pi_{\textnormal{QROT}}$)\\
\label{thm:OKDsecurity}
The protocol $\pi_{\textnormal{QROT}}$ is a statistically correct, computationally secure against dishonest sender, and statistically secure against dishonest receiver $n$-bit ROT protocol.
\end{theorem}
A high-level proof of Theorem~\ref{thm:OKDsecurity}, including the derivation of the security bounds from Lemmas~\ref{lem:correctness} and~\ref{lem:dishonestreceiver} can be found in Section~\ref{sec:secanalysis} and further details can be found in Appendix~\ref{sec:longproof}. The security of $\pi_{\textnormal{QROT}}$ is given by its main security parameter $N_0$, as well as the security parameters of the underlying commitment and IR schemes $k$ and $k'$, respectively. These values can be computed for the statistical security features of the protocol and are given by the following lemmas:

\begin{lemma} (Correctness) \\
    \label{lem:correctness}
    The outputs of $\pi_{\textnormal{QROT}}$ when run by honest sender and receiver satisfy
    \begin{align}
        \label{eq:lemma_correctness}
        \rho_{M_0,M_1,C,M_C} \approx_{\varepsilon} \frac{p_{\textnormal{succ}}}{2^{(2n+1)}}
        \!\!\!\!\!\!\!\!\!
        \sum\limits_{\substack{m_0,m_1 \in \bs^{n} \\ c \in \bs}}
        \!\!\!\!\!\!\!\!\!
        \Big( \ketbra{m_0}_{M_0} \ketbra{m_1}_{M_1} \ketbra{c}_{C} \ketbra{m_c}_{M_C} \Big),
    \end{align}
    with
    \begin{equation}
    \varepsilon = 2^{-\frac{1}{2}(N_{\textnormal{raw}}-n)} + 2\varepsilon_{\textnormal{IR}}(k'),
    \end{equation}
    where $\varepsilon_{\textnormal{IR}}$ is a negligible function given by the security of the underlying IR scheme.
\end{lemma}

\begin{lemma} (Security against dishonest receiver) \\
    \label{lem:dishonestreceiver}
    For the final state after running the protocol of $\pi_{\textnormal{QROT}}$ with an honest sender, there exists a binary probability distribution given by $(p_0, p_1)$ such that
    \begin{equation}
        \label{eq:dishonestreceiver}
        \rho_{R,M_0,M_1} \approx_{\varepsilon'}  \sum_{b} \left(p_b \, \rho_{R,M_{\bar{b}}}^b \otimes  \mathsf{U}_{M_{b}}\right), 
    \end{equation}
    with 
\begin{align}
    \label{eq:epsilondisonestbob}
    \varepsilon' &= \sqrt{2} \left(e^{-\frac{1}{2} (1-\alpha)^{2} N_{\textnormal{test}} \delta^{2}_1} +  e^{-\frac{1}{2} N_{\textnormal{check}} \delta^{2}_1} \right)^{\frac{1}{2}} + e^{-D_{KL}(\frac{1}{2} - \delta_2|\frac{1}{2})(1-\alpha) N_0}  + \varepsilon_{\textnormal{bind}}(k) \\ \nonumber
    &\quad  + \frac{1}{2} \cdot 2^{\frac{1}{2}\left(n - N_{\text{raw}} \left( \frac{1}{2} - \frac{2\delta_2}{1-2\delta_2} - h\left( \frac{p_{\text{max}} + \delta_1}{\frac{1}{2} - \delta_2}\right) - f \cdot h(p_{\text{max}} + \delta_1) \right) \right)}.
\end{align}
    where $\mathcal{H}_R$ denotes the Hilbert space associated to all of the receiver's memory registers and $\varepsilon_{\textnormal{bind}}$ is a negligible function given by the security of the underlying commitment scheme.
\end{lemma}

We can use these results to find the minimum requirements, both in terms of channel losses and number of shared entangled qubits, necessary to securely realize ROT for a given security level. We focus on the quantity
\begin{equation}
    \varepsilon_{\max}=\varepsilon+\varepsilon'.
\end{equation}
For the purposes of this analysis, we assume that the commitment and IR schemes, as well as their security parameters $k,k'$, are appropriately chosen to satisfy the desired security level and we focus on the dependence of $\varepsilon_{\max}$ on the channel error rate, characterized by the parameter $p_{\max}$, and the number of quantum signals $N_0$.  
We are also interested in a quantity known as the secret key rate $R_{\textnormal{key}}$. For given values of $N_0$, $\alpha$, 
$\delta_1$, $\delta_2$, $p_{\max}$, and $\varepsilon_{\max}$, let $n_{\max}$ be the largest number for which the associated $n_{\max}$-bit ROT has at least security $\varepsilon_{\max}$, then
\begin{equation}
    R_{\textnormal{key}} = \frac{n_{\max}}{N_0},
\end{equation}
represents the ratio in which the original measurements of the shared qubits ``transform'' into the oblivious key. In Figure~\ref{fig:MaxKeyRatesPlot} we can see the behavior of $R_{\textnormal{key}}$ as $p_{\max}$ increases. Note that, similarly to the case of quantum key distribution, there is a critical error $p_{\textnormal{crit}}$ after which $R_{\textnormal{key}}$ becomes negative and no secure key can be generated. The value of $p_{\textnormal{crit}}$ is upper bounded by $\approx 0.028$, which is achieved when we set $\alpha, \delta_1, \delta_2 \rightarrow 0$ and $N_0 \rightarrow \infty$.

\begin{figure}
    \centering
    \includegraphics[width=0.6\linewidth]{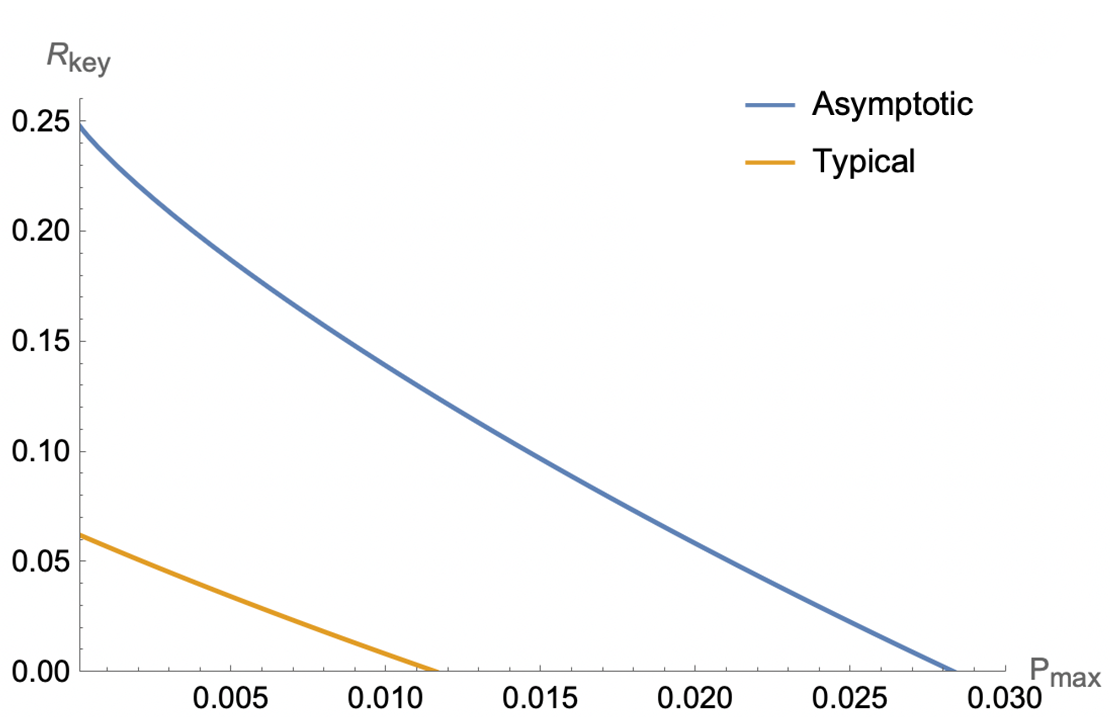}
    \caption{\textbf{Maximum key rate output $\frac{n}{N_0}$ versus error rate $p_{\max}$.} The blue line represents the upper bound for the key rate, when $N_0 \rightarrow \infty$, $\alpha, \delta_1, \delta_2$ are taken to be 0 and $f = 1$. The orange line represents a more typical case with $\alpha = 0.35$, $\delta_1 = 0.01, \delta_2 = 0.025$, and $f=1.2$.}
    \label{fig:MaxKeyRatesPlot}
\end{figure}
\begin{figure}
    \centering
    \includegraphics[width=0.6\linewidth]{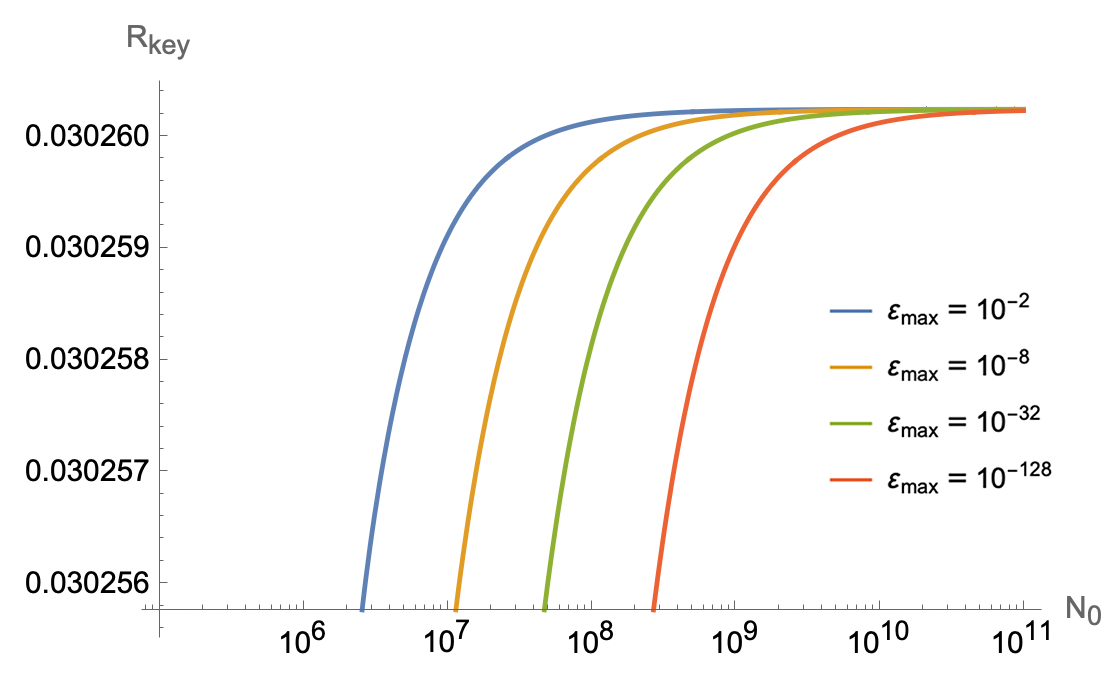}
    \caption{\textbf{Maximum key rate behaviour as a function of $N_0$ for different security levels.} Parameter values used are $\alpha =0.35; \delta_1 = \num{9.20e-3}; \delta_2 = \num{3.00e-3};p_{\textnormal{max}} = 0.01; f=1.2$.}
    \label{fig:KeyRatevsNPlot}
\end{figure}

Another important aspect to analyze is the relation between $R_{\textnormal{key}}$ and $N_0$, which is shown in Figure~\ref{fig:KeyRatevsNPlot}. Fixing the $\alpha, \delta_1, \delta_2, p_{\max}$, there is a clearly marked phase transition-like behaviour in which, for each $\varepsilon_{\max}$, there is a critical value of $N_0 = N_{\textnormal{crit}}$ before which  $R_{\textnormal{key}} = 0$, and after which it quickly reaches its maximum value. This result comes from the fact that the parameter estimation requires relatively big sample sizes to reach high confidence. It shows that, even for small $n$, there is a minimum amount of entangled qubits needed to be shared. In some cases, for instance, generating a $1$-bit oblivious key or a $128$-bit one may have similar costs in terms of quantum communication. Because the use of resources of the protocol scales with $N_0$, the parameters $\alpha, \delta_1, \delta_2$ should be chosen such that $N_{\textnormal{crit}}$ is the smallest. Figure~\ref{fig:MinNvsEPlot} exemplifies the dependency of $N_{\textnormal{crit}}$ on $\varepsilon_{\max}$.

\begin{figure}
    \centering
    \includegraphics[width=0.6\linewidth]{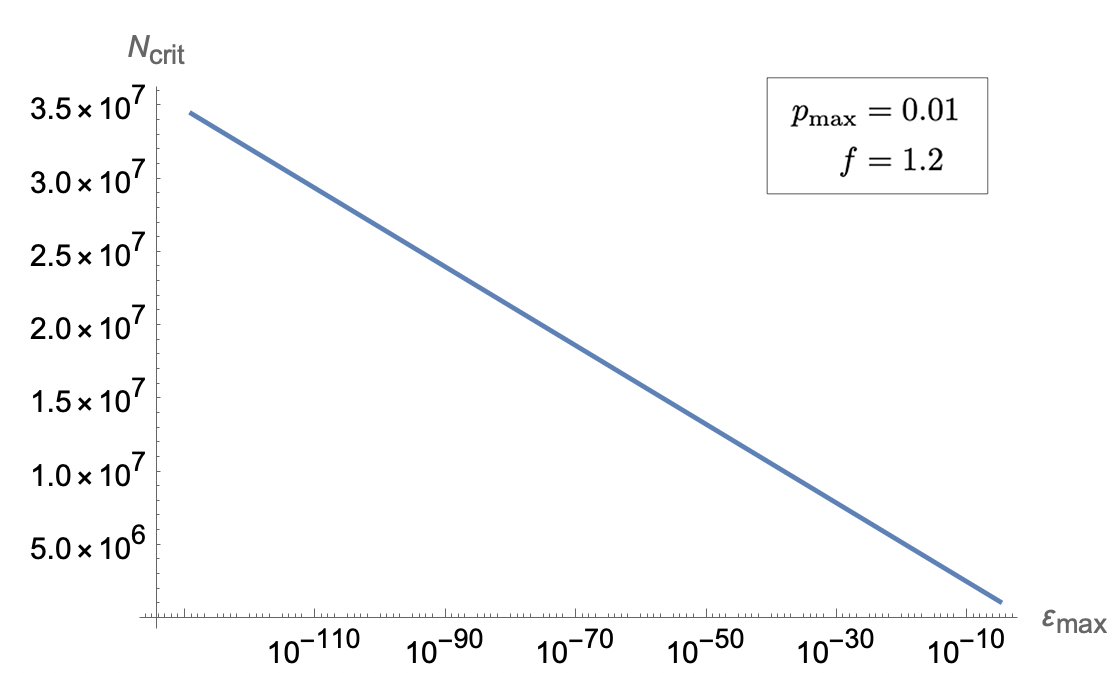}
    \caption{\textbf{Critical value $N_{\textnormal{crit}}$ of the number of shared qubits needed to obtain positive key rates as a function of the security level.} The values of $N_{\textnormal{crit}}$ were computed using the parameters $\alpha, \delta_1, \delta_2$ that minimize the value of $N_{\textnormal{crit}}$ for each $\varepsilon_{\max}$.}
    \label{fig:MinNvsEPlot}
\end{figure}

\subsection{Experimental implementation performance} 

An experiment was implemented to test the performance of the $\pi_{\textnormal{QROT}}$ protocol with contemporary technology. Data was acquired using a picosecond pulsed photon source in a Sagnac configuration ~\cite{sagnac06}, producing wavelength degenerate, polarization-entangled photons at 1550nm. In this setup, entangled photons were produced via spontaneous parametric down conversion (SPDC) by applying a laser pump beam into a 30mm long periodically-poled potassium titanyl phosphate (ppKTP) crystal. The photon pairs were split using a half-wave plate (HWP) and a polarizing beam splitter (PBS), and then sent to each party where they are detected using superconducting nanowire single-photon detectors.

To test the OT speed of this implementation, different values for the power $P$ of the laser pump were tested, as well as the use of multiplexing. As the $P$ increases, the amount of coincidences detected per second $R_{\textnormal{c}}$ increases, but the fidelity of the produced entangled pairs decreases, resulting in larger values for qubit error rate, which is represented by the protocol parameter $p_{\max}$. The number of maximum potential OT instances per second is computed as 
\begin{equation}
    R_{OT} = \frac{R_{\textnormal{c}}}{N_{\textnormal{crit}}},
\end{equation}
where $N_{\textnormal{crit}}$ is computed using the optimal values of $\alpha, \delta_1, \delta_2$ for the respective error rate $p_{\max}$ and undetected multi-photon rate $p_{\textnormal{multi}}$ associated to $P$, assuming perfectly efficient information reconciliation, $f=1$ (see Section \ref{sec:experiment} for the details on the implementation and its security). As seen in Figure~\ref{fig:expperformance}, for this implementation, the additional coincidence rate gained by increasing $P$ is not enough to compensate for the induced increased error. This result is not immediately obvious, as $N_{\textnormal{crit}}$ does not depend explicitly on $p_{\max}$. The decrease in performance comes from the fact that increasing $p_{\max}$ limits the values that $\delta_1$ can have while maintaining positive key rates. This restriction on the values of $\delta_1$ ultimately results in an increase in $N_{\textnormal{crit}}$ and therefore, a reduction on $R_{OT}$.

\begin{figure}[H]
    \centering
    \includegraphics[width=0.6\linewidth]{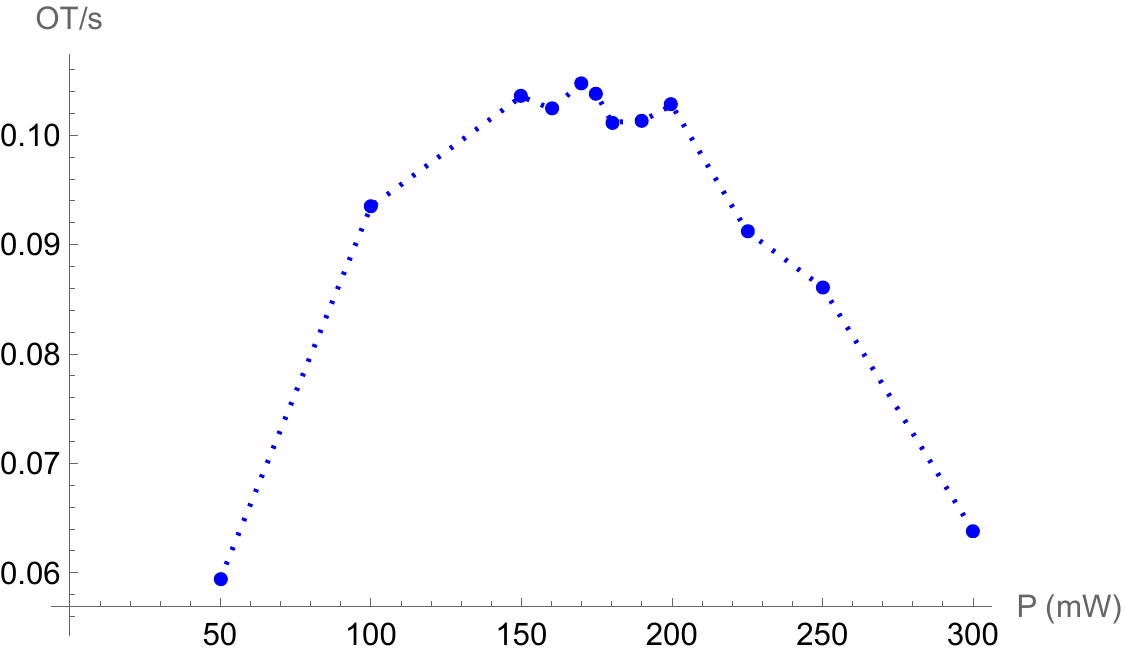}
    \caption{Maximum potential ROT rates as a function of the pump power P for $\varepsilon_{\max} = 10^{-7}$. We see that the best performance is obtained at a laser pump power of $P = 170$ mW, corresponding to a coincidence rate close to $\num{2.45}$ kHz. The uncertainty on the power measurement (x-axis) along with the error bars resulting from the Poissonian noise on the coincidence counts (used to calculate y-values) are negligible with respect to the current plot scale.}
    \label{fig:expperformance}
\end{figure}

Table~\ref{tab:expperformance} shows an example of the performance of the protocol in a real-world implementation using the data from the experimental setup. For the commit/open phase, the weakly-interactive string commitment protocol introduced in \cite{naor91} was implemented using the BLAKE3 hash function algorithm as a one-way function. For the post-processing phase, a low density parity check (LDPC) code was used for IR, and random binary matrices were used to implement the universal hash family for privacy amplification. We evaluated the performance by the number of $128$-bit ROT instances able to be completed per second (It is worth noting that, using a Mac mini M1 2020 16GB computer, the post-processing throughput was enough to handle all the data from the experiment, the bottleneck being the quantum signal generation rate).

\begin{table}[htbp]
\centering
  \begin{tabular}{lcc}
   \toprule
   \multicolumn{1}{c}{\textbf{Parameter}}    & \textbf{Symbol} & \textbf{Value} \\ \midrule
    Message size (bits)         &   $n$                                 &  $128$                \\ \midrule
    Security level              &   $\varepsilon_{\textnormal{max}}$    &  $\num{1.91e-8}$         \\ \midrule
    Cost in quantum signals     &   $N_0$                               &  $\num{5.86e6}$       \\ \midrule
    Max allowed QBER            &   $p_{\textnormal{max}}$              &  $1.14\%$              \\ \midrule
    Testing set ratio           &   $\alpha$                            &  $0.35$               \\ \midrule
    Statistical parameter 1     &   $\delta_1$                          &  $\num{9.00e-3}$       \\ \midrule
    Statistical parameter 2     &   $\delta_2$                          &  $\num{3e-3}$         \\ \midrule
    IR verifiability security   &   $\varepsilon_{\textnormal{IR}}$     &  $2^{-32}$         \\ \midrule
    Commitment binding security &   $\varepsilon_{\textnormal{bind}}$   &  $2^{-32}$         \\ \midrule
    Efficiency of IR            &   $f$                                 &  $1.64$               \\ \midrule
    Max allowed multi-photon rate   &   $p_{\textnormal{multi}}$    &  $\num{3.67e-3}$             \\ \midrule
                                &   ROT rate                            &  $\num{0.023}$ ROT/s       \\ 
   \bottomrule
  \end{tabular}
  \caption{Table of protocols parameters and the resulting performance. The values of $N_0$ and $\delta_1$ and the laser pump power were optimized to yield the highest ROT rate for an LDPC code with efficiency f = 1.61.}
  \label{tab:expperformance}
\end{table}

\section{Security Analysis}
\label{sec:secanalysis}
In this section, we prove the main security result, which relates the overall security of the protocol as a function of its parameters $N_0, \alpha, \delta_1$, and $\delta_2$ in Theorem~\ref{thm:OKDsecurity}. For clarity of presentation, we have compacted some of the properties into lemmas, for which detailed proofs can be found in Appendix~\ref{sec:longproof}. Definitions and properties of entropic quantities can be found in Appendix~\ref{sec:preliminaries} 

\subsection{Correctness}
\label{sec:correctness}
In order to prove correctness we need to show that either the protocol either finishes with Alice outputting uniformly distributed messages $m_0, m_1$ and Bob outputting a uniformly random bit $c$ and the corresponding message $m_c$, or it aborts, except with negligible probability. 

Recall that we model the aborted state of the protocol as the zero operator. This way, whenever we have a mixture of states, some of which trigger aborting and some that do not, the abort operation removes the events that trigger it from the mixture, effectively reducing its trace by the probability of aborting. There are three instances where the protocol can abort: first during Step (7) if the estimated qubit error rate is larger than $p_{\max}$; the second one is during Step (9) if Bob does not have enough (mis)matching bases to construct the sets $I_0, I_1$; and finally during Step (11) if the IR verification fails. The probability of aborting in Steps (7) and (11) depends on the particular transformation that the states undergo when being sent from Alice's to Bob's laboratory, about which we make no assumptions. We can group these three abort events and denote by $p_{\textnormal{abort}}$ the probability of the protocol aborting by the end of Step (11). The state at this point can be written as $(1-p_{\textnormal{abort}}) \rho^{\top}$, where $\rho^{\top}$ represents the normalized state conditioned that the protocol has not aborted by this point. As Lemma~\ref{lem:IRverifiability} states, the verifiability property of the Information Reconciliation scheme guarantees that the states that ``survive'' past Step (11) have the property that Bob's corrected string $y^{B}$ is the same as Alice's outcome string $x^{A}_{I_0}$, which is uniformly distributed. 

\begin{lemma}
    \label{lem:IRverifiability}
    Let $X^{A}_{I_0}, X^{A}_{I_1}, C, Y^{B}$ denote the systems holding the information of the respective values $x^{A}_{I_0}, x^{A}_{I_1}, c$, and  $y^{B}$ of $\pi_{\textnormal{QROT}}$. Denote by $\rho^{\top}$ the parties' joint state at the end of Step (11) conditioned that Bob constructed the sets $(I_0, I_1)$ during Step (9) and the protocol has not aborted. Assume both parties follow the Steps of the protocol, then
    \begin{equation}
        \rho^{\top}_{X^{A}_{I_0}, X^{A}_{I_1}, C, Y^{B}} \approx_{\varepsilon_{\textnormal{IR}}(k')} \tilde{\rho}^{ \top}_{X^{A}_{I_0}, X^{A}_{I_1}, C, Y^{B}},
    \end{equation}
    where $\varepsilon_{\textnormal{IR}}(k')$ is a negligible function given by the security of the underlying Information Reconciliation scheme, $k'$ its associated security parameter, and
        \begin{equation}
        \label{eq:IRverifiability}
        \tilde{\rho}^{ \top}_{X^{A}_{I_0}, X^{A}_{I_1}, C, Y^{B}} = \frac{1}{2^{(2N_{\textnormal{raw}}+1)}}
        \!
        \sum_{\substack{x_{I_0}, x_{I_1} \\ c}}
        \!
        \ketbra{x_{I_0}}_{X^{A}_{I_0}} \ketbra{x_{I_1}}_{X^{A}_{I_1}}  \ketbra{x_{I_0}}_{Y^{B}} \ketbra{c}_{C}.
    \end{equation}
\end{lemma}

During Step (12) universal hashing is used in both $x^{A}_{I_0}$ and $x^{A}_{I_1}$ to obtain $m_c$ and $m_{\Bar{c}}$. Because Eq.~\eqref{eq:IRverifiability} describes a state for which the $X^{A}_{I_0}, X^{A}_{I_1}$, and $C$ subsystems are independent and uniformly distributed, it follows from Lemma~\ref{lem:IRverifiability} that 
\begin{align}
    \label{eq:entropyboundcorrectness}
    H_{\textnormal{min}}^{\varepsilon_{\textnormal{IR}}(k')}( X^{A}_{I_0} Y^{B}\vert X^{A}_{I_1}C)_{\rho^{\top}} &= H_{\textnormal{min}}^{\varepsilon_{\textnormal{IR}}(k')}( X^{A}_{I_1}\vert X^{A}_{I_0} Y^{B}C)_{\rho^{\top}} \nonumber \\
    &= N_{\textnormal{raw}}.
\end{align}
Finally, using the quantum leftover hash Lemma~\ref{lem:qlhl} twice (once for $m_0$ and $m_1$) with the corresponding entropy terms given by Eq.~\eqref{eq:entropyboundcorrectness}, together with Lemma~\ref{lem:indistproperties} (1), we conclude that the state $\rho^{(\textnormal{out})}_{M_0,M_1,C,M_C}$ of the output systems after the post processing phase satisfies (substituting $p_{\textnormal{succ}} = 1-p_{\textnormal{abort}}$)
    \begin{equation}
        \rho_{M_0,M_1,C,M_C}^{(\textnormal{out})} \approx_{\varepsilon} \frac{p_{\textnormal{succ}}}{2^{(2n+1)}}
        \!\!\!\!\!\!\!\!\!
        \sum\limits_{\substack{m_0,m_1 \in \bs^{n} \\ c \in \bs}}
        \!\!\!\!\!\!\!\!\!
        \ketbra{m_0}_{M_0} \ketbra{m_1}_{M_1}  \ketbra{c}_{C} \ketbra{m_c}_{M_C},
\end{equation}
with 
\begin{equation}
    \varepsilon \leq 2^{-\frac{1}{2}(N_{\textnormal{raw}}-n)}   + 2\varepsilon_{\textnormal{IR}}(k').
\end{equation}

\subsection{Security against dishonest sender}
\label{sec:dishonestalice}

For this scenario we show that, in the case of an honest Bob and Alice running an arbitrary program, the resulting state after the protocol successfully finishes satisfies Eq.~(\ref{eq:uncertain_b}). In other words, independently of what quantum state Alice shares at the beginning of the protocol and which operations she performs on her systems, her final state is independent of the value of $c$. We assume that Alice's laboratory consists of everything outside Bob's. In particular, this means that she controls the environment, which includes the transmission channels. We also assume that Alice is limited to performing efficient computations.

Let $A$ be the system consisting of all of Alice's laboratory after Step (1) of the protocol, that is, $A$ contains her part of the shared system and every other ancillary system she may have access, but does not contain any system from Bob's laboratory, including Bob's part of the system shared in Step (1). During the execution of the protocol, Alice receives external information from Bob exactly three times: the commitment information shared during Step (4), the opening information $\oopen_{I_t}$ for the commitments associated to the test set $I_t$ in Step (6), and the information of the pair of sets $(J_0, J_1) = (I_c, I_{\Bar{c}})$ during Step (9). Let $\CCOM = (\CCOM_i)_{i=1}^{N_0}$ and $\OOPEN = (\OOPEN_i)_{i=1}^{N_0}$ be the respective systems used by Bob to store the information of the strings $\ccom = (\ccom_i)_{i=1}^{N_0}$ and $\oopen = (\oopen_i)_{i=1}^{N_0}$, and let $\textnormal{SEP}$ be the system holding the string separation information $(J_0, J_1)$. We want to show that, by the end of the protocol, the state of the system $A, \CCOM, \OOPEN_J, \textnormal{SEP}, C$ satisfies:
\begin{equation}
    \label{eq:honestreceiver2}
    \rho_{A, \CCOM, \OOPEN_{\bar{I}_t}, \textnormal{SEP}, C} \approx^{(c)} \rho_{A, \CCOM, \OOPEN_{\bar{I}_t}, \textnormal{SEP}} \otimes \mathsf{U}_{C}.
\end{equation}
To guarantee that Alice will not be able to obtain information about the value of $c$ during the string separation phase, it is necessary to show that Alice does not have access to the information of Bob's bases choices $\theta^B_{I_0,I_1}$ from the commitments sent by Bob during Step (4) of the protocol. As shown by Lemma~\ref{lem:commithidingsecurity}, the shared state of the parties after the commitment information is sent is computationally indistinguishable from a state where Alice's information is independent of $\theta^{B}_{\Bar{I}_t}$.
\begin{lemma}
    \label{lem:commithidingsecurity}
    Assuming Bob follows the protocol, for any $J \subseteq I$, the state of the system $A, \CCOM, \OOPEN_J, \Theta^{B}_{\bar{J}}$ after Step (4) satisfies
    \begin{equation}
        \label{eq:commithidingsecurity}
        \rho_{A, \CCOM, \OOPEN_J, \Theta^{B}_{\bar{J}}} \approx^{(c)} \rho_{A, \CCOM, \OOPEN_J} \otimes \mathsf{U}_{\Theta^{B}_{\bar{J}}}.
    \end{equation}
\end{lemma}

At Step (8) of the protocol, Alice sends Bob the system $\Theta^{A}_{\bar{I}_t}$ intended to have the information of her measurement bases. Bob then is able to determine the indices for which $\theta^{A}_{\bar{I}_t}$ and $\theta^{B}_{\bar{I}_t}$ coincide. With this information, he randomly selects sets $I_0, I_1 \in \Bar{I}_t$ of size $N_{\textnormal{raw}}$ for which all indices are associated with matching (for $I_0$) or nonmatching (for $I_1$) bases. Then he computes $(J_0, J_1) = (I_c, I_{\Bar{c}})$, by flipping the order if the pair $(I_0, I_1)$ depending on the value of $c$. Clearly, $(J_0, J_1)$ depend on both $\theta^{B}_{\Bar{I}_t}$ and $c$, but as Lemma~\ref{lem:stringsepsecurity} states, any correlation between $(J_0, J_1)$, $c$, and Alice's information disappears if one does not have access to $\theta^{B}_{\bar{I}_t}$. 
\begin{lemma}
    \label{lem:stringsepsecurity}
    Denote by $A'$ the system representing Alice's laboratory at the start of Step (9). Let $\mathcal{E}^{(I_t)}: \mathcal{D}(\mathcal{H}_{A',\Theta^{A}_{\bar{I}_t},{\Theta^{B}_{\bar{I}_t}}, C}) \rightarrow \mathcal{D}(\mathcal{H}_{A',\Theta^{A}_{\bar{I}_t}, {\Theta^{B}_{\bar{I}_t}}, C, \textnormal{SEP}})$ be the quantum operation used by Bob to compute the string separation information $(J_0, J_1)$ during Step (9) of the protocol. The resulting state after applying $\mathcal{E}^{(I_t)}$ to a product state of the form
    \begin{equation}
        \label{eq:preseparation}
        \mathcal{E}^{(I_t)}(\rho_{A',\Theta^{A}_{\bar{I}_t}} \otimes \mathsf{U}_{\Theta^{B}_{\bar{I}_t}} \otimes \mathsf{U}_{C}) = \sigma_{A',\Theta^{A}_{\bar{I}_t}, {\Theta^{B}_{\bar{I}_t}}, C, \textnormal{SEP}}
    \end{equation}
    satisfies
    \begin{equation}
    \label{eq:stringindependence1}
        \Tr_{\Theta^{A}_{\bar{I}_t},\Theta^{B}_{\bar{I}_t}} \big[ \sigma_{A',\Theta^{A}_{\bar{I}_t}, {\Theta^{B}_{\bar{I}_t}}, C, \textnormal{SEP}} \big] =  \sigma_{A'} \otimes \sigma_{\textnormal{SEP}} \otimes \mathsf{U}_{C}.
    \end{equation}
\end{lemma}

A proof of both Lemmas~\ref{lem:commithidingsecurity} and~\ref{lem:stringsepsecurity} can be found in Appendix~\ref{sec:commithidingproof}. By setting $J = I_t$, Lemma~\ref{lem:commithidingsecurity} guarantees that Alice's system's state after the opening information has been sent is computationally indistinguishable from one that is completely uncorrelated with Bob's measurement basis information in $\Bar{I}_t$. 

Additionally, by recalling that the value of $c$ is sampled independently of any of the considered systems, we know that the state $\rho_{{A',\Theta^{A}_{\bar{I}_t}, \Theta^{B}_{\bar{I}_t}}, C}$ before $(J_0, J_1)$ is computed has the required product form and, from Lemma~\ref{lem:stringsepsecurity}, we conclude that the state of all of Alice's system at this point is computationally indistinguishable from a state uncorrelated with $C$. Let $\mathcal{E}$ be the operation Alice performs in her system from here to the end the protocol. By using Lemma~\ref{lem:cindistproperties} (4) and grouping all of Alice's systems into $S$, we obtain the desired result: 
\begin{equation}
    \rho_{S,C}  \approx^{(c)}  \rho_{S} \otimes \mathsf{U}_{C}.
\end{equation}

\subsection{Security against dishonest receiver}
We consider now the scenario in which Alice runs the protocol honestly and Bob runs an arbitrary program. For this analysis, note that Alice trusts her quantum state preparation and detection. We want to show that the state after finishing the protocol successfully satisfies Eq.~(\ref{eq:uncertain_m}). This means that the state at the end of the protocol can be described as a mixture of states where Bob's system is uncorrelated with at least one of the two strings outputted by Alice. Similarly to the dishonest sender's case, we assume that Bob's laboratory consists of everything outside Alice's, which means that he controls the communication channels and the environment. However, we do not assume that Bob is restricted to efficient computations. 

The values of Alice's output strings depend on several quantities: Alice's measurement outcomes, the choice of the $I_t, J_0, J_1$ subsets, and the choice of hashing function $f$ during the post-processing phase of the protocol. From all of these, the only ones that are not made explicitly public during the protocol's execution are Alice's measurement outcomes. Instead, partial information of these outcomes is revealed at different steps of the protocol. Let $x^{A}_{J_0}, x^{A}_{J_1}$ be the sub-strings of measurement outcomes used to compute Alice's outputs $m_0, m_1$, respectively, and let $R$ denote Bob's system at the end of the protocol (which includes all the systems that Alice sent during the execution of the protocol). In order to prove security we need to show that the joint state of the system $X^{A}_{J_0}, X^{A}_{J_1}, R$ can be written as a mixture of states $\rho^{b}$ (with $b\in \bs$) such that the conditional min-entropy $ H_{\textnormal{min}}^{\varepsilon} (X^{A}_{J_b}|R)_{\rho^{b}}$ is high enough, so that we can use the leftover hash Lemma~\ref{lem:qlhl} to guarantee that the outcome of the universal hashing $m_b = f(x^{A}_{J_b})$ is uncorrelated with $R$.

At the start of the protocol the parties share a completely correlated entangled system. If the parties make measurements as intended, their outcomes will be only partially correlated, but if Bob was able to postpone his measurement until after Alice's reveals her measurement bases, Bob could potentially obtain the whole information of $x^{A}$ by measuring in the appropriate basis on his system. To prevent this, Bob is required to commit his measurement bases and results to Alice before knowing which set is going to be tested. Then a statistical test is performed in Step (7) to estimate the correlation of Alice's measurement outcomes with with the ones that Bob committed. As Lemma~\ref{lem:postmeasuremententropy} states, any state passing the aforementioned test is such that, regardless of how Bob defines the sets $(J_0, J_1)$ during the string separation phase, there is a minimum of uncertainty that he has with respect to Alice's measurement outcomes. Recall that, when Alice is honest, the overall state of the protocol before Step (8) will be a partially classical state, which could be written as a mixture over all of Alice's classical information. Let $\Vec{\tau} = (x^{A}_{I_t},\theta^{A},r,\ccom,I_t,I_s, \oopen_{I_t})$ denote the \textit{transcript of the protocol} up to Step (8), and let $\rho_{X^{A} B}(\Vec{\tau}, J_0, J_1)$ be the joint state of Alice's measurement outcomes and Bob's laboratory conditioned to $\Vec{\tau}, J_0, J_1$. 

\begin{lemma}
    \label{lem:postmeasuremententropy}
      Assuming Alice follows the protocol, let $T, \textnormal{SEP}, B$ denote the systems of the protocols transcript, the strings $J_0, J_1$, and Bob's laboratory at the end of Step (9) of the protocol, and let $\rho_{T, \textnormal{SEP}, X^{A},B}$ be the state of the joint system at that point. There exists a state $\tilde{\rho}_{T, \textnormal{SEP}, X^{A}, B}$, which is classical in $T$ and $\textnormal{SEP}$ such as:       
      \begin{enumerate}
        \item The conditioned states $\tilde{\rho}_{X^{A}, B}(\Vec{\tau}, J_0, J_1)$ satisfy:
        \begin{align}
                H_{\textnormal{min}} (X^{A}_{J_0}|X^{A}_{J_1} B)_{\tilde{\rho}(\Vec{\tau}, J_0, J_1)} + H_{\textnormal{min}} (X^{A}_{J_1}|X^{A}_{J_0} B)_{\tilde{\rho}(\Vec{\tau}, J_0, J_1)} \geq 2 N_{\text{raw}} \left( \frac{1}{2} - \delta_2 - h\left( \frac{p_{\text{max}} + \delta_1}{\frac{1}{2} - \delta_2}\right) \right),
        \end{align}
        \item $\rho_{T, \textnormal{SEP}, X^{A}, B} \approx_{\varepsilon} \tilde{\rho}_{T, \textnormal{SEP}, X^{A}, B}$, with 
        \begin{equation}
        \label{eq:epsilonsecuritydishonestb}
        \varepsilon = \left(2 (e^{-\frac{1}{2} \alpha (1-\alpha)^{2} N_0 \delta^{2}_1} +  e^{-\frac{1}{2} (\frac{1}{2} -\delta_2)\alpha N_0 \delta^{2}_1}) \right)^{\frac{1}{2}} + e^{-D_{KL}(\frac{1}{2} - \delta_2|\frac{1}{2})(1-\alpha) N_0} + \varepsilon_{\textnormal{bind}}(k),
        \end{equation}
        where $h(\cdot)$ and $D_{KL}(\cdot | \cdot)$ denote the binary entropy and the binary relative entropy functions, respectively, and $\varepsilon_{\textnormal{bind}}(k)$ is a negligible function given by the binding property of the commitment scheme.
    \end{enumerate}
\end{lemma}

To reach the desired result, we will first show that a state $\tilde{\rho}_{T, \textnormal{SEP}, X^{A}, B}$ satisfying Lemma~\ref{lem:postmeasuremententropy} (1) also satisfies a tighter version Lemma~\ref{lem:dishonestreceiver}, and then use Lemma~\ref{lem:postmeasuremententropy} (2) to attain the bound for the real protocol's outcome. Since $\tilde{\rho}_{T, \textnormal{SEP}, X^{A}, B}$ is classical in both $T$ and $\textnormal{SEP}$ we can write the state of the joint system of Alice's measurement outcomes and Bob's laboratory as a mixture over all the possible transcripts at that point, that is:
\begin{equation}
    \label{eq:bigjointmixture}
    \tilde{\rho}_{X^{A} B} = \sum_{\substack{\Vec{\tau} \\ {J_0,J_1}}} P(\Vec{\tau}, J_0,J_1)  \tilde{\rho}_{X^{A}B}(\Vec{\tau}, J_0,J_1),
\end{equation}
where $P(\Vec{\tau}, J_0,J_1)$ defines a probability distribution which is dependent on Bob's behavior during the previous steps. We can now separate the $\tilde{\rho}_{X^{A}B}(\Vec{\tau}, J_0,J_1)$ in two categories depending on which of the $x^{A}_{J_0}, x^{A}_{J_1}$ is the least correlated with Bob's system. Consider the function $b(\Vec{\tau}, J_0,J_1)$ to be equal to $0$ if $H_{\textnormal{min}}^{\varepsilon} (X^{A}_{J_0}|X^{A}_{J_1} B)_{\rho(\Vec{\tau}, J_0,J_1)} \geq H_{\textnormal{min}}^{\varepsilon} (X^{A}_{J_1}|X^{A}_{J_0} B)_{\rho(\Vec{\tau}, J_0,J_1)}$, and equal to $1$ otherwise. By regrouping the terms from~\eqref{eq:bigjointmixture} for which the value of $b$ is the same, we can rewrite the joint state as:
\begin{equation}
    \label{eq:smalljointmixture}
    \tilde{\rho}_{X^{A} B} = \sum_{b \in \bs} P_b \tilde{\rho}^{b}_{X^{A}B},
\end{equation}
where, from Lemma~\ref{lem:postmeasuremententropy} and recalling that, as Lemma~\ref{lem:entropyproperties} (5) states, the min-entropy of a mixture is lower bounded by that of the term with the least min-entropy, we know that 
  
\begin{align}
    H_{\textnormal{min}} (X^{A}_{J_b}|X^{A}_{J_{\Bar{b}}} B)_{\tilde{\rho}^{b}} \geq N_{\text{raw}} \left( \frac{1}{2} - \frac{2\delta_2}{1-2\delta_2} - h\left( \frac{p_{\text{max}} + \delta_1}{\frac{1}{2} - \delta_2}\right) \right).
\end{align}
At Step (10), Alice shares with Bob the syndromes $S_0 = \syn(x^{A}_{J_{0}})$ and $ S_1 =\syn(x^{A}_{J_{1}})$. Since these syndromes are completely determined by the respective sub-strings $x^{A}_{J_{i}}$, we know that 
\begin{align}
    \label{eq:totalleak1}
    H_{\textnormal{min}} (X^{A}_{J_b}|S_{b} S_{\Bar{b}} B) &\geq H_{\textnormal{min}} (X^{A}_{J_b}|S_{b} X^{A}_{J_{\Bar{b}}} B) \nonumber \\
    &\geq H_{\textnormal{min}}(X^{A}_{J_b}|X^{A}_{J_{\Bar{b}}} B) - H_{\max}(S_{b}) \nonumber \\
    &\geq N_{\text{raw}} \left( \frac{1}{2} - \frac{2\delta_2}{1-2\delta_2} - h\left( \frac{p_{\text{max}} + \delta_1}{\frac{1}{2} - \delta_2}\right) - f \cdot h(p_{\text{max}} + \delta_1)\right),
\end{align}
where the second inequality follows from Lemma~\ref{lem:entropyproperties} (3) and (4), and the max entropy term $H_{\max}(S_{b})$ is upper bounded by the size in bits of the syndrome $\ell =  N_{\text{raw}}(f \cdot h(p_{\text{max}} + \delta_1))$. Now we can apply Equation~\eqref{eq:totalleak1} to Lemma~\ref{lem:qlhl}, which states that, for the outcomes $M_0,M_1$ of the universal hashing by Alice in Step (12) and the Bob's final system $R$ it holds that
\begin{equation}
    \label{eq:postpadishonestbob}
    \tilde{\rho}_{R,M_0,M_1}^b \approx_{\varepsilon'} \tilde{\rho}_{R,M_{\bar{b}}}^b \otimes  \mathsf{U}_{M_{b}},
\end{equation}
with
\begin{equation}
    \varepsilon' = \frac{1}{2} \cdot 2^{\frac{1}{2}\left(n - N_{\text{raw}} \left( \frac{1}{2} - \frac{2\delta_2}{1-2\delta_2} - h\left( \frac{p_{\text{max}} + \delta_1}{\frac{1}{2} - \delta_2}\right) - f \cdot h(p_{\text{max}} + \delta_1) \right) \right)}.
\end{equation}
Finally, by applying Lemma~\ref{lem:indistproperties} (3) and (4) to Equations~\eqref{eq:smalljointmixture} and ~\eqref{eq:postpadishonestbob}, and then Lemma~\ref{lem:indistproperties} (1) to Eq.~\eqref{eq:epsilonsecuritydishonestb} we get the desired result:
\begin{equation}
    \rho_{R,M_0,M_1} \approx_{\varepsilon} \sum_{b} P_b \, \rho_{R,M_{\bar{b}}}^b \otimes  \mathsf{U}_{M_{b}},
\end{equation}
with
\begin{align}
    \varepsilon &= \sqrt{2} \left(e^{-\frac{1}{2} (1-\alpha)^{2} N_{\textnormal{test}} \delta^{2}_1} +  e^{-\frac{1}{2} N_{\textnormal{check}} \delta^{2}_1} \right)^{\frac{1}{2}} + e^{-D_{KL}(\frac{1}{2} - \delta_2|\frac{1}{2})(1-\alpha) N_0}  + \varepsilon_{\textnormal{bind}}(k) \\ \nonumber
    &\quad  + \frac{1}{2} \cdot 2^{\frac{1}{2}\left(n - N_{\text{raw}} \left( \frac{1}{2} - \frac{2\delta_2}{1-2\delta_2} - h\left( \frac{p_{\text{max}} + \delta_1}{\frac{1}{2} - \delta_2}\right) - f \cdot h(p_{\text{max}} + \delta_1) \right) \right)}.
\end{align}

\subsection{Composability considerations}
\label{sec:composabiltiy}

Since OT protocols are mainly used as a subroutine of larger applications it is important to understand the composability properties of $\pi_{\textnormal{ROT}}$. In general, this is done through simulation-based composability frameworks. As mentioned in Section~\ref{sec:introduction}, this protocol is based on the BBSC construction, which has been proven secure in the quantum Universal Composability (UC) framework by Unruh~\cite{unruh10} assuming access to an ideal commitment functionality. This means that we can understand the composability properties of $\pi_{\textnormal{ROT}}$ by understanding the respective properties of the underlying weakly-interactive commitment protocol.

It is well known that UC commitments are impossible to realize in the plain model~\cite{canetti01,Maurer11}. Because of this, protocols are often analyzed within a hybrid model, where the parties have access to some base external functionality. We show in Appendix~\ref{sec:composabiltiy} that there exists a family of commitment schemes that are both weakly-interactive and UC-secure in the \textit{classical access} Random Oracle Model (ROM)~\cite{yamakawa21}. This, in tandem with the aforementioned reduction of OT to commitments, results in the following theorem:

\begin{theorem}
    \label{thm:composability}
    There exists a family of weakly-interactive commitment schemes in relation to which $\pi_{\textnormal{ROT}}$ is UC-secure in the classical access ROM.
\end{theorem}

In relation to Theorem~\ref{thm:composability}, we want to emphasize that, even though limiting the access to the random oracle to be classical may seem at first strong in the context of a quantum protocol (where the parties are required access to some quantum capabilities), it has little impact in the resulting security of larger MPC protocols for which the security is analyzed in the classical setting. 

Finally, we would like to stress the merits of Def.~\ref{def:ROT} by itself. In particular, this definition was studied in~\cite{schaffner10} and~\cite{konig12} and stated to ensure security when the protocol is executed sequentially. Furthermore, the indistinguishability properties stated in Def.~\ref{def:ROT} provide a very strong security guarantee and, because the protocol does not have external inputs and the indistinguishability relations include arbitrary external systems, these properties will still hold in any environment, which makes it relatively straightforward to analyze as part of bigger applications.

\section{Experimental Implementation}
\label{sec:experiment}
\subsection{Description of the Setup}
\begin{figure}
    \centering
    \includegraphics[width=0.6\linewidth]{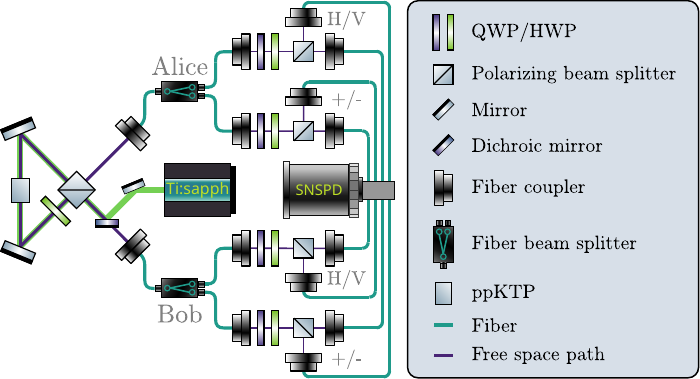}
    \caption{\textbf{Experimental setup.} Polarization-entangled photon pairs are created using spontaneous parametric down conversion. Alice's and Bob's photons are individually fiber coupled and each sent to 50/50 fiber beam splitters, which probabilistically route them to free-space polarization projection stages - one projecting onto the linear, and one onto the diagonal basis each for Bob and Alice.}
    \label{fig:experimental_setup}
\end{figure}
A schematic representation of the experimental setup can be seen in Fig.~\ref{fig:experimental_setup}. Spontaneous parametric down conversion (SPDC), attributed to Alice, is used to create polarization entangled photon pairs in the state $\ket{\Psi^+}=\frac{1}{\sqrt{2}} \left( \ket{HH}+\ket{VV }\right)$, which are coupled into optical fiber. One photon is sent through a 50/50 fiber beam splitter, probabilistically routing it to one of two polarization projection stages. There, a quarter-wave plate (QWP), a half-wave plate (HWP) and a polarizing beam splitter (PBS) are used to project the photons state onto the linear ($H/V$) or diagonal ($+/-$) basis, respectively. All photons are sent to superconducting nanowire single photon detectors (SNSPDs) and their arrival time is recorded using a time tagging module (TTM).
The second photon of the state  $\ket{\Psi^+}$, attributed to Bob, travels through an equivalent probabilistic projection setup.

Entangled photon pairs are generated using collinear type-II SPDC in a periodically poled $\mathrm{KTiOPO_4}$-crystal with a poling period of \SI{46.2}{\micro\meter} inside of a sagnac interferometer. The pump light is produced by a pulsed Ti:Sapphire laser (Coherent Mira 900HP) with a pulse width of \SI{2.93}{\pico\second} and a central wavelength of $\lambda_p = \SI{773}{\nano\meter}$, creating degenerate single-photon pairs at $\lambda_s = \lambda_i = \SI{1546}{\nano\meter}$. The laser's inherent pulse repetition rate of \SI{76}{\mega\hertz} is doubled twice to \SI{304}{\mega\hertz} using a passive temporal multiplexing scheme \cite{Greganti2018}. More precisely, for
$n$ simultaneously emitted pairs and
$k$ multiplexing stages, each doubling the repetition rate, higher-order pair production events are attenuated by a factor of $1 / (2^k)^{n-1} $. In our experiment, $k=2$, so this scheme reduces the probability of emitting a double pair ($n=2$) by a factor of $4$ compared to a source relying on the pump's inherent repetition rate, while the single-pair emission probability remains constant. Finally, about \SI{100}{\meter} of single mode fiber separate the experimental setup from the \SI{1}{\kelvin} cryostat housing the SNSPDs with a detection efficiency of around \SI{95}{\percent} and a dark-count rate of around \SI{300}{\hertz}.

We note that our entanglement-based implementation presents two main technological advantages over prepare-and-measure configurations:
\begin{itemize}
\item It circumvents the need for a certified quantum random number generator or for classical pseudorandomness that may compromise the security of the quantum phase: instead of feeding random (or pseudorandom) sequences into the active polarization modulator of a prepare-and-measure scheme, the choice of BB84 state is performed in a passive and uniformly random way by the beamsplitters present in both Alice and Bob's measurement setups (also known as "remote state preparation"). 
\item In free space, it avoids the need for active polarization modulation, which imposes a strict upper limit on the protocol's repetition rate governed by the bandwidth of the Pockels Cell and its high-voltage amplifier, typically achieving a few hundred kHz to a few hundred MHz \cite{HJW:Arx25}. By generating entangled photons that are passively projected onto one of the four BB84 states instead, our OT rate is not limited by any active prepare-and-measure encoding routine, but only by our picosecond-pulsed pump rate of around $300$ MHz. With other SPDC sources reaching the GHz \cite{JSM:SciRep14} to tens of GHz regimes \cite{WTF:OptExp20}, our passive state preparation routine can perform even better.
\end{itemize}

\subsection{Practical protocol}\label{sec:practical_protocol}

The protocol is identical to that from $\pi_{\textnormal{ROT}}$ as described in Section~\ref{sec:protocol}, with the following amendments:
\begin{itemize}

    \item The parties agree on an additional parameter $p_{\textnormal{multi}}$ -- the accepted ratio of multi-photon events. 
    \item During the quantum phase of the protocol, Alice may observe detection patterns that are incompatible with the emission of a single photon pair. Instead of sharing $N_0$ states in Step (1), she continues sharing states until, after agreeing on coincidence time-tags with Bob, the parties obtain $N_0$ coincidences associated with single-photon events on Alice's side. Let $N_{\textnormal{tot}}$ be the number of coincidences obtained at this point and $N_{\textnormal{multi}} = N_{\textnormal{tot}} - N_0$. Alice computes the value
    \begin{equation}
        \label{eq:pmultiphoton}
        p'_{\textnormal{multi}} = \frac{N_{\textnormal{multi}} }{3 N_{\textnormal{tot}}},
    \end{equation}
    and aborts the protocol if $p'_{\textnormal{multi}} \geq p_{\textnormal{multi}}$.
    \item Similarly to Alice, Bob may also observe multi-click patterns. While reporting its detection events he uses the following rules:
    \begin{enumerate}[label=(\alph*)]
        \item $1$ click: assign the correct measured bit value and report a successful round
        \item $2$ clicks from the same basis: assign a random bit value to the measurement result and report a successful round
        \item any other click pattern: report an unsuccessful round
    \end{enumerate}
\end{itemize}

\subsection{Practical security}

Any photonic implementation of quantum cryptography presents experimental imperfections, which can be exploited by dishonest parties to enhance their cheating probability and violate ideal security assumptions. Important examples of such imperfections include multiphoton noise, lossy/noisy quantum channels, non-unit detection efficiency and detector dark counts. 

\textit{Dishonest sender.} In our experiment, threshold detectors cannot resolve the incident photon number, and unexpected click patterns can occur. For example, several of the four detectors may simultaneously click for a given round, which leads to an inconclusive measurement outcome that has to be back-reported by the honest receiver. This in turn allows a dishonest sender to gain a significant amount of information about the receiver's measurement basis choice. Adopting the reporting strategy presented above makes the protocol secure against this type of attack. For a complete analysis of both the attack and its countermeasures, see \cite{bozzio21}.

\textit{Dishonest receiver.} Due to Poisson statistics in the SPDC process, emission of double pairs can occur for a given round. When the two photons kept by the sender are projected onto the same state (i.e. only a single click is recorded in the four detectors), the two photons sent to Bob have the same polarization. In this case, a dishonest receiver can split the two photons and measure one in each basis. Assuming $4$ detectors with equal efficiencies (which can be guaranteed in practice by appropriate attenuation the higher efficiency ones), and using the fact that for an SPDC source, whenever multiple pairs are produced, there is no correlation among them , we know that the number of undetected multi-photon events is approximately $\frac{1}{3}$ of the number of detected ones. We can then estimate the probability $p'_{\textnormal{multi}}$ of an accepted coincidence to be associated with a multi-photon event with Eq.~\eqref{eq:pmultiphoton}.

Note that the statistical check performed by Alice in the second step of the amended protocol (Section \ref{sec:practical_protocol}) ensures security under the assumption that there is no coherence in the photon-number basis. This is the case in our implementation, since SPDC produces states of the form $\sum_{n=0}^{\infty}\sqrt{c_n}\ket{n}_1\ket{n}_2$ in the number basis $\{\ket{n}\}$ \cite{Braczyk2010}, leaving the individual subsystems in incoherent mixtures of the form $\sum_{n=0}^{\infty}c_n\ket{n}\bra{n}$.  

To account for the leakage caused by undetected multi-photon emissions to our OT rate expression, we effectively grant Bob an amount of information about Alice's measurement outcomes equal to the number of indices in $I_{\bar{I}_t}$ associated to multi-photon events, upper bounded by $p_{\textnormal{multi}} (1- \alpha) N_0 $ for large $N_0$. Subtracting this leak to the total entropy expression in Eq.~\eqref{eq:totalleak1} leads to a version of Lemma~\ref{lem:dishonestreceiver} for security against dishonest receiver corrected for the experimental implementation, which differs from the theoretical version by replacing Eq~\eqref{eq:epsilondisonestbob} with
\begin{align}
    \label{eq:epsilondisonestbobexp}
    \varepsilon'_{\textnormal{exp}} &= \sqrt{2} \left(e^{-\frac{1}{2} (1-\alpha)^{2} N_{\textnormal{test}} \delta^{2}_1} +  e^{-\frac{1}{2} N_{\textnormal{check}} \delta^{2}_1} \right)^{\frac{1}{2}} + e^{-D_{KL}(\frac{1}{2} - \delta_2|\frac{1}{2})(1-\alpha) N_0}  + \varepsilon_{\textnormal{bind}}(k) \\ \nonumber
    &\quad \quad  + \frac{1}{2}\cdot 2^{\frac{1}{2}\left(n  - N_{\textnormal{raw}} \left( \frac{1}{2} - \frac{2\delta_2}{1-2\delta_2} - h\left( \frac{p_{\textnormal{max}} + \delta_1}{\frac{1}{2} - \delta_2}\right) - f \cdot h(p_{\textnormal{max}}+\delta_1) -\frac{p_{\textnormal{multi}}}{\frac{1}{2}-\delta_2} \right) \right)},
\end{align}    

\section{Discussion}

Using Naor's protocol~\cite{naor91} in conjunction with a linear time OWF (such as a hash function fron the SHA3 or BLAKE family), it is possible to implement the required 2-bit commitment in linear time in $k$. On the other hand, using an LDPC code with soft-decision decoding and hash based verification, one can implement an IR scheme which is linear in both the block size $N_{\textnormal{raw}}$ (and therefore $N_0$) and $k'$.  Finally, by taking the universal hash set $\textnormal{\textbf{F}}$ to be the set of Toeplitz matrices of size $N_{\textnormal{raw}} \times n$, and using the FFT algorithm for matrix-vector multiplication, the computation of the output strings can be done in time $O(N_{\textnormal{raw}} \log(N_{\textnormal{raw}}))$. Considering that the protocol requires $N_0$ commitments and all the remaining computations of random subsets and checks can be done in linear time in $N_0$, the total protocol running time is $O(N_{0} (k + k' + \log(N_{0}))$.

Regarding the practicality of implementing $\pi_{\textnormal{QROT}}$, the protocol is designed to be compatible with BB84-based QKD setups, both from the physical layer up to the post-processing, only requiring the addition of the commitment scheme. The most important difference to note is that $\pi_{\textnormal{QROT}}$ has significantly lower tolerance for Qubit Error Rate (QBER). While most common QKD protocols can produce keys through QBERs above $10\%$, this protocol is limited to a maximum of $2.8\%$. This comparatively reduces the distances at which the protocol can be successful. However, it is important to note that, as opposed to key distribution between trusting parties, there are legitimate use-cases for OT at short range. While being in proximity to each other can help two trusting parties isolate themselves from a third party eavesdropper, mistrusting parties do not gain anything (security wise) from being in the same place while attempting to do MPC, making the protocol useful regardless of the distance between the users.

Comparisons between classical and quantum protocols can be difficult because physical/technological assumptions, such as access to quantum communication or noisy quantum storage, do not straightforwardly compare with computational hardness assumptions. Furthermore, there is no natural way of quantitatively comparing statistical versus computational security. We can, however, contrast the (dis-)advantages of using a computationally-secure quantum OT protocol as compared to both fully classical computationally-secure protocols, as well as statistically-secure quantum ones.

Classical OT protocols based on asymmetric cryptography comprise the overwhelming majority of current real-world implementations of OT. The obvious main advantage of quantum OT is the weaker computational hardness assumption (OWF vs asymmetric cryptography), while the main advantage of current post-quantum classical OT implementations is speed. As shown in Fig.~\ref{fig:expperformance}, the presented experimental setup is able to produce up to $~0.10$ OT/s, which pales in comparison to contemporary classical protocols, such as~\cite{mansy2019, mi2018, mi2019, branco2021}, that can achieve upwards of $10^5$ OT/s (not including latency between parties) with current off-the-shelf hardware (for more details, see~\cite{branco2021}). This difference can be mitigated by the use of OT extension algorithms, as the difference in speed would only matter during the generation of the base OTs. Note that in this case one should use a OT extension that matches the computational assumption of this work, such as~\cite{costa21}.

Quantum protocols, both discrete variable (DV)~\cite{erven2014} and continuous variable (CV)~\cite{furrer2018}, have been shown to achieve statistically-secure OT in the Quantum Noisy-Storage model (QNS). Their experimental implementations show comparable values of quantum communication cost in terms of shared signals: $10^8$ (no memory encoding assumption), and $10^5$ (Gaussian encoding) for CV, and $10^7$ for DV. As shown in Fig.~\ref{fig:MinNvsEPlot}, our protocol requires $10^6$ quantum signals when matching their security ($\varepsilon = 10^{-7}$), which improves upon the alternatives when no additional assumption on the memory encoding of the adversary is made. Less straightforward to compare is the strength of the assumptions of noisy storage and OWFs. We note that the existence of OWFs is an assumption that permeates modern cryptography, from block cipher encryption and message authentication up to public-key cryptography protocols~\cite{Goldreich06}, which makes $\pi_{\textnormal{QROT}}$ more suited to be introduced in current cipher suites than protocols with alternative assumptions. In particular, as noted above, OWFs are required for OT extension algorithms. A summary of comparisons between the different approaches can be found in Fig.\ref{tab:otcomarison}.

\begin{table*}[htbp]
\centering
\renewcommand{\arraystretch}{1.8}
\begin{tabular}{p{3cm}|>{\centering\arraybackslash}p{3cm}|>{\centering\arraybackslash}p{3cm}|>{\centering\arraybackslash}p{3cm}|>{\centering\arraybackslash}p{3cm}}
\hline
\textbf{Protocol} & \textbf{Type} & \textbf{Assumption} & \textbf{Quantum Cost} & \textbf{Security} \\
\hline\hline

This work & \begin{tabular}{>{\centering\arraybackslash}p{3cm}}Quantum\\ Discrete Variable\end{tabular} & OWF & $\mathcal{O}(N)$ & \begin{tabular}{>{\centering\arraybackslash}p{3cm}}Indistinguishability\\ UC ROM\end{tabular} \\
\hline

GLSV21~\cite{grilo21}* & \begin{tabular}{>{\centering\arraybackslash}p{3cm}}Quantum\\ Discrete Variable\end{tabular} & OWF & $Poly(N)$ & Stand-alone plain model\\
\hline

S10~\cite{schaffner10,erven2014} & \begin{tabular}{>{\centering\arraybackslash}p{3cm}}Quantum\\ Discrete Variable\end{tabular} & QNS & $\mathcal{O}(N)$ & Indistinguishability \\
\hline

FGSPSW18~\cite{furrer2018} & \begin{tabular}{>{\centering\arraybackslash}p{3cm}}Quantum\\ Continuous Variable\end{tabular} & QNS & $\mathcal{O}(N)$ & Indistinguishability \\
\hline

MR19~\cite{mansy2019} & Classical & DDH & -- & Stand-alone ROM \\
\hline

BFGMMS21~\cite{branco2021} & Classical & RLWE & -- & UC ROM \\
\hline

P16~\cite{pitalua16}* & \begin{tabular}{>{\centering\arraybackslash}p{3cm}}Quantum/Relativistic\\ Discrete Variable\end{tabular} & SLS& $\mathcal{O}(N)$ & Other \\
\hline

\end{tabular}
\caption{Comparison of our work with other approaches for OT. $N$ denotes the respective security parameter. Acronyms for assumptions are as follows: OWF - One Way Functions; QNS - Quantum Noisy Storage; DDH - Decisional Diffie-Helmann; RLWE - Ring Learning With Errors - SLS - Space-Like Separation enforced. Protocols marked with * do not have a reference experimental implementation at the time of writing.}
\label{tab:otcomarison}
\end{table*}


Regarding potential improvements and further work, we can identify two main directions to build upon this work: performance and security. Regarding performance, we note that dominant term in the expression for $\varepsilon_{max}$ is the one associated with the significance of the parameter estimation (the first term in Eq.~\ref{eq:epsilondisonestbob}). This translates into the relatively large values of $N_0$ needed to achieve adequate security, which was the bottleneck in the performance of our implementation. One way to reduce the number of signals needed per OT is to modify the protocol to perform many concurrent ROTs in a single run. This would mean performing one single estimation, albeit of a larger sample, that would work for many OTs in such a way that the required number of signals per ROT is decreased. On the topic of increasing security two main directions come to mind. First, we can consider the constructions of \textit{collapsing} hash functions proposed in~\cite{unruh16, zhandry22} to implement statistically hiding, computationally collapse binding commitments, which in turn allow for OT protocols that feature forward security (the OT remains secure even if the underlying hash function can be attacked after the commit/open phase of the protocol). The second direction would be a deeper exploration of the composable security of the protocol in the ROM. This can come from generalizing Theorem~\ref{thm:composability} for any weakly-interactive commitments (currently the proof applies only to the LRV25 construction), or applying the techniques developed in~\cite{agarwal23} to prove UC security of commitments in the quantum ROM to remove the adversary's limitation of classical access to the oracle. From the practical implementation perspective, it seems natural to integrate quantum OT into both QKD setups for a unified physical layer capable of providing secure communication and computation powered by OT extension and MPC algorithms, bringing the benefits of quantum OT closer to real world usage. 


\section{Acknowledgements}

M.L., P.M., and N.P. acknowledge Fundação para a Ciência e Tecnologia (FCT), Instituto de Telecomunicações Research Unit, ref. UID/50008/2020, and PEst-OE/EEI/LA0008/2013, as well as FCT projects QuantumPrime reference PTDC/EEI-TEL/8017/2020. M.L. acknowledges PhD scholarship PD/BD/114334/2016. N.P. acknowledges the FCT Estímulo ao Emprego Científico grant no. CEECIND/04594/2017/CP1393/CT000. P.S., M.B. and P.W. acknowledge funding from the European Union’s Horizon Europe research and innovation program under Grant Agreement No. 101114043 (QSNP), along with the Austrian Science Fund FWF 42 through [F7113] (BeyondC) and the AFOSR via FA9550-21-1-0355 (QTRUST). D.E. was supported by the JST Moonshot R\&D program under Grant JPMJMS226C.
M.G. acknowledges FCT Portugal financing refs. UIDB/50021/2020 and UIDP/50021/2020 (resp. DOI 10.54499/UIDB/50021/2020 and 10.54499/UIDP/50021/2020).

\begin{appendices}

\section{Preliminaries}
\label{sec:preliminaries}

\subsection{Quantum computational efficiency and distinguishability}

We model the quantum capabilities of parties through programs running on quantum computers, for which we adopt a model based on deterministic-control quantum Turing Machines~\cite{mateus17}. For the purposes of the following definitions, a quantum computer is a device that has a classical interface and a quantum part, which contains the quantum memory registers available to the party. The classical interface has the capabilities of a classical computer augmented with the ability to perform a predefined universal set of quantum operators on the quantum memory registers and perform measurements in the canonical (computational) basis. Given a specified type of quantum computer, a quantum program is a classical description of a set of instructions to be run by the computer, including the quantum operations and measurements to be executed in the quantum part, as well as any classical computation. Quantum programs can be compared with probabilistic classical programs as they both have natural numbers as inputs/outputs. When a quantum computer runs the program $T$ with input $x \in \mathbb{N}$, we assume that the quantum part of the computer starts with some predefined initial state, performs a sequence of operations on its quantum registers, and upon halting, it outputs $T(x)\in \mathbb{N}$ on its classical interface by reading the appropriate registers associated with the program's output. Each execution of a quantum program is then associated to a quantum operation, which is the result of all the operations performed on the quantum part during the execution of the program.

\begin{definition} (Computational efficiency) \\
\label{def:efficient}
Let $T$ be a quantum program. We say that $T$ is computationally efficient (or polynomial-time) if there exists a polynomial $P$ such that the running time of $T(x)$ is $O(P(x))$. 
\end{definition}

\begin{definition} (Distinguishing Advantage) \\
\label{def:Dadvantage}
Let $X_{1}, X_{2}$ be two random variables with values in $\mathbb{N}$. For any quantum program $T$, the {\em distinguishing advantage} of $X_{1}, X_{2}$ using $T$ is defined as
\begin{equation}
    \textnormal{Adv}_{T}(X_{1}, X_{2}) =\big| \Pr[T(X_1)=1] - \Pr[T(X_2)=1] \big|,
\end{equation}
Analogously, let $\rho_{1}, \rho_{2} \in \mathcal{D}(\mathcal{H})$. For any quantum program $T$, the {\em distinguishing advantage} of $\rho_{1}, \rho_{2}$ using $T$ is defined as
\begin{equation}
    \textnormal{Adv}_{T}(\rho_{1}, \rho_{2}) =\big| \Pr[T^{(\rho_{1})}=1] - \Pr[T^{(\rho_{2})}=1] \big|,
\end{equation}
where $T^{(\rho)}$ denotes the classical output of the program starting with the quantum state $\rho$ and zero classical input.
\end{definition}

\begin{definition} (Indistinguishability - Finite) \\
\label{def:distinguishableQ}
Let $\rho_1, \rho_2 \in \mathcal{D}(\mathcal{H})$ and $\varepsilon \geq 0$. We say that $\rho_{1}$ and $\rho_{2}$ are { $\varepsilon$-indistinguishable}, denoted by $\rho_{1} \approx_{\varepsilon} \rho_{2}$, whenever
\begin{equation}
\!\!\!    \textnormal{Adv}_{T} \big(\rho_{1}, \rho_{2} \big) \leq \varepsilon,  \,\, \textnormal{for all quantum programs } T.
\end{equation}
$\varepsilon$-indistinguishability for random variables is defined analogously.
\end{definition}

As the following proposition states, to show that two states are $\varepsilon$-indistinguishable, it is enough to upper bound their trace distance $D$. (for more detail on the relationship of these quantities, see~\cite{helstrom76, fuchs99}).
\begin{proposition}
    \label{prop:L1_indistinguishability}
    For any pair of quantum states $\rho_1, \rho_2 \in \mathcal{D}(\mathcal{H})$ it holds that
    \begin{equation}
        \rho_1 \approx_{D(\rho_1, \rho_2)} \rho_2.
    \end{equation}
\end{proposition}

\begin{definition} (Indistinguishability -- Asymptotical) \\
\label{def:distinguishableQA}
Let $\lbrace \rho_1^{(k)} \in \mathcal{D}(\mathcal{H}_k) \rbrace$ and $\lbrace \rho_2^{(k)} \in \mathcal{D}(\mathcal{H}_k) \rbrace$ be two families of density operators. We say that the two families are {\em statistically indistinguishable} if there exists a negligible function $\varepsilon(k) \geq 0$ such that
\begin{equation}
    \rho_1^{(k)} \approx_{\varepsilon(k)} \rho_{2}^{(k)} \quad \textnormal{for all } k \in \mathbb{N}.
\end{equation}
Furthermore, we say the two families are {\em computationally indistinguishable} if for every efficient quantum program $T$, there exists a negligible function $\varepsilon_{T}(k) \geq 0$ such that
\begin{equation}
    \textnormal{Adv}_{T} \big( \rho_1^{(k)}, \rho_{2}^{(k)} \big) \leq \varepsilon_{T}(k) \quad \textnormal{for all } k \in \mathbb{N}.
\end{equation}
Statistical and computational indistinguishability for random variables is defined analogously.
\end{definition}

Recall from Section~\ref{sec:basicdefinitions} that, when the parameter $k$ is implicit, we may omit the explicit dependence on $k$ and use $\approx$ and $\approx^{(c)}$ for statistical and computational indistinguishability, respectively. We now turn our attention to the properties of indistinguishable states. It is worth noting that computational indistinguishability is only meaningful in terms of information security when the adversary is assumed to have limited computational capabilities. It is important then to define the type of quantum operations such adversary can perform:

\begin{definition} (Efficient quantum operation) \\
    \label{def:efficientQ}
    We say that a family $\lbrace \mathcal{E}^{(k)} \rbrace_{k=1}^{\infty}$ of quantum operations is efficient if there exists an efficient quantum program $T$ such that, for each $k$, $\mathcal{E}^{(k)}$ is the associated operation applied to the quantum part of the machine while running $T$ on input $k$
\end{definition}

The following properties are straightforward to prove from Definitions~\ref{def:distinguishableQ} and~\ref{def:distinguishableQA} and the properties of trace distance:
\begin{lemma} (Properties of indistinguishable states~I) \\
Let $\rho_{1}, \rho_{2},\rho_{3} \in \mathcal{D}(\mathcal{H})$:
    \begin{enumerate}
        \item $\rho_{1} \approx_{\varepsilon} \rho_{2} \wedge \rho_{2} \approx_{\varepsilon'} \rho_{3} \Rightarrow \rho_{1} \approx_{\varepsilon + \varepsilon'} \rho_{3}$.
        \item $\rho_{1} \approx_{\varepsilon} \rho_{2} \wedge \sigma_{1} \approx_{\varepsilon'} \sigma_{2} \Rightarrow \rho_{1}\otimes\sigma_{1} \approx_{\varepsilon + \varepsilon'}  \rho_{2}\otimes\sigma_{2}$.
        \item Let $x\in \mathcal{X}$. For any probability distribution $P_x$, assume that $\big( \forall x \in \mathcal{X} \big) \: \rho_1^{x} \approx_{\varepsilon^{x}} \rho_2^{x} $. Then 
        \begin{equation*}
            \sum_{x \in \mathcal{X}}P_x\rho_1^{x} \approx_{\varepsilon^{\textnormal{max}}} \sum_{x \in \mathcal{X}}P_x\rho_2^{x} \:\:\textnormal{where}\:\: \varepsilon^{\textnormal{max}} = \max_{x \in \mathcal{X}}\lbrace \varepsilon^{x} \rbrace.
        \end{equation*}
        \item $\rho_1 \approx_{\varepsilon} \rho_2 \Rightarrow \mathcal{E}(\rho_1) \approx_{\varepsilon} \mathcal{E}(\rho_2)$, for any completely positive, trace non-increasing map $\mathcal{E}$.
    \end{enumerate} 
    \label{lem:indistproperties}
\end{lemma}

\begin{lemma} (Properties of indistinguishable states~II) \\
Let  $\lbrace \rho_{1}(k) \rbrace, \lbrace \rho_{2}(k) \rbrace, \lbrace \rho_{3}(k) \rbrace$ be families of density operators parameterized by $k=1,2,\ldots$ The following statements hold for asymptotic computational indistinguishability: 
    \begin{enumerate}
        \item $\rho_{1} \approx^{(c)} \rho_{2} \wedge \rho_{2} \approx^{(c)} \rho_{3} \Rightarrow \rho_{1} \approx^{(c)} \rho_{3}$.
        \item $\rho_{1} \approx^{(c)} \rho_{2} \wedge \sigma_{1} \approx^{(c)} \sigma_{2} \Rightarrow \rho_{1}\otimes\sigma_{1} \approx^{(c)}  \rho_{2}\otimes\sigma_{2}$.
        \item Let $x\in \mathcal{X}$. For any probability distribution $P_x$, assume that $\big( \forall x \in \mathcal{X} \big) \: \rho_1^{x} \approx^{(c)} \rho_2^{x} $. Then 
        \begin{equation*}
            \sum_{x \in \mathcal{X}}P_x\rho_1^{x} \approx^{(c)} \sum_{x \in \mathcal{X}}P_x\rho_2^{x}.
        \end{equation*}
        \item $\rho_1 \approx^{(c)} \rho_2 \Rightarrow \mathcal{E}(\rho_1) \approx^{(c)} \mathcal{E}(\rho_2)$, where $\lbrace \mathcal{E}^{(k)} \rbrace$ is an efficient family of quantum operations acting on the respective $\rho_i(k)$.
    \end{enumerate} 
    \label{lem:cindistproperties}
\end{lemma}

\subsection{Entropic quantities}
We start off by defining a useful pair of quantities for measuring information in quantum systems: the max-entropy and the conditional min-entropy. The max entropy is a measure of the number of possible different outcomes that can result from measuring a quantum state, whereas the conditional min-entropy is a way of measuring the information that a party can infer from a quantum system given access to another correlated quantum system. This measures will be useful to bound the distance between states based on their internal correlations.

\begin{definition} (Max-entropy) \\
\label{def:maxentropy}
Let $\rho \in \mathcal{D}(\mathcal{H})$. The max-entropy of $\rho$ is defined as
\begin{equation}
    \label{eq:maxentropy}
    H_{\max}(\rho) = \log\big(\dim(\textnormal{supp} (\rho))\big),
\end{equation}
where $\textnormal{supp} (\rho)$ denotes the support subspace of $\rho$ and $\dim$ denotes its dimension.
\end{definition}

\begin{definition} (Min-entropy and conditional min-entropy) \\
\label{def:minentropy}
Let $\rho \in \mathcal{D}(\mathcal{H})$ and $\lambda_{\max}(\rho)$ denote the maximum eigenvalue of $\rho$. The min-entropy of $\rho$ is defined as
\begin{equation}
    \label{eq:minentropy}
    H_{\textnormal{min}}(\rho) = -\log(\lambda_{\max}(\rho)).
\end{equation}
Let $\rho_{AB} \in \mathcal{D}(\mathcal{H}_A \otimes \mathcal{H}_B)$ and $\sigma_{B} \in \mathcal{D}(\mathcal{H}_B)$. The conditional min-entropy of $\rho_{AB}$ given $\sigma_{B}$ is defined as
\begin{equation}
    \label{eq:minentropy1}
    H_{\textnormal{min}}(\rho_{AB}|\sigma_B) = -\log(\lambda_{\sigma_B}),
\end{equation}
where $\lambda_{\sigma_B}$ is the minimum real number such that $\lambda_{\sigma_B} (\mathbb{1}_{A} \otimes \sigma_B) - \rho_{AB}$ is non-negative. The conditional min-entropy of $\rho_{AB}$ given $\mathcal{H}_B$ is defined as
\begin{equation}
    \label{eq:minentropy2}
    H_{\textnormal{min}}(A|B)_{\rho} = 
    \!\!\!\!\!\!
    \sup_{\sigma_B \in \mathcal{D}(\mathcal{H}_B)} 
    \!\!\!\!\!\!
    H_{\textnormal{min}}(\rho_{AB}|\sigma_B),
\end{equation}
Furthermore, let $\varepsilon > 0$. The $\varepsilon$-smooth conditional min-entropy is defined as
\begin{equation}
    \label{eq:smoothminentropy}
    H_{\textnormal{min}}^{\varepsilon}(A|B)_{\rho} = 
    \!\!\!\!\!\!\!\!\!
    \sup_{\rho'_{AB} \in \mathcal{B}^{\varepsilon}(\rho_{AB})}  
    \!\!\!\!\!\!\!\!\!
    H_{\textnormal{min}}(A|B)_{\rho'},
\end{equation}
where $\mathcal{B}^{\varepsilon}(\rho_{AB}) = \lbrace \rho'_{AB} : D(\rho_{AB} , \rho'_{AB}) < \varepsilon \rbrace$.
\end{definition}

The smooth conditional min-entropy is in general hard to compute. Because of this, it is useful to have some tools to bound it for states that have some specific forms. In our case we are interested in states that are partially {\em classical}.

\begin{definition} (Partially classical states) \\
\label{def:classical}
A quantum state described by the density operator $\rho_{AB} \in \mathcal{D}(\mathcal{H}_A \otimes \mathcal{H}_B)$ is classical in $\mathcal{H}_A$ (or classical in $A$) if it can be written in the form
\begin{equation}
    \label{eq:classical}
    \rho_{AB} = \sum_{x} \lambda_x \ketbra{x}_{A} \otimes \rho^{x}_{B},
\end{equation}
where the set $\lbrace \ket{x} \rbrace_x$ is an orthonormal basis for $\mathcal{H}_A$. A multipartite state is said to be classical if it is classical in all its parts.
\end{definition}

When dealing with partially classical states as shown in Eq.~\eqref{eq:classical}, we will refer to the operators $\rho^{x}_{B}$ as the state of the system $B$ \textit{conditioned} to $x$.  

\begin{lemma} (Properties of min- and max-entropy) \\
\label{lem:entropyproperties}
Let $\varepsilon, \varepsilon' \geq 0$:
\begin{enumerate}
    \item $H_{\textnormal{min}} (\rho_{A} \otimes \rho_{B} \vert \rho_{B}) = -\log(\lambda_{\max}(\rho_A))$. 
    \item $H_{\textnormal{min}}^{\varepsilon+\varepsilon'}(AA'|BB')_{\rho \otimes \rho'} \geq H_{\textnormal{min}}^{\varepsilon}(A|B)_{\rho} + H_{\textnormal{min}}^{\varepsilon'}(A'|B')_{\rho'}$.
    \item $H_{\textnormal{min}}^{\varepsilon}(A|BC)_{\rho} \leq H_{\textnormal{min}}^{\varepsilon}(A|B)_{\rho}$.
    \item $H_{\textnormal{min}}^{\varepsilon}(AB|C)_{\rho} \leq H_{\textnormal{min}}^{\varepsilon}(A|BC)_{\rho} + H_{\max}(\rho_{B})$.
    \item $H_{\textnormal{min}}^{\varepsilon}(A|B)_{\rho} \geq \inf_x \lbrace 
    H_{\textnormal{min}}^{\varepsilon}(\rho_A^{x}) \rbrace$, whenever the state $\rho_{AB}$ is classical on $B$.
\end{enumerate}
\end{lemma}

We use universal hashing to implement randomness extraction in the final steps of the protocol. The proof both Lemmas~\ref{lem:entropyproperties} and~\ref{lem:qlhl} can be found in~\cite{renner08}.

\begin{definition} (Universal hashing)\\
\label{def:uh}
A set of functions $\textnormal{\textbf{F}}=\lbrace f_{i}: \{0,1\}^m \rightarrow \{0,1\}^n \rbrace$ is a universal hash family if, for all $ x,y\in \{0,1\}^m$, such that $x\neq y $, and $i$ chosen uniformly at random, we~have
\begin{equation}
	\textnormal{Pr} [f_i(x) = f_i(y) ] \leq \frac{1}{2^n} .
\end{equation}
\end{definition}

\vspace{0mm}

\begin{lemma}(Quantum leftover hash)\\ 
\label{lem:qlhl}
Let $\textnormal{\textbf{F}}=\lbrace f_{i}: \{0,1\}^n \rightarrow \{0,1\}^{\ell} \rbrace$ be a universal hash family, let $\mathcal{H}_A,\mathcal{H}_B,\mathcal{H}_F,\mathcal{H}_E$ be Hilbert spaces such that $\lbrace \ket{x} \rbrace_{x \in \bs^{n}}, 
\lbrace \ket{f_{i}} \rbrace_{f_{i} \in \textnormal{\textbf{F}}}$, and $\lbrace \ket{e} \rbrace_{e \in \bs^{\ell}}$ are orthonormal bases for $\mathcal{H}_A, \mathcal{H}_F,$ and $\mathcal{H}_E$ respectively. Then for any $\varepsilon \geq 0$ and any state of the form 
\begin{align}
    \label{eq:posthashstate}
    \rho_{ABFE} = \frac{1}{\vert \textnormal{\textbf{F}}\vert} \sum_{\substack{x \in \bs^{n} \\ f_{i} \in \textnormal{\textbf{F}}}}  \Big( &\lambda(x) \ketbra{f_i}_F  \ketbra{f_i(x)}_E  \nonumber \\
    &\otimes \ketbra{x}_A \rho^{x}_B \Big), 
\end{align}
it holds that
\begin{equation}
    \label{eq:qlhl}
    \rho_{EBF} \approx_{\varepsilon'} \mathsf{U}_E \otimes \rho_{BF}, 
\end{equation}
with
\begin{equation}
    \varepsilon' = \varepsilon + \frac{1}{2}\cdot2^{-\frac{1}{2}(H_{\textnormal{min}}^{\varepsilon}(A \vert B)_{\rho} - \ell)}.
\end{equation}
\end{lemma}

\section{Detailed proof of Theorem~\ref{thm:OKDsecurity}}
\label{sec:longproof}

\subsection{Supporting lemmas}

One of the main features to analyze for the security against a dishonest receiver is the potential information that he can learn about the the sender's strings given the quantum state that remains with him after the commit/open phase. In order to talk about the security of the protocol independently of the specific cheating strategy that may be used by the dishonest parties (or possible effects that the environment can have in the shared quantum state), we want to understand the properties that a quantum state that passes Alice's test at Step (7) can have. We do this through the following lemma, a version of which was originally proven in~\cite{damgard09}. Here, we provide a more self contained statement and make explicit the trace distance bound.

\begin{lemma} 
\label{lem:testsecurity}
Let $\varepsilon > 0$, $I = \lbrace 1,\ldots,N \rbrace$, and $\rho_{T,\hat{X}, X, E}$ be a density operator of the form
\begin{equation}
    \label{eq:committedstate}
    \begin{split}
    \rho_{T,\hat{X}, X, E} &= 
    \!\!\!\!\!
    \sum_{\substack{I_{1},I_{2} \\\in \mathcal{P}(I) \setminus \lbrace \emptyset \rbrace}}
    \!\!\!\!\!
    q(I_{1},I_{2}) \ketbra{I_{1},I_{2}}_T \otimes \ketbra{\hat{x}}_{\hat{X}} \otimes \ketbra{\psi}_{X,E} \\
    \ket{\psi}_{X,E} &= \sum_{x \in \bs^N} \beta^{x} \ket{x}_X \ket{\phi^x}_E,
    \end{split}
\end{equation}
where $\dim(\mathcal{H}_{\hat{X}})=\dim(\mathcal{H}_{X})= 2^N$ and $\mathcal{P}(I)$ denotes the set of all subsets of $I$. For each  $I_{1},I_{2}$ define the set
\begin{equation}
    B_{I_{1},I_{2}} = \lbrace x \in \bs^N : \vert r_{\textnormal{H}}(x_{I_{1}} \oplus\hat{x}_{I_{1}})-r_{\textnormal{H}}(x_{I_{2}}\oplus\hat{x}_{I_{2}}) \vert < \varepsilon \rbrace.
\end{equation}
Additionally, let $\mathcal{Q}(N,\varepsilon)$ be a function such that, whenever the subsets $I_{1},I_{2}$ are sampled according to $q$, it holds that 
\begin{equation}
    \label{eq:deviationbound}
    \Pr[|r_{\textnormal{H}}(x_{I_{1}} \oplus\hat{x}_{I_{1}})-r_{\textnormal{H}}((x_{I_{2}}\oplus\hat{x}_{I_{2}})| > \varepsilon]  \leq \mathcal{Q}(N,\varepsilon)
\end{equation}
independently of $x$. There exists a state $\Tilde{\rho}_{T,\hat{X}, X, E}$ of the form
\begin{equation}
    \begin{split}
    \tilde{\rho}_{T,\hat{X}, X, E} &= 
    \!\!\!\!\!
    \sum_{\substack{I_{1},I_{2} \\\in \mathcal{P}(I) \setminus \lbrace \emptyset \rbrace}} 
    \!\!\!\!\!
    q(I_{1},I_{2}) \ketbra{I_{1},I_{2}}_T \otimes \ketbra{\hat{x}}_{\hat{X}} \otimes \ketbra{\psi_{I_{1},I_{2}}}_{X,E} \\
    \ket{\psi_{I_{1},I_{2}}} &= \sum_{x \in B_{I_{1},I_{2}}} \tilde{\beta}^{x}_{I_{1},I_{2}} \ket{x}_X \ket{\phi^{x}_{I_{1},I_{2}}}_E,
    \end{split}
\end{equation}
such that
\begin{equation}
    \label{eq:commitbound}
    D(\rho_{T,\hat{X}, X, E} , \Tilde{\rho}_{T,\hat{X}, X, E}) \leq \mathcal{Q}(N,\varepsilon)^{\frac{1}{2}}. 
\end{equation}
\end{lemma}

\begin{proof}
First, we choose an adequate definition for the $\tilde{\beta}^{x}_{I_{1},I_{2}}$ and then show that, under that choice, the bound in Eq.~\eqref{eq:commitbound} holds. Note that we can write the state
\begin{equation}
\begin{split}
    \ket{\psi}_{X,E} &= \sum_{x \in \bs^N} \beta_x \ket{x}_X \ket{\phi^x}_E \\
    &= \underbrace{\left(\sum_{x \in B_{I_{1},I_{2}}} |\beta_x|^2 \right)^{\frac{1}{2}}}_{\lambda_{I_{1},I_{2}}} 
    \underbrace{\frac{\sum_{x \in B_{I_{1},I_{2}}} \beta_x \ket{x}_X \ket{\phi^x}_E}{\left(\sum_{x \in B_{I_{1},I_{2}}} |\beta_x|^2 \right)^{\frac{1}{2}}}}_{\ket{\psi_{I_{1},I_{2}}}_{X,E}} +
    \underbrace{\left(\sum_{x \notin B_{I_{1},I_{2}}} |\beta_x|^2 \right)^{\frac{1}{2}}}_{\lambda^{\perp}_{I_{1},I_{2}}} 
    \underbrace{\frac{\sum_{x \notin B_{I_{1},I_{2}}} \beta_x \ket{x}_X \ket{\phi^x}_E}{\left(\sum_{x \notin B_{I_{1},I_{2}}} |\beta_x|^2 \right)^{\frac{1}{2}}}}_{\ket{\psi^{\perp}_{I_{1},I_{2}}}_{X,E}} 
\end{split}
\end{equation}
The trace distance between the pure states $\ket{\psi}$ and $\ket{\psi_{I_{1},I_{2}}}$ is given by $\sqrt{1-|\bracket{\psi}{\psi_{I_{1},I_{2}}}|^2} = \lambda^{\perp}_{I_{1},I_{2}}$, hence the trace distance between the complete joint states is given by
\begin{equation}
\label{eq:commitdistance}
    \begin{split}
    D(\rho_{T,\hat{X}, X, E}, \Tilde{\rho}_{T,\hat{X}, X, E})^2 &\leq\  \left( \sum_{I_{1},I_{2}  \in \mathcal{T}(N,\alpha)} q D( \ketbra{\psi}_{X,E} ,  \ketbra{\psi_{I_{1},I_{2}}}_{X,E} ) \right)^2 \\
    &= \left(\sum_{I_{1},I_{2}  \in \mathcal{T}(N,\alpha)} q \lambda^{\perp}_{I_{1},I_{2}} \right)^2 \\
    &\leq  \sum_{I_{1},I_{2}  \in \mathcal{T}(N,\alpha)} q {\lambda^{\perp}_{I_{1},I_{2}}}^2,
    \end{split}
\end{equation}
where the Jensen's inequality was used in the last Step. We proceed now to bound the right side of Eq.~\eqref{eq:commitdistance}.  For that purpose consider the function
\begin{equation}
    \xi(I_{1},I_{2},x)= 
    \begin{cases}
    0 &\quad \text{if} \quad x \in B_{I_{1},I_{2}} \\
    1 &\quad \text{otherwise}
    \end{cases},
\end{equation}
so that
\begin{equation}
\begin{split}
    \sum_{I_{1},I_{2}} q(I_{1},I_{2}) \xi(I_{1},I_{2},x) = \Pr[|r_{\textnormal{H}}(\hat{x}\oplus x\vert_{I_{1}})-r_{\textnormal{H}}(\hat{x}\oplus x\vert_{I_{2}})| > \varepsilon] = \mathcal{Q}(N,\varepsilon).
\end{split}
\end{equation}
Hence,
\begin{equation}
    \begin{split}
    ||\rho_{T,\hat{X}, X, E} - \Tilde{\rho}_{T,\hat{X}, X, E}||^2 &\leq  \sum_{I_{1},I_{2}} q(I_{1},I_{2}) {\lambda^{\perp}_{I_{1},I_{2}}}^2 \\
    &= \sum_{I_{1},I_{2}} q(I_{1},I_{2})\sum_{x \notin B_{I_{1},I_{2}}} |\beta_x|^2 \\
    &= \sum_{I_{1},I_{2}} q(I_{1},I_{2})\sum_{x \in \bs^N}\xi(I_{1},I_{2},x) |\beta_x|^2 \\
    &= \sum_{x \in \bs^N} |\beta_x|^2 \sum_{I_{1},I_{2}} q(I_{1},I_{2})\xi(I,x) \\
    &\leq \sum_{x \in \bs^N} |\beta_x|^2 \mathcal{Q}(N,\varepsilon) = \mathcal{Q}(N,\varepsilon),
    \end{split}
\end{equation}
as required.
\end{proof}

In order to use the above result in the context of the $\pi_{\textnormal{QROT}}$ protocol, we need to find an appropriate function $\mathcal{Q}(N,\varepsilon)$ that satisfies Eq.~\eqref{eq:deviationbound} for the case when the $\hat{x},x$ are the respective measurement outcomes of Alice and Bob when measuring in the same basis. We do this through the following lemma based on the Hoeffding inequality for sampling without replacement.

\begin{definition}
\label{def:sampling}
Given a set $I$ and an integer $n \leq |I|$, define the set $\mathcal{T}(n, I)$ as the set of all subsets of $I$ with size $n$.
\end{definition}

\begin{lemma} (Hoeffding's inequalities)\\
\label{lem:hoeffding}
Let $x \in \bs^N$, $\delta > 0$, $I = \lbrace 1,\ldots,N \rbrace$ and $0 < \alpha < \frac{1}{2}$ such that $\alpha N \in \mathbb{N}$.
\begin{enumerate}[label=(\alph*)]
    \item (Inequality for sampling without replacement comparing the sampled subset with the whole set) For $I_t \in \mathcal{T}(\alpha N, I)$ sampled uniformly, it holds that 
        \begin{equation}
        \label{eq:hoeffding1}
        \Pr[|r_{\textnormal{H}}(x\vert_{I_t})-r_{\textnormal{H}}(x)| > \delta] \leq 2e^{-2 \alpha N \delta^{2}}.
    \end{equation}
    \item (Inequality for sampling without replacement comparing the sampled subset with its complement)
        \begin{equation}
        \label{eq:hoeffding2}
        \Pr[|r_{\textnormal{H}}(x\vert_{I_t})-r_{\textnormal{H}}(x\vert_{\bar{I}_t})| > \delta] \leq 2 e^{-2 \alpha (1-\alpha)^{2} N \delta^{2}}.
    \end{equation}
    \item (Inequality for sampling without replacement comparing the sampled subset with its complement and ignoring part of the sample) Let $n \in \lbrace n_0,
    \ldots, \alpha N \rbrace$ be sampled according to some distribution $P(n)$. For $I_s \in \mathcal{T}(n, I_t)$ sampled uniformly, it holds that 
        \begin{equation}
        \label{eq:hoeffding3}
        \Pr[|r_{\textnormal{H}}(x\vert_{I_s})-r_{\textnormal{H}}(x\vert_{\bar{I}_t})| > \delta] \leq 2 (e^{-\frac{1}{2} \alpha (1-\alpha)^{2} N \delta^{2}} +  e^{-\frac{1}{2} n_0 \delta^{2}}).
    \end{equation}
\end{enumerate}
\end{lemma}
\begin{proof}
\begin{enumerate}[label=(\alph*)]
    \item This is the original Hoeffding inequality for sampling without replacement, the proof of which can be found in \cite{hoeffding63}.
    \item Note that we can write $r_{\textnormal{H}}(x) = \alpha r_{\textnormal{H}}(x\vert_{I_t}) + (1-\alpha) r_{\textnormal{H}}(x\vert_{\bar{I}_t})$. Substituting $r_{\textnormal{H}}(x)$ in (\ref{eq:hoeffding1}) we get
    \begin{align}
        \Pr[|r_{\textnormal{H}}(x\vert_{I_t})-\alpha r_{\textnormal{H}}(x\vert_{I_t}) + (1-\alpha) r_{\textnormal{H}}(x\vert_{\bar{I}_t})| > \delta'] 
        &= \Pr[|r_{\textnormal{H}}(x\vert_{I_t}) - r_{\textnormal{H}}(x\vert_{\bar{I}_t})| > \delta'/(1-\alpha)] \\
        &\leq 2e^{-2 \alpha N \delta'^{2}}.
      \end{align}
      The result is obtained by taking $\delta' = (1-\alpha)\delta$
      \item Let us consider first the case where $n$ is fixed. From the triangle inequality we know that
      \begin{align}
           |r_{\textnormal{H}}(x\vert_{I_{s(n)}})-r_{\textnormal{H}}(x\vert_{\bar{I}_t})| > \delta &\Rightarrow |r_{\textnormal{H}}(x\vert_{I_{s(n)}})-r_{\textnormal{H}}(x\vert_{I_t})| + |r_{\textnormal{H}}(x\vert_{I_t})-r_{\textnormal{H}}(x\vert_{\bar{I}_t})| > \delta \\
           &\Rightarrow |r_{\textnormal{H}}(x\vert_{I_{s(n)}})-r_{\textnormal{H}}(x\vert_{I_t})| > \delta/2 \vee |r_{\textnormal{H}}(x\vert_{I_t})-r_{\textnormal{H}}(x\vert_{\bar{I}_t})| > \delta/2,
      \end{align}
      and hence, by the union bound
      \begin{align}
          \Pr[|r_{\textnormal{H}}(x\vert_{I_{s(n)}})-r_{\textnormal{H}}(x\vert_{\bar{I}_t})| > \delta] &\leq \Pr[|r_{\textnormal{H}}(x\vert_{I_{s(n)}})-r_{\textnormal{H}}(x\vert_{I_t})| + |r_{\textnormal{H}}(x\vert_{I_t})-r_{\textnormal{H}}(x\vert_{\bar{I}_t})| > \delta] \\
          &\leq \Pr[|r_{\textnormal{H}}(x\vert_{I_{s(n)}})-r_{\textnormal{H}}(x\vert_{I_t})| > \delta/2] \nonumber \\ &\quad \quad \quad + \Pr[|r_{\textnormal{H}}(x\vert_{I_t})-r_{\textnormal{H}}(x\vert_{\bar{I}_t})| > \delta/2] \label{eq:triangleinequality}\\
          &\leq 2e^{-\frac{1}{2} n \delta^{2}} + 2e^{-\frac{1}{2} \alpha (1-\alpha)^{2} N \delta^{2}},
      \end{align}
      where the last expression comes from applying the (b) and (a) inequalities to the first and second terms of (\ref{eq:triangleinequality}) respectively. Using this, we can consider the case in which $n$ is not fixed, but instead follows a probability distribution $P(n)$ such that $P(n) = 0$ for $n<n_0$. For this case
      \begin{align}
           \Pr[|r_{\textnormal{H}}(x\vert_{I_s})-r_{\textnormal{H}}(x\vert_{\bar{I}_t})| > \delta] &= \sum_n P(n) \Pr[|r_{\textnormal{H}}(x\vert_{I_{s(n)}})-r_{\textnormal{H}}(x\vert_{\bar{I}_t})| > \delta] \\
           &\leq \sum_n P(n) 2 (e^{-\frac{1}{2} \alpha (1-\alpha)^{2} N \delta^{2}} +  e^{-\frac{1}{2} n \delta^{2}}) \\
           &\leq 2\sum_n P(n) (e^{-\frac{1}{2} \alpha (1-\alpha)^{2} N \delta^{2}} +  e^{-\frac{1}{2} n_0 \delta^{2}}) \\
           &= 2(e^{-\frac{1}{2} \alpha (1-\alpha)^{2} N \delta^{2}} +  e^{-\frac{1}{2} n_0 \delta^{2}}).
      \end{align}
\end{enumerate}
\end{proof}

The following lemma helps us bound the conditional min-entropy of a partially measured pure state by comparing it with the one of an appropriately chosen, partially measured mixed state. A proof of this result can be found in~\cite{bouman10}.

\begin{lemma} (Entropy bound for post-measurement states)\\
\label{lem:minentropybound}
Let $\mathcal{H}_A$ and $\mathcal{H}_E$ be Hilbert spaces and $\lbrace \ket{x} \rbrace_{x \in \mathcal{X}}, \lbrace \ket{y} \rbrace_{y \in \mathcal{Y}}$ be orthonormal bases for $\mathcal{H}_A$. Let $J \subseteq \mathcal{X}$, define the states
\begin{equation}
    \rho_{AE} = \ketbra{\phi}_{AE} \;\;\;\; \textit{with} \;\;\;\; \ket{\phi}_{AE} = \sum_{x \in J} \beta_x\ket{x}_A\ket{\phi^{x}}_{E},
\end{equation}
\begin{equation}
    \rho^{\textnormal{mix}}_{AE}= \sum_{x \in J} {|\beta_{x}|^{2}} \ketbra{x}_A \otimes \ketbra{\phi^{x}}_{E}.
\end{equation}
Denote by $\sigma_{YE}$ and  $\sigma^{\textnormal{mix}}_{YE}$ the states resulting from measuring the subsystem $A$ of $\rho_{AE}$ and $\rho^{\textnormal{mix}}_{AE}$ respectively in the basis $\lbrace \ket{y} \rbrace_{y \in \mathcal{Y}}$, storing the result in the system $Y$, and then tracing out the $A$ subsystem; then it holds that
\begin{equation}
    H_{\textnormal{min}}(Y | E)_{\sigma} \geq H_{\textnormal{min}}(Y | E)_{\sigma^{\textnormal{mix}}} - \log|J|.
\end{equation}
\end{lemma}

\subsection{Proof of Lemma~\ref{lem:IRverifiability}}
\label{sec:parestproof}

Here we present a proof of Lemma~\ref{lem:IRverifiability} used in the protocol's correctness analysis in Section~\ref{sec:secanalysis}.
\begin{lemma}
    \label{lem:IRverifiability2}
    Let $X^{A}_{I_0}, X^{A}_{I_1}, C, Y^{B}$ denote the systems holding the information of the respective values $x^{A}_{I_0}, x^{A}_{I_1}, c$, and  $y^{B}$ of $\pi_{\textnormal{QROT}}$. Denote by $\rho^{\top}$ the parties' joint state at the end of Step (11) conditioned that Bob constructed the sets $(I_0, I_1)$ during Step (9) and the protocol has not aborted. Assume both parties follow the steps of the protocol, then
    \begin{equation}
        \rho^{\top}_{X^{A}_{I_0}, X^{A}_{I_1}, C, Y^{B}} \approx_{\varepsilon_{\textnormal{IR}}(k')} \tilde{\rho}^{ \top}_{X^{A}_{I_0}, X^{A}_{I_1}, C, Y^{B}},
    \end{equation}
    where $\varepsilon_{\textnormal{IR}}(k')$ is a negligible function given by the security of the underlying Information Reconciliation scheme, $k'$ its associated security parameter, and
    \begin{equation}
        \label{eq:IRverifiability2}
        \tilde{\rho}^{ \top}_{X^{A}_{I_0}, X^{A}_{I_1}, C, Y^{B}} = \frac{1}{2^{(2N_{\textnormal{raw}}+1)}}
        \!
        \sum_{\substack{x_{I_0}, x_{I_1} \\ c}}
        \!
        \ketbra{x_{I_0}}_{X^{A}_{I_0}} \ketbra{x_{I_1}}_{X^{A}_{I_1}}  \ketbra{x_{I_0}}_{Y^{B}} \ketbra{c}_{C}.
    \end{equation}
\end{lemma}

\begin{proof}
Note that, because the state shared by Alice at Step (1) of the protocol is a tensor product of maximally entangled states, the state of Alice's part is a product of maximally mixed states. This means that, regardless of the measurement bases $\theta^{A}$, the outcome of her measurements $x^{A}$ is always uniform in $\bs^{N_0}$. Let $\rho^{(2)}_{\Theta^{A}, \Theta^{B}, X^{A},X^{B}}$ be the state of the parties' respective measurement bases and outcomes at the end of Step (2) of the protocol, we can write 
\begin{align}
    \label{eq:correctnesspostmeasurement}
    \rho^{(2)}_{\Theta^{A}, \Theta^{B}, X^{A},X^{B}} =& \frac{1}{2^{2N_{0}}} \sum_{\theta^{A}, \theta^{B}} \ketbra{\theta^{A}, \theta^{B}}_{\Theta^{A} \Theta^{B}} \frac{1}{2^{N_{0}}} \sum_{x^{A}} \ketbra{x^{A}}_{X^{A}} \nonumber \\
    &\otimes \sum_{x^{B}} P(x^{B} | x^{A}, \theta^{A}, \theta^{B})\ketbra{x^{B}}_{X^{B}},
\end{align}
where $P(x^{B} | x^{A}, \theta^{A}, \theta^{B})$ denotes the probabilities of Bob's outcomes given each parties measurement bases and Alice's measurement outcomes, which in turn depends on the effect of the transmission channel when the state was shared from Alice's laboratory. All operations will be classical from this point onwards. To arrive to Eq~\eqref{eq:IRverifiability2}, we first need to show that the abort operations within the protocol do not bias the distribution of possible values of $x_{I_0}$ and $x_{I_1}$, and then use the verifiability property of the IR scheme to ensure that $y^{B} = x_{I_0}$ with high probability.

We will consider now the two abort instructions at Steps (7) and (9) as a single quantum operation $\mathcal{E}$ that maps the state to the zero operator if any of the two abort conditions is satisfied and applies the identity map otherwise. Let us first separate the values of $\theta^{{A/B}}$ and $x^{B}$ that ``survive'' the abort operation. For any given values of $x^{A}$ and $I_{t}$, define the sets:
\begin{equation}
    J^{(1)}_{I_{t}, x^{A}} = \lbrace (\theta^{A}, \theta^{B}) : w_{\textnormal{H}}\left(\overline{\theta^{A}_{I_t} \oplus \theta^{B}_{I_t}} \right) \geq N_{\textnormal{check}} \land N_{\textnormal{raw}}  \leq w_{\textnormal{H}}(\theta^{A}_{\bar{I}_t} \oplus \theta^{B}_{\bar{I}_t}) \leq (1-\alpha)N_0 - N_{\textnormal{raw}} \rbrace
\end{equation}
\begin{align}
    J^{(2)}_{I_{t}, x^{A}, \theta^{A}, \theta^{B}} = \lbrace x^{B} : r_{\textnormal{H}}(x^{A}_{I_s} \oplus x^{B}_{I_s}) \leq p_{\textnormal{max}} \rbrace,
\end{align}
where $w_{H}(\cdot)$ denotes the Hamming weight function. Let $\mathcal{T} = \mathcal{T}(\alpha N_0, I)$ be the set of all subsets of $I=\lbrace 1,\ldots, N_0 \rbrace$ of size  $\alpha N_0$, and denote by $S$ the system where Bob holds the information of the sets $I_0$ and $I_1$. The joint state of the systems $C, X^{A/B}, \Theta^{A/B}, S$ at the end of Step (10) of the protocol is 
\begin{align}
    \rho^{(10)}_{C,X^{A},X^{B}, \Theta^{A}, \Theta^{B}, S} =& \frac{1}{2}  \sum_{c} \ketbra{c}_{C} \frac{1}{2^{N_0}}  \sum_{x^{A}} \ketbra{x^{A}}_{X^{A}}
    \frac{1}{|\mathcal{T} |\cdot 2^{2 N_0}}
    \sum_{I_t}
    \!
    \sum_{\substack{(\theta^{A}, \theta^{B}) \\ \in J^{(1)}}}
    \!\!\!\!
    \ketbra{\theta^{A}, \theta^{B}}_{\theta^{A}, \theta^{B}} \nonumber \\
    &\otimes 
    \!\!\!\!
    \sum_{x^{B}\in J^{(2)}}
    \!\!\!\!
    P(x^{B} | x^{A}, \theta^{A}, \theta^{B}) \ketbra{x^{B}}_{X^{B}}
    \sum_{I_0,I_1}
    P(I_0,I_1 | I_{t}, \theta^{A}, \theta^{B})
    \ketbra{I_0,I_1}_{S},
\end{align}
where the conditional distribution $P(I_0,I_1 | I_{t}, \theta^{A}, \theta^{B})$ notably does not depend on $x^{A}$ or $c$. Tracing out the $\Theta^{A/B}$ systems and rearranging terms we get
\begin{align}
    \label{eq:correctnessbeforeIR}
    \rho^{(10)}_{S,C,X^{A},X^{B}} =& \sum_{I_0,I_1} P(I_0,I_1) \ketbra{I_0,I_1}_{S} \nonumber \\
    &\otimes \underbrace{\frac{1}{2^{N_0 +1}}  \sum_{x^{A},c} \ketbra{x^{A}, c}_{X^{A}, C} \sum_{x^{B}} P(x^{B} | I_0,I_1, x^{A}) \ketbra{x^{B}}_{X^{B}}}_{\textnormal{conditioned state}\quad \rho^{(10)}_{C,X^{A},X^{B}}(I_0,I_1) },
\end{align}
where
\begin{equation}
    P(I_0,I_1) = \frac{1}{|\mathcal{T} |\cdot 2^{2 N_0}} \sum_{I_t} \sum_{\substack{(\theta^{A}, \theta^{B}) \\ \in J^{(1)}}} P(I_0,I_1 | I_{t}, \theta^{A}, \theta^{B}),
\end{equation}
and
\begin{equation}
    P(x^{B} | I_0,I_1, x^{A}) = \frac{\sum_{I_t} \sum_{\substack{(\theta^{A}, \theta^{B}) \\ \in J^{(1)}}} P(x^{B} | x^{A}, \theta^{A}, \theta^{B}) P(I_0,I_1 | I_{t}, \theta^{A}, \theta^{B})} {\sum_{I_t} \sum_{\substack{(\theta^{A}, \theta^{B}) \\ \in J^{(1)}}} P(I_0,I_1 | I_{t}, \theta^{A}, \theta^{B})}.
\end{equation}
Now that we have a form for the conditioned state as pointed out in Eq.~\eqref{eq:correctnessbeforeIR}, we can move to the action of Step (11), where Bob computes $y^{B} = \dec(\syn(x^{A}_{I_0}, x^{B}_{I_0}))$. The resulting state of the systems $C,X^{A}_{I_0}, X^{A}_{I_1}, Y^{B}$ is then given by:
\begin{align}
    \rho^{(11)}_{C,X^{A}_{I_0}, X^{A}_{I_1}, Y^{B}}(I_0,I_1) =& \frac{1}{2^{2N_{\textnormal{raw}} +1}} \sum_{\substack{x^{A}_{I_0}, x^{A}_{I_1} \\ c }} \ketbra{x^{A}_{I_0}, x^{A}_{I_1}}_{X^{A}_{I_0}, X^{A}_{I_1}} \ketbra{c}_{C} \nonumber \\ &\otimes \left( P_{\textnormal{correct}}\ketbra{x^{A}_{I_0}}_{Y^{B}} + P_{\bot} \ketbra{\bot}_{Y^{B}} + P_{\textnormal{error}} \sigma_{Y^{B}}\right),
\end{align}
for some coefficients $P_{\textnormal{correct}}, P_{\bot}, P_{\textnormal{error}}$, and state $\sigma$ orthogonal to both $\ketbra{x^{A}_{I_0}}$ and $\ketbra{\bot}$. By applying the verifiability property of the IR scheme with security parameter $k'$ (where $P_{\textnormal{error}} = \varepsilon_{\textnormal{IR}}(k')$), and conditioning the resulting state to not having aborted, we get the desired result
\begin{equation}
    \rho^{(11),\top}_{X^{A}_{I_0}, X^{A}_{I_1}, Y^{B},C}(I_0,I_1) \approx_{\varepsilon_{\textnormal{IR}}(k')} \frac{1}{2^{2N_{\textnormal{raw}} +1}} \sum_{\substack{x^{A}_{I_0}, x^{A}_{I_1} \\ c }} \ketbra{x^{A}_{I_0}, x^{A}_{I_1}}_{X^{A}_{I_0}, X^{A}_{I_1}} \ketbra{x^{A}_{I_0}}_{Y^{B}} \ketbra{c}_{C}.
\end{equation}
\end{proof}

\subsection{Proof of Lemma~\ref{lem:commithidingsecurity}}
\label{sec:commithidingproof}

Here we present a proof of the Lemma~\ref{lem:commithidingsecurity} introduced in Section~\ref{sec:dishonestalice} regarding the hiding property of the commitment in the context of $\pi_{\textnormal{QROT}}$.

\begin{lemma}
    \label{lem:commithidingsecurityApp}
    Assuming Bob follows the protocol, for any $J \subseteq I$, the state of the system $A, \CCOM, \OOPEN_J, \Theta^{B}_{\bar{J}}$ after Step (4) satisfies
    \begin{equation}
        \label{eq:commithidingsecurityApp}
        \rho_{A, \CCOM, \OOPEN_J, \Theta^{B}_{\bar{J}}} \approx^{(c)} \rho_{A, \CCOM, \OOPEN_J} \otimes \mathsf{U}_{\Theta^{B}_{\bar{J}}},
    \end{equation}
    where 
    \begin{equation}
        \mathsf{U}_{\Theta^{B}_{\bar{J}}} = \frac{1}{2^{N_0 - |J|}}
        \sum_{\theta^{B}_{\bar{J}}}
        \ketbra{\theta^{B}_{\bar{J}}}_{\Theta^{B}_{\bar{J}}},
    \end{equation} 
    denotes the uniform distribution over all possible values of $\theta^{B}_{\bar{J}}$.
\end{lemma}

\begin{proof}
We start by describing the general form of the state prepared by Alice at the beginning of the protocol, which is sent to Bob. Because the value of $r \in \bs^{n_r}$ sent by Alice in Step (3) as part of the commitment scheme is independent of any of Bobs actions, we can consider without loss of generality that it is prepared at the start of the protocol. The state shared at the beginning of the protocol (after Bob receives his qubit shares) has a general form given by
\begin{equation}
    \label{eq:cheataliceinit}
    \rho^{(0)}_{\Phi^{B},R,A} = \ketbra{\psi^{(0)}}_{\Phi^{B}, R, A} \quad \quad \quad \ket{\psi^{(0)}}_{\Phi^{B}, R, A} = \sum_{x, r} \alpha^{x,r} \ket{x}_{\Phi^{B}} \ket{r}_{R} \ket{\phi^{x,r}}_{A},
\end{equation}
where the $\ket{\phi^{x,r}}$ are not necessarily orthogonal. Using the Hadamard operator $H$, we can define the states
\begin{equation}
     \ket{x,\theta} = H^{\theta}\ket{x} = H^{\otimes \theta_1} \otimes \ldots \otimes H^{\theta_{N_0}}\ket{x},
\end{equation}
and write, for any string of basis choices $\theta \in \bs^{N_0}$, the state $\ket{\psi^{(0)}}_{\Phi^{B}, R, A}$ as
\begin{align}
    \label{eq:premeasurementdishonestlalice}
    \ket{\psi^{(0)}}_{\Phi^{B}, R, A} &= \sum_{x',r} \alpha^{x',r} \ket{\phi^{x',r}}_{A} \ket{r}_{R} \sum_{x}\bra{x}H^{\theta}\ket{x'}\ket{x, \theta}_{\Phi^{B}} \nonumber \\
    &=  \sum_{x, r} \beta^{x, \theta, r} \ket{x, \theta}_{\Phi^{B}} \ket{r}_{R} \ket{\phi^{x, \theta, r}}_{A},
\end{align}
with
\begin{align}
     \beta^{x, \theta, r} = \left( \sum_{x'} |\bra{x}H^{\theta}\ket{x'} \alpha^{x',r}|^2 \right)^{\frac{1}{2}} \quad \quad \quad  \ket{\phi^{x, \theta, r}} =   \left(\beta^{x, \theta, r}\right)^{-1} \sum_{x'} \bra{x}H^{\theta}\ket{x'} \alpha^{x',r} \ket{\phi^{x',r}}.
\end{align}

After uniformly sampling the values of $\theta^{B}$, Bob proceeds to perform his measurement on his qubit shares. Let $\mathcal{H}_{X^{B}}$ denote the system where Bob records the outcome string $x^{B}$. Additionally, at Step (3) Bob receives the value of $r$, this is a classical message, which we model as Bob receiving the $\mathcal{H}_{R}$ system and measuring it in the computational basis upon arrival. We can now easily use Eq.~\eqref{eq:premeasurementdishonestlalice} to get the post-measurement state at the end of Step (3) after tracing out $\mathcal{H}_{\Phi^{B}}$, which is given by
\begin{equation}
    \rho^{(1)}_{X^{B}, \Theta^{B}, R, A} = \frac{1}{2^{N_0}}\sum_{x^{B}, \theta^{B}, r} |\beta^{x^{B}, \theta^{B}, r}|^2 \ketbra{\theta^{B}}_{\Theta^{B}} \ketbra{x^{B}}_{X^{B}} \ketbra{r}_{R} \ketbra{\phi^{x^{B}, \theta^{B}, r}}_{A}.
\end{equation}
Before proceeding with the protocol, it will be useful to state some basic properties of the above state. Note that even though each of the $\ket{\phi^{x^{B}, \theta^{B}, r}}$ depends on $\theta^{B}$, the partial trace 
\begin{equation}
    \Tr_{X^{B}}[\rho^{(1)}] = \frac{1}{2^{N_0}}\sum_{\theta^{B}} \ketbra{\theta^{B}}_{\Theta^{B}} \otimes \sum_{x, r} |\alpha^{x, r}|^{2} \ketbra{r}_{R} \ketbra{\phi^{x,r}}_{A} 
\end{equation}
has a product form. Furthermore, because honest Bob measures each of his qubits independently, for any $I' \subseteq I$, the partial trace
\begin{align}
    \label{eq:basisindependence}
    \Tr_{X^{B}_{I'}}[\rho^{(1)}] = \frac{1}{2^{|I'|}}
    \!\!\!\!\!
    \sum_{\substack{\theta^{B}_{I'} \\ \in \bs^{|I'|}}} 
    \!\!\!\!\!
    \ketbra{\theta^{B}_{I'}}_{\Theta^{B}_{I'}} \otimes \rho^{(1)}_{X^{B}_{\bar{I'}}, \Theta^{B}_{\bar{I'}}, R, A}
\end{align}
also has a product form. As Alice will be able to perform quantum operations on her part of the joint state, it's important to note that the above property holds even after the $A$ subsystem undergoes an arbitrary CPTP transformation independent of $\Theta^{B}_{I'}$ and $X^{B}_{I'}$. During Step (4) Bob commits his values of $\theta^{B}$ and $x^{B}$, for that he samples the values of $s = \left( s_1, \ldots, s_{N_0} \right)$ and computes 
\begin{align}
    \ccom &= \left( \com((\theta^{B}_1, x^{B}_1), s_1, r), \ldots, \com((\theta^{B}_{N_0}, x^{B}_{N_0}), s_{N_0}, r) \right) \nonumber \\
    \textnormal{open} &= \left( \open((\theta^{B}_1, x^{B}_1), s_1), \ldots, \open((\theta^{B}_{N_0}, x^{B}_{N_0}), s_{N_0}) \right),
\end{align}
leading to the state
\begin{align}
    \label{eq:postcom}
    \rho^{(2)} =& \frac{1}{2^{N_0}}\sum_{x^{B}, \theta^{B}, r} |\beta^{x^{B}, \theta^{B}, r}|^2 \ketbra{\theta^{B}}_{\Theta^{B}} \ketbra{x^{B}}_{X^{B}} \ketbra{r}_{R} \ketbra{\phi^{x^{B}, \theta^{B}, r}}_{A}  \nonumber \\
    &\bigotimes_{i \in I} \bigg(  \frac{1}{2^{n_s}} \sum_{s_i} \ketbra{\com((\theta^{B}_i, x^{B}_i), s_i, r)}_{\textnormal{COM}_i}  \nonumber \\
    & \otimes \ketbra{\open((\theta^{B}_i, x^{B}_i), s_i)}_{\textnormal{OPEN}_i}  \bigg).
\end{align}
Let $J \subseteq I$, we now want to use the hiding property of the commitment scheme to approximate the state \eqref{eq:postcom} to one where the values of $\ccom$ and $\oopen_{J}$ don't provide any information about $\theta^{B}_{\Bar{J}}$. First, we proceed to rewrite the expression for the $\textnormal{COM}_i$ and $\textnormal{OPEN}_i$ subsystems by grouping the individual values of $\ccom_i$
\begin{align}
    \label{eq:factorcoms}
    &\frac{1}{2^{n_s}} \sum_{s_i} \ketbra{\com((\theta^{B}_i, x^{B}_i), s_i, r)}_{\textnormal{COM}_i} \ketbra{\open((\theta^{B}_i, x^{B}_i), s_i)}_{\textnormal{OPEN}_i} \nonumber \\
    =& \sum_{\substack{\ccom_i \\ \in \bs^{n_c}}}
    \!\!\!\!\! 
    P^{\theta_i,x_i,r}_{\textnormal{com}}(\ccom_i) \ketbra{\ccom_i}_{\textnormal{COM}_i}
    \!\!\!\!\!\!
    \sum_{\substack{\textnormal{open}_i \\ \in \mathcal{C}_r(\textnormal{com}_i)}}
    \!\!\!\!\!\!
    P^{\theta_i,x_i,r}_{\textnormal{open}}(\ccom_i,\oopen_i)\ketbra{\textnormal{open}_i}_{\textnormal{OPEN}_i} \nonumber\\
    =& \:\: \sigma^{\theta_i,x_i,r}_{\CCOM_i, \OOPEN_i},
\end{align}
where $P^{\theta_i,x_i,r}_{\textnormal{com}}$ is the respective distribution for $\ccom_i$ for uniformly sampled $s_i$, which depends on the commitment scheme used, and $\mathcal{C}_r(\textnormal{com}_i)$ is the set of strings $\textnormal{open}_i$ that satisfy $\ver(\textnormal{com}_i, \textnormal{open}_i, r) \neq \bot$. Substituting Eq.~\eqref{eq:factorcoms} into Eq.~\eqref{eq:postcom} and tracing out $\OOPEN_{\bar{J}}$ and $R$ results in
\begin{align}
    \label{eq:factorcoms2}
   \rho^{(2)} = \frac{1}{2^{N_0}}\sum_{x^{B}, \theta^{B}, r} |\beta^{x^{B}, \theta^{B}, r}|^2 \ketbra{\theta^{B}}_{\Theta^{B}} \ketbra{x^{B}}_{X^{B}} \ketbra{\phi^{x^{B}, \theta^{B}, r}}_{A} \underbrace{\bigotimes_{i\in J} \sigma^{\theta_i,x_i,r}_{\CCOM_i, \OOPEN_i}}_{\sigma^{\theta_J,x_J,r}_{\CCOM_J, \OOPEN_J}}
   \underbrace{\bigotimes_{i\in \bar{J}} \sigma^{\theta_i,x_i,r}_{\CCOM_i}}_{\sigma^{\theta_{\Bar{J}},x_{\Bar{J}}, r}_{\CCOM_{\Bar{J}}}}.
\end{align}
The hiding property of the commitment scheme states that for any fixed $r$, the distributions $P^{\theta_i,x_i,r}_{\textnormal{com}}$ are computationally indistinguishable among themselves. Let $P^{r}_{\ccom} = P^{0,0,r}_{\ccom}$, then 
\begin{equation}
    \label{eq:comhiding3}
    \sigma^{\theta_i,x_i, r}_{\CCOM_i} \approx^{(c)} \tilde{\sigma}^{r}_{\CCOM_i},
\end{equation}
with
\begin{equation}
    \tilde{\sigma}^{r}_{\CCOM_i} = \!\!\!\!\! 
    \sum_{\substack{\ccom_i \\ \in \bs^{n_c}}}
    \!\!\!\!\! 
    P^{r}_{\ccom}(\ccom_i) \ketbra{\ccom_i}_{\textnormal{COM}_i}.
\end{equation}
Applying Eq.~\eqref{eq:comhiding3} to the $\bar{J}$ subset in Eq.~\eqref{eq:factorcoms2}, and from Lemma~\ref{lem:cindistproperties} (2) and (3) we get that
\begin{equation}
    \label{eq:comhiding4}
    \rho^{(2)}_{\Theta^{B}, X^{B}, A, \CCOM, \OOPEN_{J}} \approx^{(c)} \tilde{\rho}^{(2)}_{\Theta^{B}, X^{B}, A, \CCOM, \OOPEN_{J}},
\end{equation}
where 
\begin{equation}
    \tilde{\rho}^{(2)} = \frac{1}{2^{N_0}}\sum_{x^{B}, \theta^{B}, r} |\beta^{x^{B}, \theta^{B}, r}|^2 \ketbra{\theta^{B}}_{\Theta^{B}} \ketbra{x^{B}}_{X^{B}} \ketbra{\phi^{x^{B}, \theta^{B}, r}}_{A} \sigma^{\theta_{J},x_{J}, r}_{\CCOM_{J}, \OOPEN_{J}} \tilde{\sigma}^{r}_{\CCOM_{\Bar{J}}}.
\end{equation}
Note that, since both $\sigma^{\theta_{J},x_{J}}_{\CCOM_{J}, \OOPEN_{J}}$ and $\tilde{\sigma}_{\CCOM_{\Bar{J}}}$ are independent of $\theta^{B}_{\Bar{J}}, x^{B}_{\Bar{J}}$, we can use Eq.\eqref{eq:basisindependence} with $I' = \Bar{J}$ such that, after tracing the $\Theta^{B}_{J},X^{B}$ subsystem, we obtain the state
\begin{equation}
    \Tilde{\rho}^{(2)}_{\Theta^{B}_{\bar{J}}, A, \CCOM, \OOPEN_{J}} = \mathsf{U}_{\Theta^{B}_{\bar{J}}} \otimes \Tilde{\rho}^{(2)}_{A, \CCOM, \OOPEN_{J}}.
\end{equation}
Finally, using Lemma~\ref{lem:cindistproperties} (1) and (2) we obtain the required result
\begin{equation}
    \rho^{(2)}_{\Theta^{B}_{\bar{J}}, A, \CCOM, \OOPEN_{J}} \approx^{(c)} \mathsf{U}_{\Theta^{B}_{\bar{J}}} \otimes \rho^{(2)}_{A, \CCOM, \OOPEN_{J}}.
\end{equation}

\end{proof}

\subsection{Proof of Lemma~\ref{lem:stringsepsecurity}}
\label{sec:stringseparationproof}
In this section we present a proof of Lemma~\ref{lem:stringsepsecurity}, which states that the string separation step of $\pi_{\textnormal{QROT}}$ does not leak any information about the random bit $c$ to the receiver.

\begin{lemma}
    \label{lem:stringsepsecurity2}
    Let $\mathcal{E}^{(I_t)}: \mathcal{H}_{A,\Theta^{A}_{\bar{I}_t},{\Theta^{B}_{\bar{I}_t}}, C} \rightarrow \mathcal{H}_{A,\Theta^{A}_{\bar{I}_t}, {\Theta^{B}_{\bar{I}_t}}, C, \textnormal{SEP}}$ be the quantum operation used by Bob to compute the string separation information $(J_0, J_1)$ during Step (9) of the protocol. The resulting state after applying $\mathcal{E}^{(I_t)}$ to a product state of the form
    \begin{equation}
        \label{eq:productform}
        \mathcal{E}^{(I_t)}(\rho_{A,\Theta^{A}_{\bar{I}_t}} \otimes \mathsf{U}_{\Theta^{B}_{\bar{I}_t}} \otimes \mathsf{U}_{C}) = \sigma_{A,\Theta^{A}_{\bar{I}_t}, {\Theta^{B}_{\bar{I}_t}}, C, \textnormal{SEP}}
    \end{equation}
    satisfies
    \begin{equation}
    \label{eq:stringindependence}
        \Tr_{\Theta^{A}_{\bar{I}_t},\Theta^{B}_{\bar{I}_t}} \big[ \sigma_{A,\Theta^{A}_{\bar{I}_t}, {\Theta^{B}_{\bar{I}_t}}, C, \textnormal{SEP}} \big] =  \sigma_{A} \otimes \sigma_{\textnormal{SEP}} \otimes \mathsf{U}_{C}.
    \end{equation}
\end{lemma}

\begin{proof}
Let $\theta^{\textnormal{ch}} =  \theta^{A} \oplus \theta^{B}$ and, for $b \in \bs$, define the sets $S_{b} = \lbrace i \in \bar{I}_t \: \vert \:  \theta^{\textnormal{ch}}_i = b \rbrace$. Bob's operation consists on randomly choosing subsets $I_0, I_1$ of size $N_{\textnormal{raw}}$, from $S_0$ and $S_1$, respectively, and then computing $J_0 = I_c, J_1 = I_{\Bar{c}}$. Denote by $N_1 = (1-\alpha)N_0$ the size of the working set $\Bar{I}_t$, so that $N_{\textnormal{raw}} = (\frac{1}{2}-\delta_2) N_1$. If the number of matching bases in $\bar{I}_t$, given by the Hamming weight $w_{\textnormal{H}}(\theta^{\textnormal{ch}}_{\bar{I}_{t}})$, is either smaller than $N_{\textnormal{raw}}$ or greater than $N_1 - N_{\textnormal{raw}}$, Bob won't be able to construct either $I_0$ or $I_1$, in which case he sends an abort message to Alice independently of the value of $c$ and Eq.~\eqref{eq:stringindependence} is satisfied. On the other hand, we will show that whenever $N_{\textnormal{raw}} \leq w_{\textnormal{H}}(\theta^{\textnormal{ch}}_{\bar{I}_{t}}) \leq N_1 - N_{\textnormal{raw}}$, the probability of choosing $I_0, I_1$ is the same for every $I_0, I_1$. 
Define 
\begin{equation}
    \label{eq:sepadmissible}
    C(I_0, I_1)= \lbrace \theta^{\textnormal{ch}} \: : \: \forall i\in I_0 \forall j \in I_1 (\theta^{\textnormal{ch}}_i = 0 \And \theta^{\textnormal{ch}}_j = 1) \rbrace.
\end{equation}
Note that, because for any two pairs $(I^{1}_0, I^{1}_1), (I^{2}_0, I^{2}_1)$ the elements of $C(I^{1}_0, I^{1}_1)$ and $C(I^{2}_0, I^{2}_1)$ are related to each other through a permutation of indices, the size of the $C(I_0, I_1)$ is independent of $I_0, I_1$. The probability of Bob choosing $I_0, I_1$ is then given by
\begin{align}
    \label{eq:sepprobability}
    P(I_{0},I_{1}) &= 
    \!\!\!\!\!\!\!
    \sum_{\theta^{B}_{\bar{I}_t} \in C(I_0, I_1)} 
    \!\!\!\!\!\!\! 
    P(I_{0},I_{1} \: \vert \: \theta^{B}_{\bar{I}_t})P(\theta^{B}_{\bar{I}_t}) \nonumber \\
    &= \sum_{n=N_{\textnormal{raw}}}^{N_0 - N_{\textnormal{raw}}} \sum_{\substack{w_{H}(\theta^{\textnormal{ch}}_{\bar{I}_{t}}) = n \\ \theta^{\textnormal{ch}} \in C(I_{0}, I_{1})}} \binom{n}{N_{\textnormal{raw}}}^{-1} \binom{N_1 - n}{N_{\textnormal{raw}}}^{-1} P(\theta^{\textnormal{ch}}) \nonumber \\
    &= \sum_{n=N_{\textnormal{raw}}}^{N_0 - N_{\textnormal{raw}}} \sum_{\substack{w_{H}(\theta^{\textnormal{ch}}_{\bar{I}_{t}}) = n \\ \theta^{\textnormal{ch}} \in C(I_{0}, I_{1})}} \binom{n}{N_{\textnormal{raw}}}^{-1} \binom{N_1 - n}{N_{\textnormal{raw}}}^{-1} 2^{-N_{1}} \nonumber \\
    &= P^{\textnormal{SEP}},
\end{align}
where the combinatorial factors come from the fact that, for each $\theta^{\textnormal{ch}}$, the $I_0, I_1$ are chosen uniformly among all available compatible combinations, and the $2^{-N_{1}}$ factor comes from the fact that both $\theta^{A}$ and $\theta^{B}$ are sampled independently and $\theta^{B}$ is sampled uniformly (as guaranteed by the product form Eq.~\eqref{eq:productform}), and the last equality comes from the fact that the number of elements in $C(I_0, I_1)$ is constant, and hence the number of terms in the summation is the same for every $(I_0, I_1)$. Importantly, note that $P^{\textnormal{SEP}}$ is independent of $(I_0, I_1)$. To obtain Eq.~\eqref{eq:stringindependence} we start by computing 
\begin{align}
    \sigma_{\textnormal{SEP}, {\Theta^{B}_{\bar{I}_t}}, C} &= \mathcal{E}^{(I_t, \Theta^{A}_{\Bar{I}_t})}(\mathsf{U}_{\Theta^{B}_{\bar{I}_t}} \otimes \mathsf{U}_{C}) \nonumber \\
    &= \frac{1}{2^{N_1}}
    \sum_{\theta^{B}_{\bar{I}_t}}
    \ketbra{\theta^{B}_{\bar{I}_t}}_{\Theta^{B}_{\bar{I}_t}}\otimes \frac{1}{2} \sum_{c}\ketbra{c}_{C} \otimes \sum_{I_{0},I_{1}} P(I_{0},I_{1} \: \vert \: \theta^{B}_{\bar{I}_t}) \ketbra{I_{c}, I_{\bar{c}}}_{\textnormal{SEP}} \nonumber \\
    &= \frac{1}{2} \sum_{c}\ketbra{c}_{C} \otimes \sum_{I_{0},I_{1}} P(I_{0},I_{1}) \ketbra{I_{c}, I_{\bar{c}}}_{\textnormal{SEP}} \otimes 
    \!\!\!\!\!\!\!\!
    \sum_{\theta^{B}_{\bar{I}_t} \in C(I_0, I_1)} 
    \!\!\!\!\!\!\!\!
    P(\theta^{B}_{\bar{I}_t} \: \vert \: I_{0},I_{1}) \ketbra{\theta^{B}_{\bar{I}_t}}_{\Theta^{B}_{\bar{I}_t}},
\end{align}
where the sum in $\textnormal{SEP}$ goes over all possible $I_0,I_1$ given $I_t$. Tracing out $\Theta^{B}_{\bar{I}_t}$ and using Eq.~\eqref{eq:sepprobability} we obtain
\begin{align}
    \Tr_{\Theta^{B}_{\bar{I}_t}} \big[ \sigma_{\textnormal{SEP},  {\Theta^{B}_{\bar{I}_t}}, C} \big] &= \frac{1}{2} \sum_{c}\ketbra{c}_{C} \otimes \sum_{I_{0},I_{1}} P^{\textnormal{SEP}} \ketbra{I_{c}, I_{\bar{c}}}_{\textnormal{SEP}} \nonumber \\
    &= \frac{1}{2} \sum_{c}\ketbra{c}_{C} \otimes \sum_{I_{0},I_{1}} P^{\textnormal{SEP}}\ketbra{I_{0}, I_{1}}_{\textnormal{SEP}} \nonumber \\
    &=  \mathsf{U}_{\textnormal{SEP}}^{I_t} \otimes \mathsf{U}_{C}.
\end{align}

\end{proof}

\subsection{Proof of Lemma~\ref{lem:postmeasuremententropy}}
\label{sec:postmeasureentropyproof}
In this section we present a proof of Lemma~\ref{lem:postmeasuremententropy}, introduced in Section~\ref{sec:secanalysis} as part of the security analysis against a dishonest receiver. Recall that the transcript of the protocol $\Vec{\tau} = (x^{A}_{I_t},\theta^{A},r,\ccom,I_t,I_s, \oopen_{I_t}, r)$ is defined to consist of all classical information (with the exception of her measurement outcomes) that Alice has access up to Step (8) of the protocol. 

\begin{lemma}
    \label{lem:postmeasuremententropy2}
      Assuming Alice follows the protocol, let $X^{A}, B$ denote the systems of Alice measurement outcomes and Bob's laboratory at the end of Step (9) of the protocol, and let $\rho_{X^{A}, B}$ be the state of the joint system at that point. There exists a state $\tilde{\rho}_{X^{A}, B}$, such as:
      \begin{enumerate}
        \item The conditioned states $\tilde{\rho}_{X^{A}, B}(\Vec{\tau}, J_0, J_1)$ satisfy
        \begin{align}
                H_{\textnormal{min}} (X^{A}_{J_0}|X^{A}_{J_1} B)_{\tilde{\rho}(\Vec{\tau}, J_0, J_1)} + H_{\textnormal{min}} (X^{A}_{J_1}|X^{A}_{J_0} B)_{\tilde{\rho}(\Vec{\tau}, J_0, J_1)} \geq 2 N_{\text{raw}} \left( \frac{1}{2} - \frac{2\delta_2}{1-2\delta_2} - h\left( \frac{p_{\text{max}} + \delta_1}{\frac{1}{2} - \delta_2}\right) \right);
        \end{align}
        \item $\rho_{X^{A}, B} \approx_{\varepsilon} \tilde{\rho}_{X^{A}, B}$, with 
        \begin{equation}
        \label{eq:epsilonsecuritydishonestb2}
        \varepsilon = \left(2 (e^{-\frac{1}{2} \alpha (1-\alpha)^{2} N_0 \delta^{2}_1} +  e^{-\frac{1}{2} (\frac{1}{2} -\delta_2)\alpha N_0 \delta^{2}_1}) \right)^{\frac{1}{2}} + e^{-D_{KL}(\frac{1}{2} - \delta_2|\frac{1}{2})(1-\alpha) N_0} + \varepsilon_{\textnormal{bind}}(k),
        \end{equation}
        where $h(\cdot)$ and $D_{KL}(\cdot | \cdot)$ denote the binary entropy and the binary relative entropy functions, respectively, and $\varepsilon_{\textnormal{bind}}(k)$ is a negligible function given by the binding property of the commitment scheme.
    \end{enumerate}
\end{lemma}

\begin{proof}
We proceed by tracking the properties of Alice's and Bob's shared state as the protocol develops in order to bound the conditional min-entropy of Alice's measurement outcomes given the information the Bob gains during the protocol, then we use Lemma~\ref{lem:qlhl} to obtain the desired result. Let $\rho^{\textnormal{rand}}_{\Theta^{A}, T, R}$ denote the quantum state associated to the systems holding Alice's basis choice $\theta^{A}$, test subset $I_t$, and the value of $r$ used in the commit/open phase, which we can treat as if they are sampled at the beginning of the protocol since their distribution is fixed, and is given by
\begin{equation}
    \rho^{\textnormal{rand}}_{\Theta^{A}, T, R} = \frac{1}{2^{N_0}}\sum_{\theta^{A}}\ketbra{\theta^{A}}_{\Theta^{A}} \otimes \frac{1}{|\mathcal{T}(\alpha N_0, I)|} \sum_{I_t}\ketbra{I_t}_{T} \otimes \frac{1}{2^{n_r}} \sum_{r}\ketbra{r}_{R},
\end{equation}
where $\mathcal{T}(\alpha N_0, I)$ denotes the set of subsets of $I = \lbrace 1,2,...,N_0\rbrace$ with $\alpha N_0$ elements. Let $S_{\textnormal{bind}}(k)$ be the set of all $r \in \bs^{n_r(k)}$ for which there exists a tuple $(\ccom, \oopen_1, \oopen_2)$ such that
\begin{equation}  
    \bot \neq \ver(\ccom, \oopen_1, r) \neq \ver(\ccom, \oopen_2, r) \neq \bot. 
\end{equation} 
From the binding property of the commitment scheme we know that there exists a negligible function $\varepsilon_{\textnormal{bind}}(k)$ such that, for a commitment security parameter $k$ it holds that
\begin{equation}
    \frac{S_{\textnormal{bind}}(k)}{2^{n_r(k)}} = \varepsilon_{\textnormal{bind}}(k), 
\end{equation}
and hence the state
\begin{equation}
    \tilde{\rho}^{\textnormal{rand}}_{\Theta^{A}, T, R} = \frac{1}{2^{N_0}}
    \sum_{\theta^{A}}\ketbra{\theta^{A}}_{\Theta^{A}} \otimes \frac{1}{|\mathcal{T}(\alpha N_0, I)|} \sum_{I_t}\ketbra{I_t}_{T} \otimes P_R 
    \!\!\!
    \sum_{r \in \bar{S}_{\textnormal{bind}}}
    \!\!\!
    \ketbra{r}_{R},
\end{equation}
where $P_R = \frac{1}{2^{n_r} - |S_{\textnormal{bind}}|}$, satisfies
\begin{equation}
    \rho^{\textnormal{rand}}_{\Theta^{A}, T, R} \approx_{\varepsilon_{\textnormal{bind}}(k)}  \tilde{\rho}^{\textnormal{rand}}_{\Theta^{A}, T, R}.
\end{equation}
In other words, the state of the system holding the value of the variable $r$ is indistinguishable to one where the commitment scheme is perfectly binding (for all $\ccom$ strings, there is at most one $\oopen$ string that passes verification).  

Additionally, the state of the shared resource system as after Bob receives his shares at the beginning of the protocol is given by:
\begin{equation}
    \label{eq:cheatbobinit}
    \rho^{(0)} = \tilde{\rho}^{\textnormal{rand}}_{\Theta^{A}, T, R} \otimes \ketbra{\psi^{(0)}} \quad \quad \quad \ket{\psi^{(0)}} = \frac{1}{\sqrt{2^{N_0}}} \sum_{x} \ket{x}_{\Phi^{A}} \ket{x}_{\Phi^{B}}.
\end{equation}

Since the measurement on Alice subsystem is performed independently from Bob's actions, we can equivalently consider a version of the protocol in which Alice doesn't measure her side of the shared resource state until it's needed to perform the check at Step (7) (for the indices in $I_t$) and the computation of the syndromes at Step (10) (for the remaining indices). 

We now turn our attention to Step (4), when Bob computes and sends his commitment strings after receiving the value of $r$. Denote by $B_0$ the system containing all of Bob's laboratory at the beginning of the protocol, and let $U_1$ be the transformation that Bob performs on his system to produce the commitments, which has the general form
\begin{align}
    U_1 \ket{r}_{R}\ket{x}_{\Phi^{B}}\ket{0}_{B_0} = 
    \sum_{\ccom} \alpha^{r, x, \ccom} \ket{\ccom}_{\CCOM} \ket{\phi^{r, x, \ccom}}_{B_1},
\end{align}
where $\mathcal{H}_{R} \otimes \mathcal{H}_{\Phi^{B}} \otimes \mathcal{H}_{B_0} = \mathcal{H}_{\CCOM} \otimes \mathcal{H}_{B_1}$, and $\ccom = (\ccom_1, \ccom_2, \ldots, \ccom_{N_0})$ with $\ccom_i \in \bs^{n_c(k)}$. Bob then proceeds to send the COM system to Alice, who measures it in the computational basis. The joint shared state as Bob sends the commitment information is
\begin{align}
    \rho^{(1)} = \frac{1}{2^{N_0} |\mathcal{T}(\alpha N_0, I)|}\sum_{\theta^{A}}\ketbra{\theta^{A}}_{\Theta^{A}} \sum_{I_t}\ketbra{I_t}_{T} 
    \!\!\!
    \sum_{\substack{r \in \bar{S}_{\textnormal{bind}} \\ \ccom}} 
    \!\!\!
    P^{r}_{\ccom} \ketbra{\ccom}_{\CCOM} \ketbra{\eta^{r, \ccom}}_{\Phi^{A} B_1},
\end{align}
where
\begin{align}
     P^{r}_{\ccom} &= \frac{P_R}{2^{N_0}}\sum_{x}|\alpha^{r, x, \ccom}|^{2} \nonumber \\
    \ket{\eta^{r, \ccom}}_{\Phi^{A} B_1} &= \sum_{x} \underbrace{\sqrt{\frac{P_R}{2^{N_0}}}(P_{\ccom}^{r})^{-\frac{1}{2}} \alpha^{r, x, \ccom}}_{\beta^{r, x,\ccom}}\ket{x}_{\Phi^{A}} \ket{\phi^{r, x, \ccom}}_{B_1}.
\end{align}
We intend to use Lemma~\ref{lem:testsecurity} to bound the form of the shared state after the parameter estimation step, and then Lemma~\ref{lem:minentropybound} to bound the amount of correlation between Alice's measurement outcomes on the system $\Phi^{A}$ and Bob's system. For that, we first need to associate Bob's commitments with their corresponding committed strings $x^{B}$ and $\theta^{B}$. For an arbitrary dishonest Bob the strings that Alice received are not guaranteed to be outcomes of the $\com$ function and may not have an associated preimage. Consider now the functions $x^{B}_{i}(r, \ccom), \theta^{B}_{i}(r, \ccom):\bs^{n_r} \times \bs^{n_{c}} \rightarrow \bs$ defined as follows,
\begin{align}
    x^{B}_{i}(r, \ccom) = \begin{cases} x &\quad \textnormal{if } \ccom_i = \com((\theta, x), s,r) \textnormal{ for some } \theta \in \bs, s \in \bs^{n_s} \\ 
    0 &\quad \textnormal{otherwise }
    \end{cases}
\end{align}
\begin{align}
    \theta^{B}_{i}(r, \ccom) = \begin{cases} \theta &\quad \textnormal{if } \ccom_i = \com((\theta, x), s,r) \textnormal{ for some } x \in \bs, s \in \bs^{n_s} \\ 
    0 &\quad \textnormal{otherwise }
    \end{cases},
\end{align}
and denote 
\begin{align}
    x^{B}(r, \ccom) = (x^{B}_{i}(r, \ccom))_{i}, \quad \theta^{B}(r, \ccom) = (\theta^{B}_{i}(r, \ccom))_{i}.
\end{align}
We know the above functions are well defined for all $r \in \bar{S}_{\textnormal{bind}}$ because, by definition of $S_{\textnormal{bind}}$, for each possible value of $\ccom_i$, there is at most a single opening that passes verification. For any $\theta^{B} \in \bs^{N_0}$ we can write the state $\ket{\eta^{r, \ccom}}_{\Phi^{A} B_1}$ in the $\theta^{B}$ basis of $\Phi^{A}$ as
\begin{align}
    \ket{\eta^{r, \ccom}}_{\Phi^{A} B_1} = \sum_{x} \underbrace{\sum_{x'}\beta^{r, x,\ccom} \bracket{x,\theta^{B}}{x'}\ket{\phi^{r, x', \ccom}}_{B_1}}_{\beta^{r, x,\theta^{B},\ccom}\ket{\phi^{r, x,\theta^{B},\ccom}}} \ket{x,\theta ^{B}}_{\Phi^{A}}.
\end{align}
Recall that $I_s(\theta^{A}, \theta^{B}, I_t) = \lbrace i \in I_t : \theta^{A}_i\oplus\theta^{B}_i = 0 \rbrace$. From Lemma~\ref{lem:testsecurity} we know that there exists a state 
\begin{equation}
    \tilde{\rho}^{(1)} = \frac{1}{2^{N_0} |\mathcal{T}(\alpha N_0, I)|} 
    \!
    \sum_{\substack{r \in \bar{S}_{\textnormal{bind}} \\ \ccom}} 
    \!\!\!
    P^{r}_{\ccom}
    \ketbra{\ccom}_{\CCOM} \sum_{\theta^{A}, I_t} \ketbra{\theta^{A}}_{\Theta^{A}} \ketbra{I_t}_{T} \ketbra{\eta^{r, \ccom,I_s,\bar{I}_t}}_{\Phi^{A} B_1},
\end{equation}
where the $\ket{\eta^{r, \ccom,I_s,\bar{I}_t}}$ have the form
\begin{align}
    \ket{\eta^{r, \ccom,I_s,\bar{I}_t}}_{\Phi^{A} B_1} = 
    \!\!\!\!\!\!\!\!\!
    \sum_{\substack{x \in \\ B(\theta^{B},r,\ccom,I_s,I_t)}}
    \!\!\!\!\!\!\!\!\!
    \tilde{\beta}^{x,\theta^{B},r,\ccom,I_s,I_t} \ket{x, \theta^{B}(r,\ccom)}_{\Phi^{A}}  \ket{\phi^{x,\theta^{B},r,\ccom,I_s,I_t}}_{B_1}
\end{align}
\begin{equation}
    B(r,\ccom,I_s,I_t) = \lbrace x : |r_{\textnormal{H}}(x_{I_s}\oplus x^{B}_{I_s}(r,\ccom)) - r_{\textnormal{H}}(x_{\bar{I_t}} \oplus x^{B}_{\bar{I_t}} (r,\ccom))| \leq \delta_1 \rbrace,
\end{equation}
such that
\begin{align}
    \label{eq:postmeasurementdistanceB}
    D(\rho^{(1)}, \tilde{\rho}^{(1)}) &\leq \sqrt{2} \left( e^{-\frac{1}{2} \alpha (1-\alpha)^{2} N_0 \delta^{2}_1} +  e^{-\frac{1}{2} (\frac{1}{2} -\delta_2)\alpha N_0 \delta^{2}_1} \right)^{\frac{1}{2}}.
\end{align}
We are ready now proceed to Step (6) of the protocol, in which Bob sends the string  $\oopen_{I_t} =(\textnormal{open}_i)_{i \in I_t}$, which is expected to contain the opening information for all the commitments $\ccom_i, i\in I_t$.  
\begin{equation}
    U_{\oopen} \ket{\phi^{x,r,\ccom,I_s,I_t}}_{B_1} = \sum_{\oopen_{I_t}} \alpha^{x,r,\ccom,I_s,I_t, \oopen_{I_t}} \ket{\phi^{x,r,\ccom,I_s,I_t, \oopen_{I_t}}}_{B_2} \ket{\oopen_{I_t}}_{\OOPEN_{I_t}},
\end{equation}
where $\mathcal{H}_{B_1} = \mathcal{H}_{B_2}\otimes\mathcal{H}_{\OOPEN_{I_t}}$. Such that 
\begin{align}
    U_{\oopen}\ket{\eta^{r,\ccom,I_s,\bar{I}_t}}_{\Phi^{A} B_1} = &\sum_{\oopen_{I_t}} 
    \sum_{x \in B} \tilde{\beta}^{x,r,\ccom,I_s,I_t} \alpha^{x,r,\ccom,I_s,I_t, \oopen_{I_t}} \ket{x, \theta^{B}(r, \ccom)}_{\Phi^{A}} \nonumber \\
    &\otimes \ket{\phi^{x,r,\ccom,I_s,I_t, \oopen_{I_t}}}_{B_2} \ket{\oopen_{I_t}}_{\OOPEN_{I_t}}.
\end{align}
During Step (7), after receiving the opening information and measuring the $\OOPEN$ system in the computational basis, she aborts the protocol unless $\ver(\ccom_i,\oopen_i,r) \neq \bot$ for all $i \in I_t$. Let $ H(r, I_t)$ be the set of strings $\ccom$ for which Alice's first check {\em can} be passed. From the binding property of the commitment scheme, we know that, for any $r \in \bar{S}_{\textnormal{bind}}$, if $\ccom \in H(r, I_t)$ there is only one $\oopen'(r,\ccom, I_t)$ for which $\ver(\ccom_i,\oopen'_i, r) \neq \bot$ for all $i \in I_t$. Because the protocol aborts if Alice's test is not passed, the state of the joint system after Alice performs this check is given by (note that from here, by removing the mixture over all opens, we are reducing the overall trace of the system. Effectively, we are keeping only the runs of the protocol that did not abort in the commitment check part of Step (7). The amount for which the trace is reduced is given by the sum of the $|\alpha^{x,r,\ccom,I_s,I_t, \oopen_{I_t}}|^2$ over the values of $\oopen \neq \oopen'(\ccom, I_t)$ or for which $I_s < N_{\textnormal{check}}$):
\begin{align}
    \tilde{\rho}^{(2)} = \frac{1}{2^{N_0} |\mathcal{T}(\alpha N_0, I)|} &\sum_{I_t, \theta^{A}}\ketbra{I_t}_{T} 
    \!\!\!\!\!\!\!\!
    \sum_{\substack{r \in \bar{S}_{\textnormal{bind}} \\ \ccom \in H(r,I_t)}}
    \!\!\!\!\!\!\!\!
    P^{r}_{\ccom} P^{r}_{\oopen'} \ketbra{\ccom}_{\CCOM} \ketbra{\oopen'}_{\OOPEN_{I_t}} \nonumber \\
    \otimes &\ketbra{\theta^{A}}_{\Theta^{A}} \ketbra{\tilde{\eta}^{r,\ccom,I_s,\bar{I}_t}}_{\Phi^{A} B_2},
\end{align}
with
\begin{align}
    P^{r}_{\oopen'} = \!\!\!\!\!\!\!\!\!
    \sum_{\substack{x \in \\ B(r,\ccom,I_s,I_t)}}
    \!\!\!\!\!\!\!\!\!
    |\tilde{\beta}^{x,r,\ccom,I_s,I_t} \alpha^{x,r,\ccom,I_s,I_t, \oopen'}|^2 \nonumber
\end{align}
and
\begin{align}
 \!\!\!\!\!\!   \ket{\tilde{\eta}^{r,\ccom,I_s,\bar{I}_t}} = 
    \!\!\!\!\!\!\!\!\!\!\!\!
    \sum_{\substack{x \in \\ B(r,\ccom,I_s,I_t)}}
    \!\!\!\!\!\!\!\!\!\!\!\!
    \Big( \underbrace{(P^{r}_{\oopen'})^{-\frac{1}{2}} \tilde{\beta}^{x,r,\ccom,I_s,I_t} \alpha^{x,r,\ccom,I_s,I_t, \oopen'} \ket{\phi^{x,r,\ccom,I_s,I_t, \oopen'}}_{B_2}}_{\gamma^{x,r,\ccom,I_s,I_t} \ket{\tilde{\phi}^{x,r,\ccom,I_s,I_t}}} 
     \otimes \ket{x, \theta^{B}(r, \ccom)}_{\Phi^{A}}\Big).
\end{align}
Alice then proceeds to measure her part of the state. Let us divide her measurement in two parts: the measurement of the qubits in $I_t$, and the measurement of the reminder qubits. For the first part, the action of measuring the subsystem $\Phi^{A}_{I_t}$ in a state $\ket{\tilde{\eta}^{\ccom,I_s,\bar{I}_t}}$ and in the $\theta^{A}_{I_t}$ basis is:
\begin{align}
    \ketbra{\tilde{\eta}^{r,\ccom,I_s,\bar{I}_t}} &\rightarrow \sum_{x^{A}_{I_t}} \ketbra{x^{A}_{I_t}, \theta^{A}_{I_t}} \left(\ketbra{\tilde{\eta}^{r,\ccom,I_s,\bar{I}_t}}\right) \ketbra{x^{A}_{I_t}, \theta^{A}_{I_t}} \nonumber \\
    &= \sum_{x^{A}_{I_t}} P_{x^{A}_{I_t}}^{r} \ketbra{x^{A}_{I_t}, \theta^{A}_{I_t}}_{\Phi^{A}_{I_t}} \ketbra{\tilde{\eta}^{r,\ccom,I_s,\bar{I}_t,\theta^{A},x^{A}_{I_t}}}_{\Phi^{A}_{\bar{I}_t} B_2},
\end{align}
where
\begin{equation}
    P_{x^{A}_{I_t}}^{r} = \sum_{x \in B} \big|\bracket{x_{I_t}, \theta^{B}_{I_t}(r,\ccom)}{x^{A}_{I_t},\theta^{A}_{I_t}} \gamma^{x,r,\ccom,I_s,I_t}\big|^{2}
\end{equation}
and
\begin{align}
    \ket{\tilde{\eta}^{r,\ccom,I_s,\bar{I}_t,x^{A}_{I_t},\theta^{A}}} = 
    \sum_{x \in B}
    &(P_{x^{A}_{I_t}}^{r})^{-\frac{1}{2}} \underbrace{\bracket{x_{I_s}, \theta^{B}_{I_s}}{x^{A}_{I_s},\theta^{B}_{I_s}}}_{\delta(x_{I_s}, x^{A}_{I_s})} \bracket{x_{{I_t \setminus I_s}}, \theta^{B}_{{I_t \setminus I_s}}}{x^{A}_{{I_t \setminus I_s}},\bar{\theta}^{B}_{{I_t \setminus I_s}}} \nonumber \\ 
    &\times \gamma^{x,\theta^{B},r,\ccom,I_s,I_t} \ket{\tilde{\phi}^{x,\theta^{A}, r, \ccom,I_t}}_{B_2} \ket{x_{\bar{I}_t}, \theta^{B}_{\bar{I}_t}}_{\Phi^{A}_{\bar{I}_t}},
\end{align}
where in the last expression, and going forward, we omit the explicit dependence of both $x^{B}$ and $\theta^{B}$ on $r, \ccom$. By defining
\begin{equation}
    G(x^{A}_{I_s}, r, \ccom) = \lbrace x_{\bar{I}_t} : |r_{\textnormal{H}}(x^{A}_{I_s} \oplus x^{B}_{I_s}) - r_{\textnormal{H}}(x_{\bar{I}_t} \oplus x^{B}_{\bar{I}_t}) | \leq \delta_1 \rbrace,
\end{equation}
we can rewrite
\begin{align}
    \ket{\tilde{\eta}^{\ccom,I_s,\bar{I}_t,x^{A}_{I_t},\theta^{A}}} = 
    \sum_{x_{\bar{I}_t} \in G} 
    &\underbrace{\sum_{x_{{I_t \setminus I_s}}} (P_{x^{A}_{I_t}}^{r})^{-\frac{1}{2}} \bracket{x_{{I_t \setminus I_s}}, \theta^{B}_{{I_t \setminus I_s}}}
    {x^{A}_{I_s},\theta^{B}_{I_s}}
    \gamma^{x_{\bar{I}_t}, x^{A}_{I_t},r,\ccom,I_s,I_t} \ket{\tilde{\phi}^{x_{\bar{I}_t}, x^{A}_{I_t},r,\ccom,I_s,I_t}}_{B_2}}_{\tilde{\gamma}^{x_{\bar{I}_t},r,\ccom,I_s,I_t} \ket{\tilde{\phi}^{x_{\bar{I}_t},r,\ccom,I_s,I_t,\theta^{A}}}} \nonumber \\
    &\otimes \ket{x_{\bar{I}_t}, \theta^{B}_{\bar{I}_t}}_{\Phi^{A}_{\bar{I}_t}} 
\end{align}
After performing the measurement, Alice aborts the protocol whenever $r_{\textnormal{H}}(x^{A}_{I_s} \oplus x^{B}_{I_s}(r,\ccom)) > p_{\text{max}}$. The state of the 
shared system after this check is (tracing out the $T, \CCOM, \OOPEN$ subsystems)
\begin{align}
    \label{eq:limitmismatchbases}
    \tilde{\rho}^{(3)} = & \frac{1}{2^{N_0} |\mathcal{T}(\alpha N_0, I)|} \sum_{I_t} 
    \!\!\!\!\!
    \sum_{\substack{r \in \bar{S}_{\textnormal{bind}} \\ \ccom \in H(r,I_t)}}
    \!\!\!\!\!
    P_{\ccom}^{r} P_{\oopen'}^{r} \sum_{\theta^{A}}\ketbra{\theta^{A}}_{\Theta^{A}} \nonumber \\
    &\otimes 
    \!\!\! 
    \sum_{\substack{x^{A}_{I_t} \in \\ J_{p_{\text{max}}}}}
    \!\!\!
    P_{x^{A}_{I_t}}^r \ketbra{x^{A}_{I_t}, \theta^{A}_{I_t}}_{\Phi^{A}_{I_t}} \ketbra{\tilde{\eta}^{r,\ccom,I_s,\bar{I}_t,x^{A}_{I_t},\theta^{A}}}_{\Phi^{A}_{\bar{I}_t} B_2},
\end{align}
where
\begin{equation}
    J_{p_{\text{max}}} = \lbrace x^{A}_{I_t} : r_{\textnormal{H}}(x^{A}_{I_s} \oplus x^{B}_{I_s}) \leq p_{\text{max}}\rbrace.
\end{equation}
Before proceeding, it will be useful to approximate the above state to a state where the number of mismatching bases in $\bar{I}_t$ is ``high enough''. More precisely, this means approximating $\tilde{\rho}^{(3)}$ to a state for which the sum over $\theta^{A}$ runs explicitly over strings  $ \theta^{A} \in K_{\ccom} = \lbrace \theta^{A} : w_{\textnormal{H}}(\theta^{\textnormal{ch}}_{{\bar{I}_{t}}}) \geq N_{\textnormal{raw}} \rbrace$.  Let 
\begin{align}
    \label{eq:postalicecheck2}
    \tilde{\tilde{\rho}}^{(3)} = & \frac{1}{2^{N_0} |\mathcal{T}(\alpha N_0, I)|} \sum_{I_t}
    \!\!\!\!
    \sum_{\substack{r \in \bar{S}_{\textnormal{bind}} \\ \ccom \in H(r,I_t)}}
    \!\!\!\!\!\!\!\!
    P_{\ccom}^{r} P_{\oopen'}^{r} a(\ccom) 
    \!\!\!
    \sum_{\theta^{A} \in K_{\ccom}}
    \!\!\!
    \ketbra{\theta^{A}}_{\Theta^{A}} \nonumber \\
    &\otimes 
    \!\!\! 
    \sum_{\substack{x^{A}_{I_t} \in \\ J_{p_{\text{max}}}}}
    \!\!\!
    P_{x^{A}_{I_t}}^{r} \ketbra{x^{A}_{I_t}, \theta^{A}_{I_t}}_{\Phi^{A}_{I_t}} \ketbra{\tilde{\eta}^{r,\ccom,I_s,\bar{I}_t,x^{A}_{I_t}}}_{\Phi^{A}_{\bar{I}_t} B_2},
\end{align}
with $a(\ccom) = \frac{2^{N_0}}{|K_{\ccom}|}$. The distance between $\tilde{\rho}^{(3)}$ and $\tilde{\tilde{\rho}}^{(3)}$ is bounded by the probability of a uniformly chosen the $\theta^{A}$ not being in $K_{\ccom}$. Using the Chernoff-Hoeffding bound we get
\begin{equation}
    \label{eq:chernoffdistance}
    D(\tilde{\rho}^{(3)}, \tilde{\tilde{\rho}}^{(3)}) \leq e^{-D_{KL}(\frac{1}{2} - \delta_2 | \frac{1}{2})(1-\alpha)N_{0}},
\end{equation}
where $D_{KL}(\frac{1}{2} - \delta_2 | \frac{1}{2})$ represents the relative entropy between the binary distributions defined by the respective probabilities $p_1 = \frac{1}{2} - \delta_2$ and $p_2 = \frac{1}{2}$. 

During Step (8), Alice sends the $\Theta^{A}$ system to Bob, who then computes $J_0, J_1$ (in the actual protocol, Alice sends only $\Theta^{A}_{{\bar{I_t}}}$, but to simplify the expressions we can assume, without loss of generality, that she sends the whole register $\Theta^{A}$). To simplify the list of dependencies, denote the transcript of the protocol up until Step (8) as $\Vec{\tau} = (x^{A}_{I_t},\theta^{A},r,\ccom,I_t,I_s, \oopen_{I_t})$. Keep in mind that, although $\Vec{\tau}$ consists of seven quantities, $I_s$ and $\oopen_{I_t}$ are completely defined by the other five. In the remaining of the proof, unless noted otherwise, the sums over $\Vec{\tau}$ run over the values of its variables as shown in Eq.~\eqref{eq:postalicecheck2}. By defining 
\begin{equation}
    P_{\Vec{\tau}} = \frac{P_{\ccom}^{r} P_{\oopen'}^{r}  P_{x^{A}_{I_t}}^{r} a(\ccom) }{2^{N_0} |\mathcal{T}(\alpha N_0, I)|},
\end{equation}
we can write
\begin{equation}
    \tilde{\tilde{\rho}}^{(3)} = \sum_{\Vec{\tau}}  P_{\Vec{\tau}} \ketbra{\theta^{A}}_{\Theta^{A}} \ketbra{x^{A}_{I_t}, \theta^{A}_{I_t}}_{\Phi^{A}_{I_t}} \ketbra{\tilde{\eta}^{r,\ccom,I_s,\bar{I}_t,x^{A}_{I_t}}}_{\Phi^{A}_{\bar{I}_t} B_2}.
\end{equation}
During Step (9), after receiving $\theta^{A}$, Bob sends the SEP system, containing the (classical) string separation information $J_0, J_1$ to Alice. By following the same treatment as in Steps (4) and (6), let $U_{\textnormal{sep}}$ be the operation that Bob performs on the $B_2$ system to compute the information to be sent to Alice in the $\textnormal{SEP}$ system: 
\begin{equation}
    U_{\textnormal{sep}}\ket{\tilde{\phi}^{x_{\bar{I}_t},r,\ccom,I_s,I_t,\theta^{A}}}_{B_2}\ket{\theta^{A}}_{\Theta^{A}} = \sum_{J_0, J_1} \alpha^{x_{\bar{I}_t},\Vec{\tau},J_0, J_1} \ket{\phi^{x_{\bar{I}_t},\Vec{\tau},J_0, J_1}}_{B_3} \ket{J_0, J_1}_{\textnormal{SEP}},
\end{equation}
where $\mathcal{H}_{B_2}\otimes \mathcal{H}_{\Theta^{A}}= \mathcal{H}_{B_3}\otimes \mathcal{H}_{\textnormal{SEP}}$  and the summation over $J_0,J_1$ goes over all possible values compatible with $I_t$. The state after Step (9) after Alice receives the SEP system and measures in the computational basis is then given by (tracing out SEP)
\begin{equation}
    \tilde{\rho}^{(4)} =  \sum_{\substack{\Vec{\tau} \\ {J_0,J_1}}} P_{\Vec{\tau},J_0,J_1} \ketbra{x^{A}_{I_t}, \theta^{A}_{I_t}}_{\Phi^{A}_{I_t}} \ketbra{\nu^{\Vec{\tau},J_0,J_1}}_{\Phi^{A}_{\bar{I}_t} B_3},
\end{equation}
where
\begin{align}
    P_{\Vec{\tau},J_0,J_1} &= \sum_{x_{\bar{I}_t \in G}} |\tilde{\gamma}^{x_{\bar{I}_t},r,\ccom,I_s,I_t} \alpha^{x_{\bar{I}_t},\Vec{\tau},J_0, J_1}|^2 \nonumber \\
    \ket{\nu^{\Vec{\tau},J_0,J_1}}_{\Phi^{A}_{\bar{I}_t} B_3} &= \sum_{x_{\bar{I}_t \in G}} \underbrace{(P_{\Vec{\tau},J_0,J_1})^{-\frac{1}{2}} \tilde{\gamma}^{x_{\bar{I}_t},r,\ccom,I_s,I_t} \alpha^{x_{\bar{I}_t},\Vec{\tau},J_0, J_1}}_{\beta^{x_{\bar{I}_t},\Vec{\tau},J_0, J_1}} \ket{x_{\bar{I}_t}, \theta^{B}_{\bar{I}_t}}_{\Phi^{A}_{\bar{I}_t}}\ket{\phi^{x_{\bar{I}_t},\Vec{\tau},J_0, J_1}}_{B_3}.
\end{align}
We can now consider Alice's measurement on the $\Phi^{A}_{\bar{I}_t}$ system. So far we have tracked the evolution of the joint state in order to describe the relationship between both parties' information. To finalize the proof we only need to keep track of the conditional min-entropy of Alice's outcomes given Bob's part of the joint system. Let $\rho_{X^{A}_{\Bar{I}_t},B_3} (\Vec{\tau}, J_0, J_1)$ be the resulting (conditioned) state after measuring the system $\Phi^{A}_{\bar{I}_t}$ in the $\theta^{A}_{\bar{I}_t}$ basis, recording the respective outcomes in the $X^{A}_{\Bar{I}_t}$ system, and tracing out the $\Phi^{A}_{\bar{I}_t}$ subsystem. We can write the state of the joint system after the measurement as 
\begin{align}
    \tilde{\rho}^{(5)} &= \sum_{\substack{\Vec{\tau} \\ {J_0,J_1}}} P_{\Vec{\tau},J_0,J_1}  \tilde{\rho}_{X^{A}_{\Bar{I}_t},B_3} (\Vec{\tau},J_0,J_1).
\end{align}
Additionally, for any given $J_0,J_1$, denote by $J_d$ the complement of $J_0 \cup J_1$ in $\bar{I}_{t}$. Following Lemma~\ref{lem:entropyproperties} (3) and (5) we know that for any $b \in \bs$
\begin{align}
    \label{eq:entropyboundchosenc}
    H_{\textnormal{min}}^{\varepsilon} (X^{A}_{J_{b}} | X^{A}_{J_{\bar{b}}} B_3)_{\tilde{\rho}^{(5)}} &\geq  H_{\textnormal{min}}^{\varepsilon} (X^{A}_{J_{b}} | X^{A}_{J_{\bar{b}}} X^{A}_{J_d} B_3)_{\tilde{\rho}^{(5)}} \nonumber \\ 
    &\geq \inf\limits_{ \Vec{\tau}, J_0,J_1 } \lbrace H_{\textnormal{min}}^{\varepsilon} (X^{A}_{J_{b}} | X^{A}_{J_{\bar{b}}} X^{A}_{J_d} B_3)_{\tilde{\rho}(\Vec{\tau}, J_0,J_1)} \rbrace \nonumber \\
    &\geq \inf\limits_{\Vec{\tau}, J_0,J_1}  \lbrace \inf \limits_{x_{{J_{\bar{b},d}}}} \lbrace H_{\textnormal{min}}^{\varepsilon} (X^{A}_{J_{b}} | B_3)_{\tilde{\rho}(x_{{J_{\bar{b},d}}}, \Vec{\tau}, J_0,J_1)} \rbrace \rbrace.
\end{align}

We can invoke Lemma~\ref{lem:minentropybound} to obtain an expression for the above quantity explicitly in terms of the protocol parameters $N_0, \alpha, \delta_1$, and $\delta_2$. For that, we must take a small detour to define the associated mixed states $\rho^{\textnormal{mix}}_{\Phi^{A}_{{J_b}} B_3}$ and compute their respective post-measurement entropy. First, for $b \in \bs$, we compute the reduced states
\begin{align}
    \rho_{\Phi^{A}_{J_b} B_3} (\Vec{\tau}, J_0,J_1) &= \Tr_{\Phi^{A}_{J_{\bar{b},d}}}  \bigg[ \ketbra{\nu^{\Vec{\tau},J_0,J_1}}_{\Phi^{A}_{\bar{I}_t} B_3} \bigg] \nonumber \\
    &= \sum_{x_{{J_{\bar{b},d}}}} P_{x_{J_{\bar{b},d}}}^{\Vec{\tau},J_0,J_1} \ketbra{\nu^{x_{J_{\bar{b},d}},\Vec{\tau},J_0,J_1}}_{\Phi^{A}_{J_b} B_3},
\end{align}
with
\begin{align*}
    P_{x_{{J_{\bar{b},d}}}}^{\Vec{\tau},J_0,J_1} = 
    \!\!\!
    \sum_{\substack{x_{{J_b}} \in \\ B_b(x_{{J_{\bar{b},d}}})}}
    \!\!\!
    \big| \beta^{x_{\bar{I}_t}, \Vec{\tau},J_0,J_1} \big|^2
\end{align*}
\begin{equation}
    \ket{\nu^{x_{J_{\bar{b},d}},\Vec{\tau},J_0,J_1}}_{\Phi^{A}_{{J_b}} B_3} = 
    \!\!\!
    \sum_{\substack{x_{{J_b}} \in \\ B_b(x_{{J_{\bar{b},d}}})}}
    \underbrace{(P_{x_{{J_{\bar{b},d}}}}^{\Vec{\tau},J_0,J_1})^{-\frac{1}{2}} \beta^{x_{\bar{I}_t}, \Vec{\tau},J_0,J_1}}_{\lambda^{x_{\bar{I}_t}, \Vec{\tau},J_0,J_1}}
    \ket{x_{{J_b}}, \theta^{B}_{{J_b}}}_{\Phi^{A}_{{J_b}}}
    \ket{\phi^{x_{\bar{I}_t},\Vec{\tau},J_0, J_1}}_{B_3},
\end{equation}
and
\begin{align}
    B_b(x_{{J_{\bar{b},d}}}) &= \lbrace x_{{J_b}} : x_{{J_{b,\bar{b},d}}} \in G(x^{A}_{I_s}, r, \ccom) \rbrace \nonumber \\
    &= \lbrace x_{{J_b}} : |(\frac{1}{2}-\delta_2)r_{\textnormal{H}}(x_{{J_b}}\oplus x^{B}_{{J_b}}) + (\frac{1}{2}+\delta_2)r_{\textnormal{H}}(x_{{J_{\bar{b},d}}}\oplus x^{B}_{{J_{\bar{b},d}}})-r_{\textnormal{H}}(x^{A}_{I_s}\oplus x^{B}_{I_s})| \leq  \delta_1 \rbrace,
\end{align}
where the explicit dependence of $B_b$ on $x^{A}_{I_s}, r, \ccom$ has been omitted for compactness. Note that since $r_{\textnormal{H}}(x^{A}_{I_s}\oplus x^{B}_{I_s}) \leq p_{\text{max}}$ the size of $B_b(x_{{J_{\bar{b},d}}})$ is upper bounded by
\begin{equation}
    \label{eq:errorbound}
    |B_b(x_{{J_{\bar{b},d}}})| \leq 2^{h\left( \frac{p_{\text{max}} + \delta_1}{\frac{1}{2} - \delta_2} \right)N_{\text{raw}}},
\end{equation}
where the $h$ stands for the binary entropy function. We can now define
\begin{align}
    \rho^{\textnormal{mix}}_{\Phi^{A}_{{J_b}} B_3} (x_{{J_{\bar{b},d}}}, \Vec{\tau}, J_0,J_1) =
    \!\!\!
    \sum_{\substack{x_{{J_b}} \in \\ B_b(x_{{J_{\bar{b},d}}})}}
    \!\!\!
    &\big|\lambda^{x_{\bar{I}_t}, \Vec{\tau},J_0,J_1} \big|^{2} 
    \ketbra{x_{{J_b}}, \theta^{B}_{{J_b}}}_{\Phi^{A}_{{J_b}}} \nonumber \\
    & \otimes \ketbra{\phi^{x_{\bar{I}_t},\Vec{\tau},J_0, J_1}}_{B_3}.
\end{align}
Measuring the above state in the $\theta^{A}_{{J_b}}$ basis, recording the results in $X_{J_{b}}$ and tracing out $\Phi^{A}_{{J_b}}$ leads to
\begin{align}
    \label{eq:rhomix}
    \rho^{\textnormal{mix}}_{X^{A}_{{J_b}} B_3}(x_{{J_{\bar{b},d}}}, \Vec{\tau}, J_0,J_1) = \sum_{x^{A}_{J_b}}
    \!\!\!
    \sum_{\substack{x_{{J_b}} \in \\ B_c(x_{{J_{\bar{b},d}}})}}
    \!\!\!
    &\big| \lambda^{x_{\bar{I}_t}, \Vec{\tau},J_0,J_1} \big|^{2} |\bracket{x^{A}_{J_b},\theta^{A}_{{J_b}}}{x_{{J_b}}, \theta^{B}_{{J_b}}}|^{2} \ketbra{x^{A}_{J_b}}_{X^{A}_{J_b}} \nonumber \\
    & \otimes \ketbra{\phi^{x_{\bar{I}_t},\Vec{\tau},J_0, J_1}}_{B_3},
\end{align}
defining $J_{b}^{0/1} = \lbrace i \in J_{b} : \theta^{\textnormal{ch}}_{i} = 0/1 \rbrace$ we can write the factors 
\begin{align}
    \big|\bracket{x^{A}_{J_{b}},\theta^{A}_{{J_b}}}{x_{{J_b}}, \theta^{B}_{{J_b}}}\big|^{2} = 
    \prod_{i \in J_{b}^{0}}  \underbrace{\big|\bracket{{x^{A}_i,\theta^{A}_i}}{x_{i}, \theta^{B}_{i}}\big|^{2}}_{\delta(x^{A}_i, x_{i})} 
    \prod_{i \in J_{b}^{1}}  \underbrace{\big|\bracket{{x^{A}_i,\theta^{A}_i}}{x_{i}, \theta^{B}_{i}}\big|^{2}}_{\big|\frac{1}{\sqrt{2}} \big|^{2}},
\end{align}
substituting in Eq.~(\ref{eq:rhomix}) we get
\begin{align}
    \rho^{\textnormal{mix}}_{X^{A}_{{J_b}} B_3}(x_{{J_{\bar{b},d}}}, \Vec{\tau}, J_0,J_1) &= \left(\frac{1}{2}\right)^{w_{\textnormal{H}}(\theta^{\textnormal{ch}}_{{J_{b}}})} \sum_{x^{A}_{{J_{b}^{1}}}} \ketbra{x^{A}_{{J_{b}^{1}}} }_{ X^{A}_{J_{b}^{1}} } \nonumber \\ 
    & \otimes 
    \!\!\!
    \sum_{\substack{x_{{J_b}} \in \\ B_b(x_{{J_{\bar{b},d}}})}}
    \!\!\!
    \big| \lambda^{x_{\bar{I}_t}, \Vec{\tau},J_0,J_1} \big|^{2}  \ketbra{x_{{I^{0}_b}}}_{ X^{A}_{{J_{b}^{0}}} }  \ketbra{\phi^{x_{\bar{I}_t},\Vec{\tau},J_0, J_1}}_{B_3},
\end{align}
which is a product state between the systems $X^{A}_{{J_{b}^{1}}}$ and $X^{A}_{{J_{b}^{0}}} B_3$. From Lemma~\ref{lem:entropyproperties} (1) and (2) we know that
\begin{align}
    \label{eq:mismatchbound}
    H_{\textnormal{min}}^{\varepsilon}(X^{A}_{J_b}|B_3)_{\rho^{\textnormal{mix}}(x_{{J_{\bar{b},d}}}, \Vec{\tau}, J_0, J_1)} &\geq H_{\textnormal{min}}^{0}(X^{A}_{J_b}|B_3)_{\rho^{\textnormal{mix}}(x_{{J_{\bar{b},d}}}, \Vec{\tau}, J_0, J_1)} \nonumber \\
    &\geq H_{\textnormal{min}}^{0}(X^{A}_{J^{0}_b}|B_3)_{\rho^{\textnormal{mix}}(x_{{J_{\bar{b},d}}}, \Vec{\tau}, J_0, J_1)} + H_{\textnormal{min}}^{0}(X^{A}_{J^1_b})_{\rho^{\textnormal{mix}}(x_{{J_{\bar{b},d}}}, \Vec{\tau}, J_0, J_1)} \nonumber \\
    &\geq H_{\textnormal{min}}^{0}(X^{A}_{J^1_b})_{\rho^{\textnormal{mix}}(x_{{J_{\bar{b},d}}}, \Vec{\tau}, J_0, J_1)} \nonumber \\
    &= -\log\left( \left(\frac{1}{2}\right)^{w_{\textnormal{H}}(\theta^{\textnormal{ch}}_{{J_{b}}})} \right) = w_{\textnormal{H}}(\theta^{\textnormal{ch}}_{{J_{b}}}).
\end{align}
Application of Lemma~\ref{lem:minentropybound} together with equations~(\ref{eq:mismatchbound}) and~(\ref{eq:errorbound}) leads to 
\begin{align}
    \label{eq:entropyboundanychoice}
    H_{\textnormal{min}}^{\varepsilon}(X^{A}_{J_b}|B_3)_{\tilde{\rho}(x_{{J_{\bar{b},d}}}, \Vec{\tau}, J_0, J_1)}  &\geq
    H_{\textnormal{min}}^{\varepsilon}(X^{A}_{J_b}|B_3)_{\rho^{\textnormal{mix}}(x_{{J_{\bar{b},d}}}, \Vec{\tau}, J_0, J_1)} - \log\left(|B_b(x_{{J_{\bar{b}, d}}}) |\right) \nonumber \\
    &\geq w_{\textnormal{H}}(\theta^{\textnormal{ch}}_{{J_{b}}}) - h\left( \frac{p_{\text{max}} + \delta_1}{\frac{1}{2} - \delta_2}\right)N_{\text{raw}}.
\end{align}
Note that the above expression depends only on the number of nonmatching bases $\theta^{\textnormal{ch}}$ associated to the indices in $J_b$ and the parameters of the protocol, which in turn makes the infimum in Eq.~(\ref{eq:entropyboundchosenc}) straightforward to compute. We can now add the respective conditional min-entropies for $X^{A}_{J_0}$ and $X^{A}_{J_1}$, which results in:
\begin{align}
    \label{eq:totalleak}
    H_{\textnormal{min}}^{\varepsilon} (X^{A}_{J_{0}} | X^{A}_{J_{1}} B_3)_{\tilde{\rho}(\Vec{\tau}, J_0, J_1)} + H_{\textnormal{min}}^{\varepsilon} (X^{A}_{J_{1}} | X^{A}_{J_{0}} B_3)_{\tilde{\rho}(\Vec{\tau}, J_0, J_1)} &\geq w_{\textnormal{H}}(\theta^{\textnormal{ch}}_{J_{0}}) + w_{\textnormal{H}}(\theta^{\textnormal{ch}}_{{J_{1}}}) - 2 h\left( \frac{p_{\text{max}} + \delta_1}{\frac{1}{2} - \delta_2}\right) N_{\text{raw}} \nonumber \\
    &\geq N_{\text{raw}} - 2\delta_2 (1-\alpha)N_0 - 2 h\left( \frac{p_{\text{max}} + \delta_1}{\frac{1}{2} - \delta_2}\right) N_{\text{raw}} \nonumber \\
    &\geq 2 N_{\text{raw}} \left( \frac{1}{2} - \frac{2\delta_2}{1-2\delta_2} - h\left( \frac{p_{\text{max}} + \delta_1}{\frac{1}{2} - \delta_2}\right) \right).
\end{align}
The result follows by recalling, from Eqs.~\eqref{eq:postmeasurementdistanceB} and~\eqref{eq:chernoffdistance}, that the real state at this point in the protocol has distance from $\tilde{\rho}^{(5)}$ bounded by
\begin{equation}
    \varepsilon = \sqrt{2} \left( e^{-\frac{1}{2} \alpha (1-\alpha)^{2} N_0 \delta^{2}_1} +  e^{-\frac{1}{2} (\frac{1}{2} -\delta_2)\alpha N_0 \delta^{2}_1} \right)^{\frac{1}{2}} + e^{-D_{KL}(\frac{1}{2} - \delta_2 | \frac{1}{2})(1-\alpha)N_{0}} + \varepsilon_{\textnormal{bind}}(k).
\end{equation}
\end{proof}

\section{UC security in the Random Oracle Model}
\label{sec:ucsecurity}
Following the discussion made in Section~\ref{sec:composabiltiy}, we prove the composability of a specific family of weakly-interactive commitment schemes in the \textit{classical access random oracle model}, which we will refer as ROM from here onwards. These commitments, originally proposed by Lorünser, Ramache, and Valbusa~\cite{lorunser25}, build upon the original Naor bit commitment~\cite{naor91} and efficiently generalize it for arbitrary $k$-bit string commitments without the need of error correcting codes. A description of the commitment protocol is shown in Fig.~\ref{fig:comprotocol}, whose correctness, binding, and hiding properties, have been proven in~\cite{lorunser25}. Instead, we will thus limit ourselves to prove that the LRV commitment protocol UC-emulates the commitment functionality $\mathcal{F}_{\text{COM}}$ when the hash function is modeled as an oracle $\mathcal{F}_{\text{RO}}$ which computes a random function. 

\begin{figure}[H]
\framebox{\parbox{\dimexpr\linewidth-2\fboxsep-2\fboxrule}{
  {\centering \textbf{LRV25 String commitment protocol} \par}
  \textbf{Parameters:} \begin{itemize}
    \item Parties Alice (Verifier) and Bob (Prover)
    \item Security parameter $k$ and message length $n$
    \item Collision-resistant hash function $H: \bs^{k} \rightarrow \bs^{3k + n}$
    \item A subroutine $O$, which on input a vector $\mathbf{r}_1 \in \bs^{3k + n}$, outputs a tuple of $n$ linearly independent vectors $(\mathbf{r}_1, \ldots, \mathbf{r}_n)$ in $\bs^{3k + n}$
  \end{itemize}
  \textbf{Inputs:} \begin{itemize}
    \item Bob receives the $n$-bit string $\mathbf{b}= (b_1, \ldots, b_n)$
  \end{itemize}
  \textit{(Commit phase)}
  \begin{enumerate}
      \item Alice uniformly samples a $(3k + n)$-bit string $\mathbf{r}_1$ and sends it to Bob
      \item Bob uniformly samples an $n$-bit string $\mathbf{x}$ and computes $(\mathbf{r}_1, \ldots, \mathbf{r}_n) = O(\mathbf{r}_1)$. Then, he computes $\mathbf{c} = H(\mathbf{x}) \oplus \sum_{i=1}^{n} b_i\cdot \mathbf{r}_i$ and sends $\mathbf{c}$ to Alice
  \end{enumerate}
  \textit{(Open phase)}
  \begin{enumerate}
      \item Bob sends $(\mathbf{b}, \mathbf{x})$ to Alice
      \item Alice computes $(\mathbf{r}_1, \ldots, \mathbf{r}_n) = O(\mathbf{r}_1)$, then, if $\mathbf{c} = H(\mathbf{x}) \oplus \sum_{i=1}^{n} b_i\cdot \mathbf{r}_i$ outputs $\mathbf{b}$.
  \end{enumerate}
  }}
  \caption{Weakly-interactive string commitment scheme based on hash functions}
  \label{fig:comprotocol}
\end{figure}

\begin{figure}[H]
\framebox{\parbox{\dimexpr\linewidth-2\fboxsep-2\fboxrule}{
  {\centering \textbf{Functionality} $\mathcal{F}_{\text{COM}}$ \par}
  \textbf{Parameters:} \begin{itemize}
      \item Parties Alice (Verifier) and Bob (Prover)
      \item Message length $n$
  \end{itemize}
  \begin{enumerate}
      \item Upon receiving an input $(\texttt{commit}, sid, \mathbf{b})$ from Bob, if no value has previously been committed, output the message $(\texttt{committed}, sid)$ to Alice
      \item Upon receiving the input $(\texttt{open}, sid)$ from Bob, if a value $\mathbf{b}$ has previously been committed, output the message $(\texttt{open}, sid, \mathbf{b})$ to Alice
  \end{enumerate}
  }}
  \caption{Commitment ideal functionality}
  \label{fig:comfunc}
\end{figure}

Let $\Pi_A$ and $\Pi_B$ represent the programs for the Verifier and Prover, respectively, as shown in Fig.~\ref{fig:comprotocol}. Note that, for simplicity, the external inputs that trigger the start and end of the Commit and Reveal phases have been omitted from Fig.~\ref{fig:comprotocol}; without loss of generality, we can consider them to take the form of the respective inputs and outputs as shown in the $\mathcal{F}_{\text{COM}}$ functionality. More specifically, the Commit phase starts when $\Pi_B$ receives the input $(\texttt{commit}, sid, \mathbf{b})$ and ends when $\Pi_A$ outputs $(\texttt{committed}, sid)$, etc.
\begin{figure}[b!]
\centering
\includegraphics[width=0.85\linewidth]{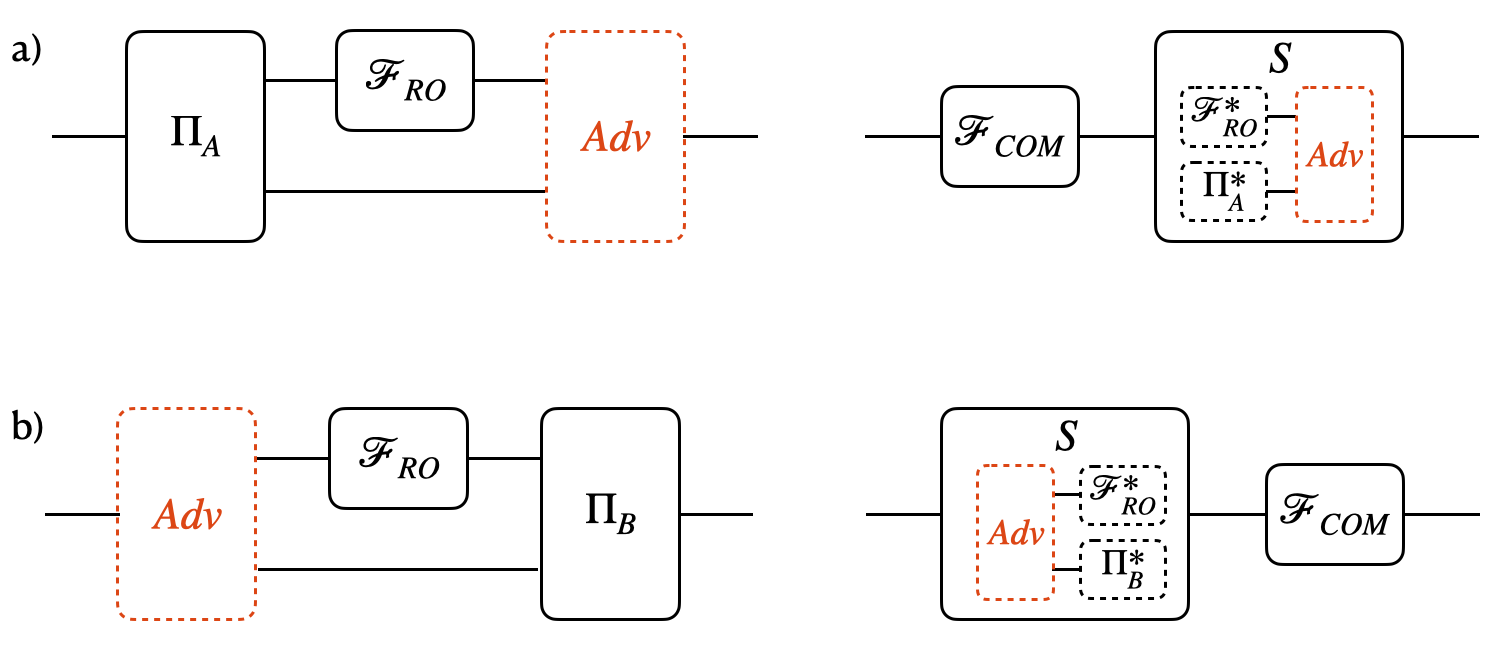}
\caption{Box diagrams for the execution of the protocol for a) dishonest Bob and b) dishonest Alice. The left sides represent the \textit{real world} protocol interacting with an adversary $Adv$ while the right sides represent the \textit{ideal world} functionality interacting with the respective simulator $\mathcal{S}$}
\label{fig:COMSecurityROM}
\end{figure} 
We proceed now to separate the security in two cases, in which the adversary controls Alice or Bob, respectively, as shown in Fig.~\ref{fig:COMSecurityROM}. In order to prove security we must show that for any efficient (i.e., polynomial-time) adversary $Adv$ with classical access to the oracle there exists a respective simulator $\mathcal{S}$ such that for any environment, which is able to send and receive inputs/outputs through the loose wires in the right and left of the diagrams, the \textit{real world} and \textit{ideal world} scenarios are indistinguishable. Denote by $H$ the function that the random oracle computes.

\subsubsection*{Dishonest Bob:}
We construct the simulator in terms of the following subprograms:
\begin{itemize}
    \item $\mathcal{F}_{\text{RO}}^*$: The same as $\mathcal{F}_{\text{RO}}$, except that it saves a list $L$ of all the queries that have been made to the internal memory of $\mathcal{S}$.
    \item $\Pi_A^*$: The same as $\Pi_A$, except that after receiving $c$ from $Adv$ it runs through the current list $L$ of queries. When it finds an $\mathbf{x}'\in L$ and $\mathbf{b} \in \bs^n$ such that
    \begin{equation}
        \label{eq:extractcondition}
        \mathbf{c} =G(\mathbf{x}') \oplus \sum_{i=1}^{n} b_i\cdot \mathbf{r}_i,
    \end{equation}
    it sends $(\texttt{commit}, sid, \mathbf{b})$ to $\mathcal{F}_{\text{COM}}$. If no pair $(\mathbf{b},\mathbf{x}')$ is found, it samples uniformly a value $\mathbf{b}$ and sends $(\texttt{commit}, sid, \mathbf{b})$ to $\mathcal{F}_{\text{COM}}$. In the reveal phase, if the check is passed, it sends $(\texttt{open}, sid)$ to $\mathcal{F}_{\text{COM}}$.
\end{itemize} 
Because of the binding property of the commitment protocol, the simulator may find at most one pair $(\mathbf{b},\mathbf{x}')$ satisfying Eq.~\eqref{eq:extractcondition} when looking through the list, except with negligible probability (this is because the probability of there existing more than one valid openings for a given value of $\mathbf{c}$ is negligible). This allows $\mathcal{S}$ to correctly extract the committed value from $\mathbf{c}$ and commit it to $\mathcal{F}_{\text{COM}}$. Note that in the case no valid opening is found from $L$, the simulator commits a random value to $\mathcal{F}_{\text{COM}}$. If the adversary is able to provide a valid opening pair $(\mathbf{b},\mathbf{x})$ in the Reveal phase the two scenarios could be distinguished. However, from the preimage resistance of random oracles, an efficient adversary cannot find a valid opening from a value of $\mathbf{c}$ without having obtained it by querying the oracle, meaning that regardless of $\mathcal{S}$ committing a random value to $\mathcal{F}_{\text{COM}}$, the probability of it being opened is negligible. 

\subsubsection*{Dishonest Alice:}
Similarly, we construct the simulator in terms of the following subprograms:
\begin{itemize}
    \item $\mathcal{F}_{\text{RO}}^*$: The same as for the dishonest Bob case, except it may be reprogrammed on individual query-output pairs.
    \item $\Pi_B^*$: The same as $\Pi_B$, except upon receiving an input of the form $(\texttt{committed}, sid)$ from $\mathcal{F}_{\text{COM}}$, it samples uniformly the value $\mathbf{c}'$ and sends it to $Adv$. In the Reveal phase, upon receiving $(\texttt{open}, sid, \mathbf{b})$ from $\mathcal{F}_{\text{COM}}$, samples a random $\mathbf{x}'$ not in $L$, sets $\mathcal{F}_{\text{RO}}^*$ so that
    \begin{equation}
    \label{eq:equivocatecondition}
        G(\mathbf{x}')= \mathbf{c}'\oplus\sum_{i=1}^{n} b_i\cdot \mathbf{r}_i,
    \end{equation}
    and sends $(\mathbf{b},\mathbf{x}')$ to $Adv$.
\end{itemize} 
From the hiding property of the commitment protocol, the value $\mathbf{c}$ received by Alice during the Commit phase does not give a significant advantage to an efficient adversary in finding the committed value $\mathbf{b}$ as compared to a random string. Because of this, an efficient adversary cannot distinguish if the randomly sampled $\mathbf{c}'$ corresponds to any possible committed value, except with negligible probability. During the reveal phase, the reprogramming of the oracle according to Eq.~\eqref{eq:equivocatecondition} guarantees that have $\mathbf{c}'$ will be consistent with the committed values from $\mathcal{F}_{\text{COM}}$. The only difference between the real and ideal scenarios is the change in the behavior of the oracle. Because the value $\mathbf{c}'$ was sampled uniformly, the associated outcome $G(\mathbf{x}')$ as defined by Eq.~\eqref{eq:equivocatecondition} is also uniformly distributed and independent on the rest of the values $G(\mathbf{x} \neq \mathbf{x}')$, resulting in both scenarios being consistent with the oracle computing a random function, and therefore indistinguishable from each other.

\end{appendices}

\bibliographystyle{unsrt}
\bibliography{bibfile}

\end{document}